\documentclass[11pt]{article}
\usepackage[letterpaper,portrait,margin=1in]{geometry}
\usepackage[T1]{fontenc}
\usepackage{lmodern}
\usepackage{amsmath,amssymb,amsthm}
\usepackage{array,longtable}
\usepackage{needspace}
\usepackage{microtype}
\usepackage{xcolor}
\usepackage{tikz}
\usetikzlibrary{arrows.meta,calc,fit,positioning}
\usepackage{aliascnt}
\usepackage[colorlinks=true,linkcolor=teal,citecolor=blue,urlcolor=teal]{hyperref}
\usepackage[nameinlink,capitalize]{cleveref}

\allowdisplaybreaks

\newtheorem{theorem}{Theorem}[section]
\newaliascnt{lemma}{theorem}
\newtheorem{lemma}[lemma]{Lemma}
\aliascntresetthe{lemma}
\newaliascnt{proposition}{theorem}
\newtheorem{proposition}[proposition]{Proposition}
\aliascntresetthe{proposition}
\newaliascnt{corollary}{theorem}
\newtheorem{corollary}[corollary]{Corollary}
\aliascntresetthe{corollary}
\newaliascnt{remark}{theorem}
\theoremstyle{remark}
\newtheorem{remark}[remark]{Remark}
\aliascntresetthe{remark}
\theoremstyle{plain}

\crefname{theorem}{Theorem}{Theorems}
\Crefname{theorem}{Theorem}{Theorems}
\crefname{lemma}{Lemma}{Lemmas}
\Crefname{lemma}{Lemma}{Lemmas}
\crefname{proposition}{Proposition}{Propositions}
\Crefname{proposition}{Proposition}{Propositions}
\crefname{corollary}{Corollary}{Corollaries}
\Crefname{corollary}{Corollary}{Corollaries}
\crefname{remark}{Remark}{Remarks}
\Crefname{remark}{Remark}{Remarks}
\crefname{equation}{equation}{equations}
\Crefname{equation}{Equation}{Equations}

\newcommand{\C}{\mathbb C}
\newcommand{\R}{\mathbb R}
\newcommand{\cA}{\mathcal A}
\newcommand{\cC}{\mathcal C}
\newcommand{\cE}{\mathcal E}
\newcommand{\cF}{\mathcal F}
\newcommand{\cH}{\mathcal H}
\newcommand{\cM}{\mathcal M}
\newcommand{\cS}{\mathcal S}
\newcommand{\cT}{\mathcal T}
\newcommand{\U}{\mathrm U}
\newcommand{\dd}{\mathrm d}
\newcommand{\diag}{\operatorname{diag}}
\newcommand{\Gr}{\operatorname{Gr}}
\newcommand{\Tr}{\operatorname{Tr}}
\newcommand{\ket}[1]{|#1\rangle}
\newcommand{\proj}[1]{|#1\rangle\!\langle #1|}
\newcommand{\ot}{\otimes}
\newcommand{\Forb}{F_{\mathrm{orb}}}
\newcommand{\Fcl}{F_{\mathrm{cl}}}
\newcommand{\Funiv}{F_{\mathrm{univ}}}
\newcommand{\doi}[2]{\href{https://doi.org/#1}{#2}}
\newcommand{\arxiv}[1]{\href{https://arxiv.org/abs/#1}{\texttt{arXiv:#1}}}

\title{Asymptotically Optimal Mixed-State Cloning}
\author{Jiani Fei\(^{*}\)\qquad
  Jinzhao Wang\(^{*,\dagger}\)}
\date{}
\begin{document}
\maketitle
\begingroup
\renewcommand{\thefootnote}{\fnsymbol{footnote}}
\footnotetext[1]{\textit{Leinweber Institute for Theoretical Physics,
  Stanford University, Stanford, CA 94305, USA}}
\footnotetext[2]{\textit{Zyphra, San Francisco, CA 94105, USA.}}
\endgroup

\begin{abstract}
We determine asymptotic global fidelities of cloning generic
mixed states with simple (non-degenerate) full-rank spectra. In fixed dimension $d$, let $n$ inputs and
$m_n$ outputs satisfy $m_n/n\to\gamma>1$.  For a known simple
full-rank spectrum $p_1>\cdots>p_d>0$, we prove that the optimal root
fidelity between the entire output and the product target state is
$$
  \Forb(\gamma,p)
  =\prod_{i<j}
  \frac{\sqrt{\gamma}+\sqrt{q_{ij}(\gamma-1+q_{ij})}}
       {\gamma+q_{ij}},
  \qquad q_{ij}=\frac{p_j}{p_i}.
$$
When the spectrum is unknown, one spectrum-independent cloner attains
$$
  \Funiv(\gamma,p)
  =\left(\frac{2\sqrt\gamma}{1+\gamma}\right)^{(d-1)/2}
   \Forb(\gamma,p),
$$
uniformly on compact subsets of the simple full-rank chamber and is minimax optimal.  The two factors quantify
spectral and eigenbasis uncertainty.  Our channels combine classical
processing of Schur--Weyl labels with a covariant Cartan-sector channel.
Local asymptotic normality and a Gaussian fidelity bound against arbitrary
channels yield the converses.  For flat projector states $\rho=P/r$,
with $P$ a rank-$r$ orthogonal projector, we separately prove the
optimum $\gamma^{-r(d-r)/2}$.  Furthermore, we estimate the exact asymptotic fidelity
of purify--clone--trace (PCT), and show that our cloners strictly
outperform PCT at every fixed gain for simple full-rank spectra and
rank-$r$ projector states with $1<r<d$.
\end{abstract}

\vspace{3em}
\begin{samepage}
\paragraph{Generative AI usage.}
The idea for the cloner construction is due to the authors. We
initially conjectured its asymptotic optimality but did not have a complete
proof.  We used OpenAI's ChatGPT (Sol~5.6 and 6) and Anthropic's Claude
(Fable~5 and 5.1) to assist with the analysis: the evaluation of the
cloning fidelity and the converse argument benefited from their
suggestions.  The same tools assisted with language editing, manuscript
organization, and figure preparation.  All AI-assisted material was
reviewed and revised by the authors, who take full responsibility for the
accuracy and content of the manuscript.
\end{samepage}

\clearpage
\tableofcontents

\section{Introduction}\label{sec:introduction-results}

The no-cloning theorem rules out perfect copying of an arbitrary unknown
quantum state \cite{WoottersZurek1982,Dieks1982}.  Approximate cloning asks
how well a channel can turn \(n\) identical inputs into \(m\) copies, at
\emph{gain} \(m/n\).  We compare the entire output with the product
target \(\rho^{\ot m}\), rather than testing its marginals as in broadcasting
\cite{BarnumCavesFuchsJozsaSchumacher1996}.  Optimal pure-state cloners
were established in the 1990s \cite{Werner1998,KeylWerner1999}, but
mixed-state cloning has remained open
\cite{ScaraniIblisdirGisinAcin2005}: existing results treat restricted
families, different loss criteria, or high-fidelity sample complexity,
without determining the optimal global fidelity in our setting
\cite{BuscemiDArianoMacchiavelloPerinotti2005Optimal,BuscemiDArianoMacchiavelloPerinotti2006,DangFan2007,FanLiuShi2007,Kazakov2007,LiTheilHarrowChuang2026,FanizzaGrinkoScharnhorstSpilecki2026,JeonSohnOh2026}.
We resolve this fixed-gain global-fidelity problem in the generic
mixed-state regime of simple full-rank spectra: the eigenvalues are
strictly positive and distinct.  Such states form an open dense
subset of the state space.\footnote{They also have full measure under
Hilbert--Schmidt volume: the excluded states are zeros of the determinant
or the characteristic-polynomial discriminant.}  We determine the
known-spectrum orbit optimum and, when the spectrum is unknown, the
compact-set minimax optimum for spectral sets with dense full-dimensional
interior.

\subsection{Results}\label{sec:main-results}

Fix an integer \(d\ge2\) and positive integers \(m_n\) such that
\begin{equation}\label{eq:fixed-gain-regime}
  \gamma_n:=\frac{m_n}{n}\longrightarrow\gamma>1.
\end{equation}
In particular, \(m_n>n\) for all sufficiently large \(n\).\footnote{At the
finitely many exceptional indices with \(m_n\le n\), our cloners are defined
as the exact partial trace channel (and the identity when \(m_n=n\)).  This convention
does not affect any limit, and all construction formulas below concern
\(m_n>n\).}
Let \(\mathsf A_d:=\{x\in\R^d:\sum_i x_i=1\}\) and define
\begin{equation}\label{eq:simple-spectrum-family}
  \Delta_d^\circ
  :=\left\{p\in\mathsf A_d:p_1>\cdots>p_d>0\right\}.
\end{equation}
For every ordered positive spectrum \(p\in\mathsf A_d\) with
\(p_1\ge\cdots\ge p_d>0\), including repeated eigenvalues, set
\(\rho_p=\diag(p_1,\ldots,p_d)\) and, for \(g\in\U(d)\),
\(\rho_{p,g}=g\rho_p g^*\).
For positive trace-class operators \(A,C\), not necessarily normalized, we
write
\begin{equation}\label{eq:root-fidelity-definition}
  F(A,C)=\|\sqrt A\sqrt C\|_1.
\end{equation}
For states, this is the root (unsquared) Uhlmann fidelity.  For nonnegative
functions or measures, the same notation denotes their Hellinger affinity. In
particular, \(F(P,Q)=\sum_x\sqrt{P(x)Q(x)}\) on a countable space.  The target
in every statement below is the entire product state
\(\rho_{p,g}^{\ot m_n}\).  For
\(p\in\Delta_d^\circ\), define
\begin{equation}\label{eq:pairwise-fidelity-factor}
  q_{ij}=\frac{p_j}{p_i},
  \qquad
  f_\gamma(q)
  =\frac{\sqrt\gamma+\sqrt{q(\gamma-1+q)}}{\gamma+q},
\end{equation}
and
\begin{equation}\label{eq:main-values}
  \Forb(\gamma,p)=\prod_{i<j}f_\gamma(q_{ij}),
  \qquad
  \Fcl(\gamma,d)
  =\left(\frac{2\sqrt\gamma}{1+\gamma}\right)^{(d-1)/2},
  \qquad
  \Funiv(\gamma,p)=\Fcl(\gamma,d)\Forb(\gamma,p).
\end{equation}
The factor \(\Forb\) measures the cost of amplifying the unknown eigenbasis.
The \emph{classical fidelity factor} \(\Fcl\) accounts for the additional
cost of reproducing the unknown spectral fluctuations.

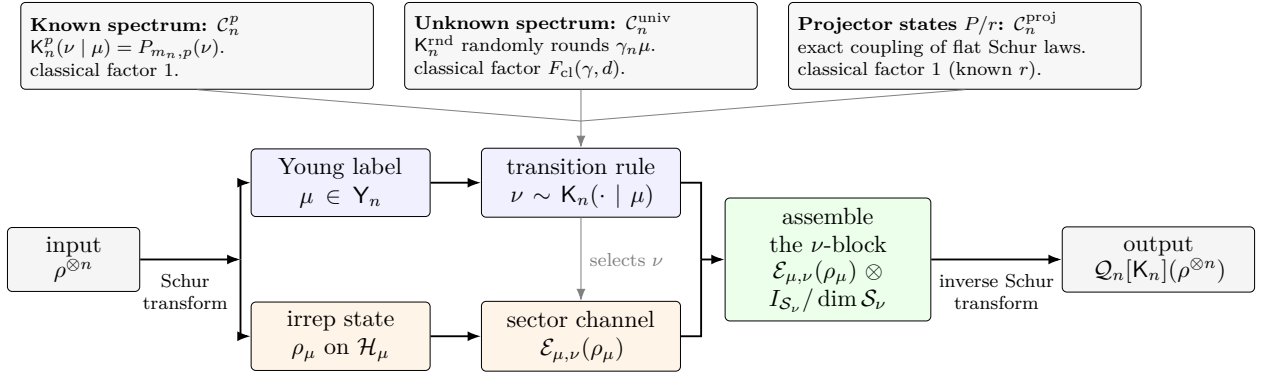
\begin{figure}[!tbp]
\centering
\resizebox{\linewidth}{!}{%
\begin{tikzpicture}[
  flow/.style={-{Latex[length=2mm]},thick},
  dependency/.style={-{Latex[length=1.8mm]},thin,gray},
  guide/.style={thin,gray},
  box/.style={draw,rounded corners=2pt,align=center,inner sep=4pt,
    minimum height=9mm,font=\small},
  labelbox/.style={box,fill=blue!6},
  sectorbox/.style={box,fill=orange!9},
  note/.style={draw,rounded corners=2pt,align=left,inner sep=4pt,
    fill=gray!6,font=\scriptsize}
]
  \node[box,fill=gray!8,text width=17mm] (input) at (-0.5,0)
    {input\\[-1pt]\(\rho^{\ot n}\)};

  \node[labelbox,text width=24mm] (labelin) at (3.5,1.15)
    {Young label\\\(\mu\in\mathsf Y_n\)};
  \node[sectorbox,text width=24mm] (sectorin) at (3.5,-1.15)
    {irrep state\\\(\rho_\mu\) on \(\cH_\mu\)};

  \node[labelbox,text width=27mm] (labelout) at (7.1,1.15)
    {transition rule\\\(\nu\sim \mathsf K_n(\cdot\mid\mu)\)};
  \node[sectorbox,text width=27mm] (sectorout) at (7.1,-1.15)
    {sector channel\\\(\cE_{\mu,\nu}(\rho_\mu)\)};

  \node[box,fill=green!8,text width=28mm] (blockout) at (10.8,0)
    {assemble the \(\nu\)-block\\[-1pt]
     \(\cE_{\mu,\nu}(\rho_\mu)
       \ot I_{\cS_\nu}/\dim\cS_\nu\)};
  \node[box,fill=gray!8,text width=26mm]
    (output) at (15.75,0)
    {output\\[-1pt]
     \(\mathcal Q_n[\mathsf K_n](\rho^{\ot n})\)};

  \coordinate (schursplit) at (2.0,0);
  \draw[flow] (schursplit) |- (labelin.west);
  \draw[flow] (schursplit) |- (sectorin.west);
  \draw[flow] (input.east) --
    node[below=2pt,pos=.45,font=\scriptsize,align=center]
      {Schur\\transform} (schursplit);
  \draw[flow] (labelin.east) -- (labelout.west);
  \draw[flow] (sectorin.east) -- (sectorout.west);
  \draw[dependency] (labelout.south) --
    node[right,font=\scriptsize,text=gray] {selects \(\nu\)} (sectorout.north);
  \coordinate (assemblemerge) at (8.9,0);
  \draw[thick] (labelout.east) -| (assemblemerge);
  \draw[thick] (sectorout.east) -| (assemblemerge);
  \draw[flow] (assemblemerge) -- (blockout.west);
  \draw[flow] (blockout.east) --
    node[below=2pt,font=\scriptsize,align=center]
      {inverse Schur\\transform} (output.west);

  \node[note,text width=50mm,minimum height=13mm] (known) at (1.35,3.2)
    {\textbf{Known spectrum:} \(\cC_n^p\)\\
     \(\mathsf K_n^p(\nu\mid\mu)=P_{m_n,p}(\nu)\).\\
     classical factor \(1\).};
  \node[note,text width=50mm,minimum height=13mm] (unknown) at (7.1,3.2)
    {\textbf{Unknown spectrum:} \(\cC_n^{\mathrm{univ}}\)\\
     \(\mathsf K_n^{\mathrm{rnd}}\) randomly rounds \(\gamma_n\mu\).\\
     classical factor \(\Fcl(\gamma,d)\).};
  \node[note,text width=50mm,minimum height=13mm] (projector) at (12.85,3.2)
    {\textbf{Projector states \(P/r\):} \(\cC_n^{\mathrm{proj}}\)\\
     exact coupling of flat Schur laws.\\
     classical factor \(1\) (known \(r\)).};
  \coordinate (kernelchoice) at ([yshift=4mm]labelout.north);
  \draw[guide] (known.south) -- (kernelchoice);
  \draw[guide] (unknown.south) -- (kernelchoice);
  \draw[guide] (projector.south) -- (kernelchoice);
  \draw[dependency] (kernelchoice) -- (labelout.north);

\end{tikzpicture}%
}
\caption{\textbf{Schur--Cartan cloner construction}.  A Young-label transition rule
  selects \(\nu\), while the covariant channel \(\cE_{\mu,\nu}\) processes
  the conditional irrep state.  Known-spectrum, unknown-spectrum, and
  projector cloners differ only in their label rules, shown above. Their
  asymptotic fidelities are given in
  \cref{thm:known-optimum,thm:unknown-optimum,thm:grassmann}.}
\label{fig:schur-cartan-cloner}
\end{figure}

\paragraph{Known spectrum.}
For fixed \(p\), the target Young-label distribution is known and can be
sampled exactly at size \(m_n\). Only the conditional irrep
sectors encoding the eigenbasis need to be cloned.

\begin{theorem}[Known-spectrum optimum]\label{thm:known-optimum}
For every \(p\in\Delta_d^\circ\),
\begin{equation}\label{eq:known-optimum}
  \lim_{n\to\infty}
  \sup_{\cM_n}\inf_{g\in\U(d)}
  F\!\left(\cM_n(\rho_{p,g}^{\ot n}),\rho_{p,g}^{\ot m_n}\right)
  =\Forb(\gamma,p),
\end{equation}
where the supremum is over all quantum channels from \(n\) to \(m_n\)
qudits.  The explicit sequence of channels \(\cC_n^p\) defined by
\cref{eq:known-cloner-definition} asymptotically attains the displayed value.
Its fidelity converges to \(\Forb(\gamma,p)\) uniformly for \(p\) in
every compact \(K\Subset\Delta_d^\circ\) and \(g\in\U(d)\).
\end{theorem}

\paragraph{Unknown spectrum.}
Without \(p\), the target Young label law must instead be generated from
the measured input Young label.  Its \(n^{-1/2}\)-scale spectral
fluctuation must be dilated to the output scale, producing the factor
\(\Fcl(\gamma,d)\).

\begin{theorem}[Unknown-spectrum cloner fidelity and compact-set optimum]
\label{thm:unknown-optimum}
There exists a sequence of spectrum-independent, \(\U(d)\)-covariant channels
\(\cC_n^{\mathrm{univ}}\), defined by \cref{eq:universal-cloner-definition},
with the following properties.
\begin{enumerate}
\item[(a)] (Achievability) For every compact \(K\Subset\Delta_d^\circ\),
\begin{equation}\label{eq:unknown-uniform-achievability}
  \lim_{n\to\infty}
  \sup_{\substack{p\in K\\g\in\U(d)}}
  \left|
  F\!\left(
    \cC_n^{\mathrm{univ}}(\rho_{p,g}^{\ot n}),
    \rho_{p,g}^{\ot m_n}
  \right)-\Funiv(\gamma,p)
  \right|=0.
\end{equation}

\item[(b)](Converse) Let \(K\Subset\Delta_d^\circ\) be nonempty and
  satisfy
  \(K=\overline{\operatorname{int}_{\mathsf A_d}K}^{\,\mathsf A_d}\).
Then
\begin{equation}\label{eq:compact-global-value}
  \lim_{n\to\infty}
  \sup_{\cM_n}
  \inf_{\substack{p\in K\\g\in\U(d)}}
  F\!\left(\cM_n(\rho_{p,g}^{\ot n}),\rho_{p,g}^{\ot m_n}\right)
  =\inf_{p\in K}\Funiv(\gamma,p),
\end{equation}
where the supremum is over all quantum channels from \(n\) to \(m_n\)
qudits.
\end{enumerate}
\end{theorem}

We give a few remarks on our claims. Part~\textup{(a)} evaluates one prescribed cloner, even on singleton or
lower-dimensional sets.  In part~\textup{(b)}, the competing channel may
depend on \(K\), but not on the actual \(p\) or \(g\).  The interior is
taken in the full \((d-1)\)-dimensional hyperplane \(\mathsf A_d\): its
spectral directions give the converse factor \(\Fcl(\gamma,d)\), and
density extends the bound to the infimum over \(K\). 

The converse in part~\textup{(b)} is minimax (cf.
\cref{sec:minimax-completion}.): no competing channel sequence
can asymptotically exceed \(\inf_{p\in K}\Funiv(\gamma,p)\) in worst-case
fidelity over \(K\times\U(d)\), even though it may outperform
\(\cC_n^{\mathrm{univ}}\) at particular \((p,g)\).
For the singleton \(K=\{p\}\), which fails the closure condition in
part~\textup{(b)}, the known-spectrum channel \(\cC_n^p\) is admissible
and attains \(\Forb(\gamma,p)>\Funiv(\gamma,p)\). 

A natural example bounds the smallest eigenvalue and every adjacent gap
away from zero.  For \(0<\kappa<2/[d(d+1)]\), take
\[
  \Delta_{d,\kappa}
  =\left\{p\in\mathsf A_d:p_d\ge\kappa,\
    p_i-p_{i+1}\ge\kappa\ (i<d)\right\}.
\]
This is a nonempty compact subset of \(\Delta_d^\circ\) and satisfies
\(\Delta_{d,\kappa}
=\overline{\operatorname{int}_{\mathsf A_d}\Delta_{d,\kappa}}
^{\,\mathsf A_d}\),
so \eqref{eq:compact-global-value} applies.\footnote{Being convex with
nonempty relative interior, \(\Delta_{d,\kappa}\) is the closure of that
interior.  For a relative interior point,
choose \(\kappa'\in(\kappa,2/[d(d+1)])\) and set
\(p_i=1/d+((d+1)/2-i)\kappa'\).  At the endpoint
\(\kappa=2/[d(d+1)]\), the constraints force a singleton, so the
compact-set theorem no longer applies.}

\Cref{fig:schur-cartan-cloner} summarizes our cloner construction.
Schur--Weyl duality separates a Young label, carrying spectral information,
from a conditional irreducible state, carrying eigenbasis information
\cite{KeylWerner2001Spectrum}.  Each cloner processes the label classically
and applies the same multiplicity-one Cartan-sector channel to the
irreducible state. Only the label transition rule changes
\cite{ParthasarathyRangaRaoVaradarajan1967,FultonHarris1991}.
We develop the rules in \cref{sec:achievability}.  For the converses,
local asymptotic normality (LAN) gives a classical Gaussian shift and
displaced thermal modes.  A flat-prior Gaussian fidelity bound applies to
arbitrary channels, while least-noise amplifiers attain the quantum factor.
See \cref{sec:gaussian-LAN}.

Although the preceding results do not cover rank-deficient or spectrally degenerate states, we also study flat projector states \(P/r\), with
\(P\) an unknown rank-\(r\) projector, and determine their optimal cloning fidelity.

\begin{theorem}[Cloning projector states]
\label{thm:grassmann}
For fixed \(1\le r<d\), let \(\Gr(r,d)\) denote the set of rank-\(r\)
orthogonal projectors on \(\C^d\).  Then
\begin{equation}\label{eq:grassmann-value}
  \lim_{n\to\infty}\sup_{\cM_n}\inf_{P\in\Gr(r,d)}
  F\!\left(\cM_n((P/r)^{\ot n}),(P/r)^{\ot m_n}\right)
  =\gamma^{-r(d-r)/2}.
\end{equation}
The limit is attained by our cloner \(\cC_n^{\mathrm{proj}}\) constructed in \cref{sec:grassmann-achievability}.
\end{theorem}

For \(r=1\), the value \(\gamma^{-(d-1)/2}\) recovers the optimal
asymptotic global fidelity for pure-state cloning \cite{Werner1998}. For
\(r=d\), the family contains only \(I_d/d\), and the value is one.
The projector cloner uses the same Schur--Cartan architecture, with a
different Young-label transition rule (\cref{fig:schur-cartan-cloner}).
Its classical factor is \(1\) and its asymptotic sector factor is
\(\gamma^{-r(d-r)/2}\).  Exact Schur-sector compression gives its fidelity,
and a direct finite-dimensional converse in
\cref{prop:grassmann-finite-converse} holds at every \(n\). When
\(m_n=\gamma n+O(1)\), its upper bound approaches this limit at rate
\(O(n^{-1})\).

A natural idea for mixed-state cloning is to first purify the mixed states and then apply Werner’s optimal pure state cloner \cite{Werner1998}. This approach was recently proposed and studied in~\cite{LiTheilHarrowChuang2026,FanizzaGrinkoScharnhorstSpilecki2026,JeonSohnOh2026}. Here we show that our optimal cloner has better performance.

For a rank bound \(1\le r\le d\), let \(E=\C^r\), let
\(\mathcal P_n^{(r)}\) be the random-purification channel with purifying
register \(E\) \cite{TangWrightZhandry2025,GirardiMeleLami2025}, and let
\(\mathcal W_{n,m}^{(dr)}\) be Werner's pure-state \(n\)-to-\(m\) cloner
in dimension \(dr\).  For \(m\ge n\), the
purify--clone--trace (PCT) channel is the following composition:
\[
  \cC_{n,m}^{\mathrm{PCT},r}
  :=\Tr_{E^{\ot m}}\circ\mathcal W_{n,m}^{(dr)}
    \circ\mathcal P_n^{(r)}.
\]
For every state \(\rho\) of rank at most \(r\), we compare
\(\cC_{n,m}^{\mathrm{PCT},r}(\rho^{\ot n})\) with the entire product target
\(\rho^{\ot m}\).  The formulas for \(\mathcal P_n^{(r)}\) and
\(\mathcal W_{n,m}^{(dr)}\) are \eqref{eq:random-purification-channel} and
\eqref{eq:grassmann-werner-cloner} in Appendix~\ref{app:pct-preliminaries}.

\begin{theorem}[Comparison with purify--clone--trace]
\label{thm:pct-comparison}
Fix \(\gamma=\lim_{n\to\infty} m_n/n>1\), and let \(\cC_{n,m}^{\mathrm{PCT},r}\) be the
purify--clone--trace channel with rank bound \(r\).
\par\noindent\textup{(a)} For \(p\in\Delta_d^\circ\), let
\(\rho_p=\diag(p_1,\ldots,p_d)\). Then PCT with rank bound \(d\) has
\begin{equation}\label{eq:full-rank-pct-comparison}
  \lim_{n\to\infty}
  F\!\left(\cC_{n,m_n}^{\mathrm{PCT},d}(\rho_p^{\ot n}),
    \rho_p^{\ot m_n}\right)
  =F_{\mathrm{PCT}}(\gamma,p)<\Funiv(\gamma,p),
\end{equation}
where \(q_{ij}=p_j/p_i\) and
\begin{equation}\label{eq:pct-fidelity-closed-form}
  F_{\mathrm{PCT}}(\gamma,p)
  =\left(\frac{\sqrt{2\gamma-1}}{\gamma}\right)^{(d-1)/2}
   \prod_{i<j}
   \frac{\sqrt{\gamma+(\gamma-1)q_{ij}}
         +\sqrt{q_{ij}(\gamma-1+\gamma q_{ij})}}
        {\gamma(1+q_{ij})}.
\end{equation}
\par\noindent\textup{(b)} Let \(1<r<d\) and let \(P\) be a rank-\(r\)
orthogonal projector on \(\C^d\). Then PCT with rank bound \(r\) satisfies
\begin{equation}\label{eq:grassmann-pct-upper-bound}
  \limsup_{n\to\infty}
  F\!\left(\cC_{n,m_n}^{\mathrm{PCT},r}((P/r)^{\ot n}),
    (P/r)^{\ot m_n}\right)
  \le\gamma^{-r(d-r)/2}
  \left(\frac{\sqrt{2\gamma-1}}{\gamma}\right)^{(r-1)/2}
  <\gamma^{-r(d-r)/2}.
\end{equation}
For \(r=1\), PCT and our cloner both reduce to Werner's pure-state
cloner.
\end{theorem}
See \cref{sec:pct-comparison} for more details. The proof is provided in Appendix~\ref{app:pct-comparison}.

\Cref{fig:pct-comparisons} illustrates the full-rank and projector-state
comparisons using two examples.

\begin{figure}[!htbp]
\centering
\begin{tikzpicture}[font=\scriptsize,
  ours/.style={green!60!black,very thick},
  pctexact/.style={purple!80!black,very thick,densely dotted},
  pctupper/.style={purple!80!black,very thick,dashed},
  grid/.style={gray!18},
  legendtext/.style={anchor=west,text=black}]
  % ---------- (a) full rank ----------
  \begin{scope}[x=1.45cm,y=7cm]
    \draw (0,0) -- (3,0);
    \draw (0,0) -- (0,0.5);
    \node[rotate=90] at (-0.5,0.25) {fidelity};
    \node at (1.5,-0.085) {gain \(\gamma\)};
    \node[font=\small] at (1.5,0.57)
      {\textup{(a)} Full rank: \(p=(3/5,2/5)\)};

    \foreach \x/\lab in {0/1,1/2,2/3,3/4}{
      \draw (\x,0.007) -- (\x,-0.007)
        node[below=1pt] {\(\lab\)};
    }
    \foreach \y/\lab in {0/{0.5},0.1/{0.6},0.2/{0.7},
                            0.3/{0.8},0.4/{0.9},0.5/{1.0}}{
      \draw (0.012,\y) -- (-0.012,\y)
        node[left=2pt] {\(\lab\)};
      \draw[grid] (0,\y) -- (3,\y);
    }

    \draw[ours]
      plot[smooth] coordinates {
        (0,0.5000) (0.25,0.4881) (0.5,0.4624) (0.75,0.4314)
        (1,0.3988) (1.25,0.3665) (1.5,0.3354) (1.75,0.3058)
        (2,0.2780) (2.25,0.2519) (2.5,0.2274) (2.75,0.2045)
        (3,0.1830)
      };
    \draw[pctexact]
      plot[smooth] coordinates {
        (0,0.5000) (0.25,0.4693) (0.5,0.4143) (0.75,0.3572)
        (1,0.3041) (1.25,0.2562) (1.5,0.2135) (1.75,0.1754)
        (2,0.1414) (2.25,0.1109) (2.5,0.0834) (2.75,0.0585)
        (3,0.0359)
      };

    \draw[ours] (0.15,0.095) -- (0.45,0.095);
    \node[legendtext] at (0.5,0.095)
      {our cloner: \(F_{\mathrm{univ}}\)};
    \draw[pctexact] (0.15,0.04) -- (0.45,0.04);
    \node[legendtext] at (0.5,0.04)
      {PCT: \(F_{\mathrm{PCT}}\)};
  \end{scope}

  % ---------- (b) projector-state fidelities; y=1+log_2(F)/5 ----------
  \begin{scope}[xshift=8.4cm,x=1.45cm,y=3.5cm]
    \draw (0,0) -- (3,0);
    \draw (0,0) -- (0,1.0);
    \node[rotate=90] at (-0.85,0.50) {fidelity (log scale)};
    \node at (1.5,-0.17) {gain \(\gamma\)};
    \node[font=\small] at (1.5,1.14)
      {\textup{(b)} Projectors: \(d=4\), \(r=2\)};

    \foreach \x/\lab in {0/1,1/2,2/3,3/4}{
      \draw (\x,0.014) -- (\x,-0.014)
        node[below=1pt] {\(\lab\)};
    }
    \foreach \y/\lab in {0/{0.03125},0.2/{0.0625},0.4/{0.125},
                            0.6/{0.25},0.8/{0.5},1.0/{1.0}}{
      \draw (0.012,\y) -- (-0.012,\y)
        node[left=2pt] {\(\lab\)};
      \draw[grid] (0,\y) -- (3,\y);
    }

    \draw[ours]
      plot[smooth] coordinates {
        (0,1.000000) (0.25,0.871229) (0.5,0.766015)
        (0.75,0.677058) (1,0.600000) (1.25,0.532030)
        (1.5,0.471229) (1.75,0.416227) (2,0.366015)
        (2.25,0.319824) (2.5,0.277058) (2.75,0.237244)
        (3,0.200000)
      };
    \draw[pctupper]
      plot[smooth] coordinates {
        (0,1.000000) (0.25,0.868284) (0.5,0.757519)
        (0.75,0.662419) (1,0.579248) (1.25,0.505405)
        (1.5,0.439036) (1.75,0.378780) (2,0.323615)
        (2.25,0.272752) (2.5,0.225571) (2.75,0.181577)
        (3,0.140368)
      };

    \draw[ours] (1.30,0.90) -- (1.60,0.90);
    \node[legendtext] at (1.65,0.90) {our projector cloner};
    \draw[pctupper] (1.30,0.79) -- (1.60,0.79);
    \node[legendtext] at (1.65,0.79) {PCT upper bound};
  \end{scope}
\end{tikzpicture}
\caption{\textbf{Fixed-gain comparison with PCT}.
  \textup{(a)}~For \(p=(3/5,2/5)\), both curves are exact asymptotic
  fidelities for protocols without the spectrum
  (\cref{thm:unknown-optimum,thm:pct-comparison}).
  \textup{(b)}~For rank-two projectors in \(d=4\), the solid curve is
  our cloner's exact fidelity \(\gamma^{-2}\). The dashed curve is the
  PCT upper bound \(\gamma^{-2}(\sqrt{2\gamma-1}/\gamma)^{1/2}\)
  (\cref{thm:grassmann,thm:pct-comparison}).  The gap is strict for
  every simple full-rank spectrum and every \(1<r<d\).}
\label{fig:pct-comparisons}
\end{figure}
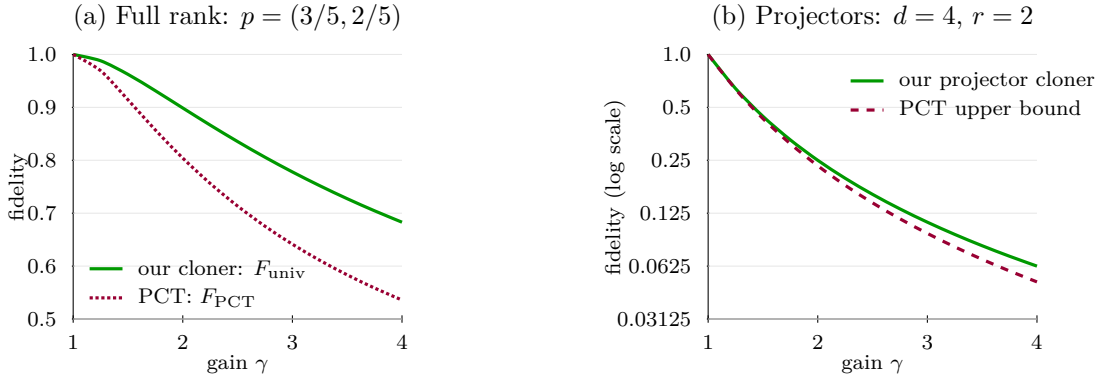

\subsection{Scope and limitations}
\label{sec:remarks}

Apart from flat projector states, our theorems require a simple full-rank
spectrum.  We do not treat nonflat rank-deficient spectra, spectral
degeneracies, or neighborhoods in which the rank changes.  The worst-case
minimax value over all states and finite-copy optimality of our channels
also remain open.  Our asymptotic estimates hold at fixed gain
\(\gamma>1\). Expected extensions to degenerate and rank-deficient spectra, and the
noncommuting limits at degeneracies, are discussed in \cref{sec:conjectures}.

\subsection{Earlier and related work}\label{sec:related-work}

\paragraph{Pure-state cloning.}
Universal qubit cloners
\cite{BuzekHillery1996,GisinMassar1997,
BrussDiVincenzoEkertFuchsMacchiavelloSmolin1998} were followed by Werner's
solution of global universal cloning in arbitrary dimension
\cite{Werner1998} and Keyl--Werner's optimality result for individual-clone
tests \cite{KeylWerner1999}.  State-dependent global cloning was developed
in \cite{BrussDiVincenzoEkertFuchsMacchiavelloSmolin1998,CheflesBarnett1999}.
The many-output limit for individual clones is tied to estimation
\cite{MassarPopescu1995,BrussEkertMacchiavello1998,BaeAcin2006,
ChiribellaDAriano2006}.  For joint fidelity, optimal measure-and-prepare
outputs need not be independent copies of an estimate
\cite{ChiribellaYang2014Optimal}.
This distinction also matters for coherent-state and phase-covariant cloning
\cite{CerfKrugerNavezWernerWolf2005,KimChitambar2022}: Cerf et al. prove
Gaussian optimality for joint fidelity, although non-Gaussian operations
can improve single-clone fidelity.  Reviews include
\cite{ScaraniIblisdirGisinAcin2005,Chiribella2011EstimationCloning,
FanWangJingYueShiZhangMu2014}.

\paragraph{Mixed-state cloning, broadcasting, and PCT.}
Exact broadcasting, and hence cloning, of a noncommuting family is
impossible \cite{BarnumCavesFuchsJozsaSchumacher1996,KalevHen2008}.
Approximate results include global-fidelity, relative-error, and
state-dependent bounds for finite ensembles
\cite{Rastegin2003Upper,Rastegin2002Relative,Rastegin2003StateDependent},
cloners for symmetric-subspace inputs and identical mixed qubits under
one-site shrinking-factor criteria \cite{Fan2003,FanLiuShi2007}, a
restricted mixed-qubit construction under averaged Hilbert--Schmidt loss
\cite{Kazakov2007}, and entropic limitations on simultaneous approximate
cloning and broadcasting \cite{LemmWilde2017}.
Cirac--Ekert--Macchiavello solved mixed-qubit purification and a
cloning/estimation problem under local fidelity
\cite{CiracEkertMacchiavello1999}.  Keyl--Werner studied asymptotic
purification rates under both one-site and joint-product tests
\cite{KeylWerner2001Purification}.  Multi-input broadcasting and
superbroadcasting also optimize marginals and permit correlations
\cite{DArianoMacchiavelloPerinotti2005,
BuscemiDArianoMacchiavelloPerinotti2005Optimal,
BuscemiDArianoMacchiavelloPerinotti2006,DangFan2007}. These criteria do not
determine fidelity with the joint product target.

Random purification \cite{TangWrightZhandry2025,
GirardiMeleLami2025} also reduces mixed-state tomography to pure-state
tomography \cite{PelecanosSpileckiTangWright2025}.  Combined with pure-state
cloning, it gives a global-fidelity guarantee for bounded-rank mixed states
\cite{LiTheilHarrowChuang2026}.  Concurrent work shows
that PCT attains the optimal sample complexity up to constants for squared
fidelity \(1-\epsilon\), uniformly over states of rank at most \(r\)
\cite{FanizzaGrinkoScharnhorstSpilecki2026,JeonSohnOh2026}.
These worst-case high-fidelity results use a different order of limits
from our exact fixed-gain values. See \cref{sec:small-error}.

\paragraph{Gaussian cloning and LAN.}
Gu{\c t}{\u a}--Matsumoto solved cloning of displaced thermal oscillator
states with known temperature under joint norm-distance loss, reducing it
to amplification and proving Gaussian optimality by stochastic ordering
\cite{GutaMatsumoto2006MixedGaussianCloning}.
Gu{\c t}{\u a}--Bowles--Adesso used LAN to derive trace-norm
teleportation benchmarks for i.i.d. qubits and displaced thermal states
\cite{GutaBowlesAdesso2010Teleportation}.
Bowles--Gu{\c t}{\u a}--Adesso used LAN for asymptotically optimal
mixed-qubit purification and dilution under global trace-norm loss
\cite{BowlesGutaAdesso2011Purification}. Related methods solve qudit channel
reversal \cite{BowlesGutaAdesso2012Reversal}.  These norm-distance results
do not determine our fidelity constants, separate spectral and orbital
costs, or finite-dimensional cloners.

The relevant LAN framework was developed in
\cite{GutaJencova2007LAN,GutaKahn2006LANQubits,
KahnGuta2009LANFiniteDimensional}.
Our fidelity amplification bound is attained by the least-noise amplifier
\cite{Caves1982LinearAmplifiers,GutaMatsumoto2006MixedGaussianCloning}.
Low-rank LAN has a different mode structure \cite{LahiryNussbaum2024}. We
treat the boundary directly only for flat projector states.

\subsection{Organization}

\Cref{sec:achievability} proves achievability for the Schur--Cartan cloners,
using the sector and label estimates of
Appendices~\ref{app:sector-fidelity} and~\ref{app:young-rounding}.
\Cref{sec:gaussian-LAN} proves full-rank optimality via LAN and the Gaussian
converse of Appendix~\ref{app:gaussian-converse}.
\Cref{sec:grassmann} treats projector states independently.
\Cref{sec:pct-comparison} compares our cloners with PCT. Its calculations
are in Appendix~\ref{app:pct-comparison}.
\Cref{sec:discussions} discusses spectral mixedness and sample complexity,
then states conjectures and open problems.

\paragraph{Acknowledgments.}
We thank Patrick Hayden and Henry Purcell for early collaborations on
mixed qubit cloning.  JW thanks Daniel Bump and Cynthia Yan for
discussions on covariant maps.  JF acknowledges support from ARO under
Grant No.~W911NF261A327.

\section*{Glossary of notation}
\addcontentsline{toc}{section}{Glossary of notation}

We use \(\U(d)\) for the unitary group and \(\cT_1(\cH)\) for
trace-class operators.  The dimension \(d\) is fixed,
\(\gamma_n=m_n/n\to\gamma>1\) gives the gain and its limit, \(N\) is
a generic sample size, and \(r_d=\binom d2\).  Fidelity is unsquared:
\(F(\rho,\sigma)=\|\sqrt\rho\sqrt\sigma\|_1\).  Unless stated otherwise,
limits are as \(n\to\infty\) with \(d,\gamma\) fixed, and \(o(1)\) is
uniform only over the parameters specified in its statement.  We usually
say \emph{Young label} for a partition \(\lambda\vdash N\), identified with
its Young diagram. \emph{Partition} and \emph{Young diagram} refer to the
same label.

\begingroup
\small
\setlength{\LTpre}{6pt}
\setlength{\LTpost}{6pt}
\renewcommand{\arraystretch}{1.42}
\begin{longtable}{@{}>{\raggedright\arraybackslash}p{.28\linewidth}>{\raggedright\arraybackslash}p{\dimexpr.72\linewidth-2\tabcolsep\relax}@{}}
\textbf{Symbols} & \textbf{Meaning} \\ \hline
\endhead
\(\mathsf A_d,\Delta_d^\circ\); \(K\Subset\Delta_d^\circ\)
& Trace-one hyperplane, simple positive spectra, and compact spectral set.
  See \eqref{eq:simple-spectrum-family} and
  \cref{thm:unknown-optimum}. \\[3pt]
\(\rho_p,\rho_{p,g}\); \(P\in\Gr(r,d)\)
& Diagonal and rotated states, also for repeated positive eigenvalues.
  \(P/r\) is the rank-\(r\) projector state.  See
  \cref{sec:main-results,thm:grassmann}. \\[3pt]
\(q_{ij},f_\gamma(q_{ij})\)
& Eigenvalue ratio and one-mode fidelity factor. See
  \eqref{eq:pairwise-fidelity-factor}. \\[3pt]
\(\Forb,\Fcl,\Funiv\)
& Known-spectrum fidelity, extra factor for unknown spectra, and their
  product. See
  \eqref{eq:main-values}. \\[3pt]
\(\mathsf Y_N\); \(\cH_\lambda,\cS_\lambda\)
& Young labels and their representation spaces. See
  \eqref{eq:schur-weyl-decomposition}. \\[3pt]
\(P_{N,p},\rho_{p,\lambda}\); \(\mathsf T_N(p)\)
& Young-label law, conditional sector state, and typical-label set. See
  \eqref{eq:schur-state} and \eqref{eq:young-typical-set}. \\[3pt]
\(\mathsf K_n^p,\mathsf K_n^{\mathrm{rnd}},\mathsf K_n^{\mathrm{proj}}\); \(\mathcal Q_n[\mathsf K_n]\)
& The three label rules and their Schur--Cartan lift. See
  \cref{sec:cartan-sector-lifting,sec:grassmann-achievability}. \\[3pt]
\(\cC_{\mu,\nu},\cE_{\mu,\nu}\)
& Sector channel and its extension to all label pairs. See
  \eqref{eq:sector-channel-formula} and
  \eqref{eq:sector-channel-extension}. \\[3pt]
\(\cM_n\); \(\cC_n^p,\cC_n^{\mathrm{univ}},\cC_n^{\mathrm{proj}}\)
& Arbitrary cloner and our three constructions. See
  \eqref{eq:known-cloner-definition}, \eqref{eq:universal-cloner-definition},
  and \eqref{eq:projector-cloner}. \\[3pt]
\(\mathsf H\); \(h,z\)
& Eigenvalue-change space and local eigenvalue/eigenbasis coordinates. See
  \cref{sec:gaussian-amplification}. \\[3pt]
\(Y_p(z),g_{n,z}^{(p)}\); \(\Omega_L\)
& Generator and unitary for a small eigenbasis rotation, bounded parameter
  box.  See
  \cref{sec:gaussian-amplification}. \\[3pt]
\(\cF,\cF_{\mathrm{orb}}\); \(\widehat N\)
& One- and \(r_d\)-oscillator spaces, and excitation-counting operator. See
  \cref{sec:gaussian-amplification}. \\[3pt]
\(\tau_p\); \(\tau_p^{(\gamma)}\)
& Thermal state and its amplified version. See
  \eqref{eq:sector-thermal-products}. \\[3pt]
\(\Phi_z^{(p)},\Psi_z^{(p,\gamma)}\)
& Gaussian input and ideal target for eigenbasis changes. See
  \eqref{eq:gaussian-amplification-model}. \\[3pt]
\(\Theta_{h,z}^{\mathrm{in}},\Theta_{h,z}^{\mathrm{tar}}\)
& Gaussian input and ideal target including eigenvalue changes. See
  \eqref{eq:full-gaussian-model}.
\end{longtable}
\endgroup
Symbols confined to one proof are defined there.

\section{Cloner construction and achievability}\label{sec:achievability}

The known- and unknown-spectrum protocols share the architecture in
\cref{fig:schur-cartan-cloner}: a transition rule on Young labels lifts to a
covariant quantum channel.  We construct the two rules and combine their
label estimates with a common sector analysis to evaluate the fidelities
in \cref{sec:achievability-completion}.

Let $\mathsf Y_N=\{\lambda\vdash N:\ell(\lambda)\le d\}$.  Schur--Weyl duality gives
\begin{equation}\label{eq:schur-weyl-decomposition}
  (\C^d)^{\ot N}
  \cong
  \bigoplus_{\lambda\in\mathsf Y_N}\cH_\lambda\ot\cS_\lambda,
\end{equation}
where \(\cH_\lambda\) is the irreducible \(\U(d)\)-module of highest
weight \(\lambda\) and \(\cS_\lambda\) is the corresponding Specht module.
For \(\rho_p=\diag(p_1,\ldots,p_d)\),
\begin{equation}\label{eq:schur-state}
  \rho_p^{\ot N}
  =\bigoplus_{\lambda\in\mathsf Y_N}
  P_{N,p}(\lambda)\,
  \rho_{p,\lambda}\ot\frac{I_{\cS_\lambda}}{\dim\cS_\lambda},
  \qquad
  P_{N,p}(\lambda)=\dim\cS_\lambda\,s_\lambda(p),
  \quad
  \rho_{p,\lambda}=\frac{\pi_\lambda(\rho_p)}{s_\lambda(p)}.
\end{equation}
Here \(\pi_\lambda\) is the irreducible polynomial representation of highest
weight \(\lambda\), acting on \(\cH_\lambda\), and
\(s_\lambda(p)=\Tr\pi_\lambda(\rho_p)\) is the corresponding Schur
polynomial.  For \(\rho_{p,g}=g\rho_p g^*\), replace \(\rho_{p,\lambda}\) by
\(\rho_{p,g,\lambda}=U_\lambda(g)\rho_{p,\lambda}U_\lambda(g)^*\), where
\(U_\lambda(g)=\pi_\lambda(g)\).  The
Young law \(P_{N,p}\) depends on \(p\), but not on \(g\).

Define the typical Young-label set
\begin{equation}\label{eq:young-typical-set}
  \mathsf T_N(p)
  =\left\{\lambda\in\mathsf Y_N:
  \left\|\frac\lambda N-p\right\|_\infty\le N^{-1/3}\right\}.
\end{equation}

\begin{lemma}[Uniform Young concentration]\label{lem:young-concentration}
For all sufficiently large \(N\), depending only on \(d\),
\begin{equation}\label{eq:uniform-young-concentration}
  \sup_{\substack{p_1\ge\cdots\ge p_d>0\\\sum_i p_i=1}}
  P_{N,p}(\mathsf T_N(p)^c)
  \le e^{-N^{1/3}/4}.
\end{equation}
Equal positive eigenvalues are allowed.
\end{lemma}

\begin{proof}
The Weyl dimension formula gives
\(\dim\cH_\lambda\le(N+1)^{d(d-1)/2}\).  Every weight \(\xi\) of
\(\cH_\lambda\) has \(\lambda-\xi\) equal to a nonnegative integral
combination of positive roots.  Since \(p_j/p_i\le1\) for \(i<j\), this gives
\(p^\xi\le p^\lambda\).  Hence
\(s_\lambda(p)\le \dim\cH_\lambda\,p^\lambda\).
The standard bounds
\(\dim\cS_\lambda\le\binom{N}{\lambda_1,\ldots,\lambda_d}\) and
\(\binom{N}{\lambda_1,\ldots,\lambda_d}p^\lambda
\le \exp[-ND(\lambda/N\|p)]\) give
\begin{equation}\label{eq:finite-Young-large-deviation}
  P_{N,p}(\lambda)
  \le (N+1)^{d(d-1)/2}\exp[-ND(\lambda/N\|p)].
\end{equation}
Here \(D\) is classical relative entropy with natural logarithms
\cite{KeylWerner2001Spectrum}.

There are at most \((N+1)^{d-1}\) partitions of \(N\) with at most \(d\)
rows.  If \(\lambda\notin\mathsf T_N(p)\), then
\(\|\lambda/N-p\|_1>N^{-1/3}\), and Pinsker's inequality gives
\(D(\lambda/N\|p)\ge\tfrac12N^{-2/3}\).  Therefore
\[
  P_{N,p}(\mathsf T_N(p)^c)
  \le (N+1)^{d(d-1)/2+d-1}e^{-N^{1/3}/2}.
\]
For all sufficiently large \(N\), the polynomial factor is at most
\(e^{N^{1/3}/4}\), which proves \eqref{eq:uniform-young-concentration}.
\end{proof}

\subsection{Cartan-sector channel and label rules}
\label{sec:cartan-sector-lifting}

Let \(\lambda,\theta\) be dominant weights, with unit highest-weight vectors
\(v_\lambda,v_\theta\).  The Cartan component \(\cH_{\lambda+\theta}\)
occurs with multiplicity one in \(\cH_\lambda\ot\cH_\theta\).  Let
\begin{equation}\label{eq:cartan-inclusion}
  V_{\lambda,\theta}:\cH_{\lambda+\theta}
  \longrightarrow\cH_\lambda\ot\cH_\theta,
  \qquad
  V_{\lambda,\theta}v_{\lambda+\theta}=v_\lambda\ot v_\theta,
\end{equation}
be the normalized intertwining isometry.  For \(\mu\in\mathsf Y_n\) and
\(\nu\in\mathsf Y_m\), write \(\nu-\mu\vdash m-n\) when the coordinatewise
difference \(\nu-\mu\) is a partition of \(m-n\), and call such a pair
\((\mu,\nu)\) \emph{compatible}.  If this holds, let
\(\omega=\nu-\mu\) and define
\begin{equation}\label{eq:sector-channel-formula}
  \cC_{\mu,\nu}(X)
  =\frac{\dim\cH_\mu}{\dim\cH_\nu}
  V_{\mu,\omega}^*(X\ot I_{\cH_\omega})V_{\mu,\omega}.
\end{equation}
The partial trace
\(\Tr_{\cH_\omega}(V_{\mu,\omega}V_{\mu,\omega}^*)\) commutes with the
irreducible \(\U(d)\)-action on \(\cH_\mu\).  Hence Schur's lemma and a trace
comparison give
\(
  \Tr_{\cH_\omega}(V_{\mu,\omega}V_{\mu,\omega}^*)
  =\frac{\dim\cH_\nu}{\dim\cH_\mu}I_{\cH_\mu}
\).
The map in \eqref{eq:sector-channel-formula} is completely positive.
The identity above makes it trace preserving, and equivariance of
\(V_{\mu,\omega}\) gives \(\U(d)\)-covariance.

When \(\mu=(n,0,\ldots,0)\) and \(\nu=(m,0,\ldots,0)\), the three
irreps in \eqref{eq:sector-channel-formula} are symmetric,
and the formula is Werner's universal pure-state \(n\)-to-\(m\) cloner
\cite{Werner1998}.  Thus the Cartan-sector channel extends Werner's
construction from symmetric tensors to arbitrary compatible Schur sectors.

For every pair \((\mu,\nu)\in\mathsf Y_n\times\mathsf Y_m\), define a
covariant channel \(\cE_{\mu,\nu}:\cT_1(\cH_\mu)\to\cT_1(\cH_\nu)\) by
\begin{equation}\label{eq:sector-channel-extension}
  \cE_{\mu,\nu}(X)=
  \begin{cases}
    \cC_{\mu,\nu}(X),&\nu-\mu\vdash m-n,\\[1mm]
    \Tr(X)I_{\cH_\nu}/\dim\cH_\nu,&\text{otherwise}.
  \end{cases}
\end{equation}
Let \(\mathsf K_n:\mathsf Y_n\rightsquigarrow\mathsf Y_{m_n}\) be any transition
rule: for each input label \(\mu\), the output label \(\nu\) is selected
with probability \(\mathsf K_n(\nu\mid\mu)\ge0\), where
\(\sum_\nu \mathsf K_n(\nu\mid\mu)=1\).  Its Schur--Cartan
lifting \(\mathcal Q_n[\mathsf K_n]\) acts as follows: apply the Schur transform, measure
\(\mu\), discard the input Specht factor, sample \(\nu\) from
\(\mathsf K_n(\cdot\mid\mu)\), apply \(\cE_{\mu,\nu}\), adjoin
\(I_{\cS_\nu}/\dim\cS_\nu\), and apply the inverse Schur transform.  Each
step is completely positive and trace preserving, so their composition is a
channel.  Because the label measurement and sampling are invariant and every
map \(\cE_{\mu,\nu}\) is covariant, \(\mathcal Q_n[\mathsf K_n]\) is
\(\U(d)\)-covariant.

\begin{proposition}[Fidelity of a lifted transition rule]
\label{prop:lifted-kernel-bound}
Fix \(p\in\Delta_d^\circ\) and \(g\in\U(d)\).  Let
\begin{equation}\label{eq:lifted-label-laws}
  \mathsf J_{n,p}(\mu,\nu)=P_{n,p}(\mu)\mathsf K_n(\nu\mid\mu),
  \qquad
  Q_{m_n,p}(\nu)=\sum_\mu\mathsf J_{n,p}(\mu,\nu),
\end{equation}
so \(\mathsf J_{n,p}\) is the joint input-output label law and \(Q_{m_n,p}\)
its output marginal.  For \(\nu\in\mathsf Y_{m_n}\), let
\begin{equation}\label{eq:lifted-output-block}
  \Xi_{n,p,\nu}
  =\sum_{\mu\in\mathsf Y_n}\mathsf J_{n,p}(\mu,\nu)\,\cE_{\mu,\nu}(\rho_{p,\mu})
\end{equation}
be the unnormalized irreducible output block, so that
\(\Tr\Xi_{n,p,\nu}=Q_{m_n,p}(\nu)\).
The lifted channel has fidelity
\begin{equation}\label{eq:lifted-exact-block-fidelity}
  F\!\left(
    \mathcal Q_n[\mathsf K_n](\rho_{p,g}^{\ot n}),
    \rho_{p,g}^{\ot m_n}
  \right)
  =\sum_{\nu\in\mathsf Y_{m_n}}
   \sqrt{P_{m_n,p}(\nu)}\,F(\Xi_{n,p,\nu},\rho_{p,\nu}).
\end{equation}
\end{proposition}

\begin{proof}
By covariance of \(\mathcal Q_n[\mathsf K_n]\) and unitary invariance of fidelity,
we may take \(g=I\).  Under the output Schur decomposition, the actual and
target states are
\[
 \bigoplus_{\nu\in\mathsf Y_{m_n}}
 \Xi_{n,p,\nu}\ot\frac{I_{\cS_\nu}}{\dim\cS_\nu},
 \qquad
 \bigoplus_{\nu\in\mathsf Y_{m_n}}
 P_{m_n,p}(\nu)\rho_{p,\nu}
 \ot\frac{I_{\cS_\nu}}{\dim\cS_\nu}.
\]
Fidelity is additive over orthogonal direct sums, multiplicative over
tensor products, and satisfies \(F(A,cB)=\sqrt c\,F(A,B)\) for \(c\ge0\).
Since \(F(I_{\cS_\nu}/\dim\cS_\nu,I_{\cS_\nu}/\dim\cS_\nu)=1\), this proves
\eqref{eq:lifted-exact-block-fidelity}.
\end{proof}

\paragraph{Known spectrum: exact target-label sampling.}

Fix \(p\in\Delta_d^\circ\).  Since \(p\) is known, we can sample the target
Young-label law \(P_{m_n,p}\) directly, independently of the input label
\(\mu\).  Typical input and output labels are compatible, so the
Cartan-sector channel applies.
Define
\begin{equation}\label{eq:known-target-label-kernel}
  \mathsf K_n^p(\nu\mid\mu)=P_{m_n,p}(\nu).
\end{equation}
Set
\begin{equation}\label{eq:known-cloner-definition}
  \cC_n^p:=\mathcal Q_n[\mathsf K_n^p].
\end{equation}
Both cases of \eqref{eq:sector-channel-extension} output into the sampled
\(\nu\)-sector.  Thus the output Young-label law on
\(\rho_{p,g}^{\ot n}\) is
\[
  Q_{m_n,p}(\nu)
  =\sum_\mu P_{n,p}(\mu)\mathsf K_n^p(\nu\mid\mu)
  =P_{m_n,p}(\nu)\sum_\mu P_{n,p}(\mu)
  =P_{m_n,p}(\nu)
\]
for every \(n\) and \(g\).

\paragraph{Unknown spectrum: randomized rounding.}

For a typical input label \(\mu\in\mathsf Y_n\), we want an output label
near \(\gamma_n\mu\), where \(\gamma_n=m_n/n\): the dilation has total size
\(m_n\) and preserves the empirical spectrum, since
\(\gamma_n\mu/m_n=\mu/n\).  Its coordinates need not be integers, so we
round it randomly.

For \(N\in\mathbb N\), let
\[
  \mathsf L_N=\left\{x\in\mathbb Z^d:\sum_i x_i=N\right\}.
\]
Given \(\mu\in\mathsf Y_n\), draw independent variables
\(y_1,\ldots,y_{d-1}\), each uniform on \([-1/2,1/2)\), and set
\begin{equation}\label{eq:randomized-rounding}
  \widetilde\nu_i
  =\left\lfloor\gamma_n(\mu_i+y_i)+\frac12\right\rfloor
  \quad(1\le i<d),
  \qquad
  \widetilde\nu_d=m_n-\sum_{i<d}\widetilde\nu_i.
\end{equation}
Here \(\lfloor x+1/2\rfloor\) rounds \(x\) to the nearest integer.
Thus \(\widetilde\nu\in\mathsf L_{m_n}\).  Let
\(\mathsf R_n(\widetilde\nu\mid\mu)\) denote its distribution.
The shifts spread mass across the output lattice.  To see why this matters,
suppose \(\gamma_n=\gamma\) is an integer and instead set \(\nu=\gamma\mu\).
This rule restricts the output to a sublattice of index
\(\gamma^{d-1}\) in \(\mathsf L_{m_n}\).
The Young local limit gives label affinity
\(\gamma^{-(d-1)/2}\Fcl(\gamma,d)\), rather than
\(\Fcl(\gamma,d)\).\footnote{For \(p=(3/5,2/5)\) and \(\gamma=2\),
the undithered and rounded label affinities are \(0.6866\) and \(0.9710\),
respectively.  After the common sector factor, the corresponding cloning
fidelities are \(0.6355\) and \(0.8988\).}
Retain the rounded label when compatible, and otherwise
use the one-row partition \((m_n,0,\ldots,0)\) as the fallback:
\begin{equation}\label{eq:rounding-kernel}
  \nu=
  \begin{cases}
    \widetilde\nu,&\widetilde\nu-\mu\vdash m_n-n,\\[1mm]
    (m_n,0,\ldots,0),&\text{otherwise}.
  \end{cases}
\end{equation}
The conditional law of \(\nu\) defines a transition rule
\(\mathsf K_n^{\mathrm{rnd}}:\mathsf Y_n\rightsquigarrow\mathsf Y_{m_n}\).
Define
\begin{equation}\label{eq:universal-cloner-definition}
  \cC_n^{\mathrm{univ}}:=\mathcal Q_n[\mathsf K_n^{\mathrm{rnd}}].
\end{equation}
Since the transition rule is independent of \(p\), this is a spectrum-independent,
\(\U(d)\)-covariant cloner.

\subsection{Sector and label fidelity estimates}
\label{sec:uniform-estimates-achievability}

The block fidelity formula \eqref{eq:lifted-exact-block-fidelity} leaves two
tasks: evaluate the conditional state fidelity within typical output
sectors, and compare the output Young-label law with the target law.  The
sector estimate must hold even when several input labels lead to the same
output label.  It applies to both cloners.  Only the unknown-spectrum rule
needs a separate label estimate.

\begin{proposition}[Uniform Cartan-sector fidelity]
\label{prop:sector-fidelity}
Let \(K\Subset\Delta_d^\circ\) and let \(\delta_n\ge0\) satisfy
\(\delta_n\to0\).  For all sufficiently large \(n\), every
\(p\in K\), \(\mu\in\mathsf Y_n\), and \(\nu\in\mathsf L_{m_n}\) satisfying
\[
  \left\|\frac\mu n-p\right\|_\infty\le\delta_n,
  \qquad
  \left\|\frac\nu{m_n}-p\right\|_\infty\le\delta_n
\]
also satisfies \(\nu-\mu\vdash m_n-n\).  In particular,
\(\nu\in\mathsf Y_{m_n}\).  Moreover,
\begin{equation}\label{eq:uniform-cartan-sector-fidelity}
  F\!\left(\cC_{\mu,\nu}(\rho_{p,\mu}),\rho_{p,\nu}\right)
  =\Forb(\gamma,p)+o(1)
\end{equation}
uniformly over all such \(p,\mu,\nu\).
More generally, for fixed \(p,\nu\) as above and any probability distribution
\((a_\mu)_\mu\) supported on input labels satisfying the same closeness
condition,
\begin{equation}\label{eq:sector-mixture-fidelity}
  F\!\left(
    \sum_\mu a_\mu\,
      \cC_{\mu,\nu}(\rho_{p,\mu}),
    \rho_{p,\nu}
  \right)
  =\Forb(\gamma,p)+o(1),
\end{equation}
uniformly over these distributions as well as \(p,\nu\).
\end{proposition}

The factor \(\Forb(\gamma,p)\) naturally appears as the fidelity of the
thermal states from \cref{sec:gaussian-amplification}:
\[
  \Forb(\gamma,p)=F(\tau_p^{(\gamma)},\tau_p).
\]
To see this in the sector calculation, let \(\alpha=e_i-e_j\) with
\(i<j\), \(q_\alpha=p_j/p_i\), and
\(q_\alpha^{(\gamma)}=1-(1-q_\alpha)/\gamma\).  On weight spaces reached
by a bounded number of lowerings from the highest weight \(\nu\), an
orthonormalized Poincar\'e--Birkhoff--Witt (PBW) basis makes \(\rho_{p,\nu}\) diagonal and
\(\cC_{\mu,\nu}(\rho_{p,\mu})\) asymptotically diagonal.  At a fixed
basis index \(\mathbf r=(r_\alpha)_\alpha\), their diagonal entries
converge to the corresponding number-basis entries:
\[
\begin{aligned}
  (\rho_{p,\nu})_{\mathbf r,\mathbf r}
  &\longrightarrow \prod_\alpha(1-q_\alpha)q_\alpha^{r_\alpha}
   =(\tau_p)_{\mathbf r,\mathbf r},\\
  (\cC_{\mu,\nu}(\rho_{p,\mu}))_{\mathbf r,\mathbf r}
  &\longrightarrow \prod_\alpha(1-q_\alpha^{(\gamma)})
   (q_\alpha^{(\gamma)})^{r_\alpha}
   =(\tau_p^{(\gamma)})_{\mathbf r,\mathbf r}.
\end{aligned}
\]
The powers of \(q_\alpha^{(\gamma)}\) in the output entries come from the
binomial expansion of the Cartan intertwiner: with
\(s_{\alpha,n}=(\mu_i-\mu_j)/(\nu_i-\nu_j)\to\gamma^{-1}\), its factor
\((1-s_{\alpha,n}+s_{\alpha,n}q_\alpha)^{r_\alpha}\) tends to
\((q_\alpha^{(\gamma)})^{r_\alpha}\).  A uniform tail bound extends this
fixed-weight comparison to the full fidelity in
\eqref{eq:uniform-cartan-sector-fidelity}.  See
Appendix~\ref{app:sector-fidelity}.

The second estimate concerns only \(\mathsf K_n^{\mathrm{rnd}}\).

\begin{proposition}[Randomized Young-label dilation]
\label{prop:young-rounding}
There is a constant \(c_{\mathrm{rnd}}<\infty\), depending only on \(d\) and
the fixed sequence \((\gamma_n)\), such that, for every compact
\(K\Subset\Delta_d^\circ\) and all sufficiently large \(n\), every
\(p\in K\) and \(\mu\in\mathsf T_n(p)\) satisfy
\begin{equation}\label{eq:rounding-support}
  \operatorname{supp}\mathsf K_n^{\mathrm{rnd}}(\,\cdot\mid\mu)
  \subseteq
  \left\{\nu\in\mathsf Y_{m_n}:
    \nu-\mu\vdash m_n-n,\quad
    \|\nu-\gamma_n\mu\|_\infty\le c_{\mathrm{rnd}}
  \right\}.
\end{equation}
For every such \(K\), the induced output-label law \(Q_{m_n,p}\) from
\eqref{eq:lifted-label-laws} with \(\mathsf K_n=\mathsf K_n^{\mathrm{rnd}}\),
\begin{equation}\label{eq:rounding-label-fidelity}
  F(Q_{m_n,p},P_{m_n,p})
  =\Fcl(\gamma,d)+o(1)
  \qquad\text{uniformly for }p\in K.
\end{equation}
\end{proposition}

The randomized transition rule \(\mathsf K_n^{\mathrm{rnd}}\) is a discrete
realization of the dilation \(\mu\mapsto\gamma_n\mu\): the uniform shifts in
\eqref{eq:randomized-rounding} spread each input label over nearby output
labels.  Represent the input Young-label law \(P_{n,p}\) by its
piecewise-constant interpolation \(\mathcal I_{n,p}P_{n,p}\), formed by
centering at \(np\), dividing by \(\sqrt n\), and spreading each label mass
over its lattice cell.  If \(\widetilde Q_{m_n,p}:=P_{n,p}\mathsf R_n\) is the law of
the rounded label \(\widetilde\nu\), the intersection-volume identity
\eqref{eq:cell-kernel}
identifies its embedded density \(\mathcal I_{m_n,p}\widetilde Q_{m_n,p}\)
with the \(\sqrt{\gamma_n}\)-dilation of \(\mathcal I_{n,p}P_{n,p}\), averaged
over each cell of the \(m_n\)-copy target-label
lattice.  Write
\begin{equation}\label{eq:young-covariance}
  \Sigma_p
  :=\left.(\diag(p)-pp^{\mathsf T})\right|_{
    \{x\in\R^d:\,\sum_i x_i=0\}}.
\end{equation}
The Young local limit says that \(\mathcal I_{N,p}P_{N,p}\) converges to
the density \(\varphi_{0,\Sigma_p}\) of \(\mathsf N_{0,\Sigma_p}\), both
for the input law at \(N=n\) and the ideal output law at \(N=m_n\).
Thus the rounded output converges to
\(\mathsf N_{0,\gamma\Sigma_p}\), while the ideal output converges to
\(\mathsf N_{0,\Sigma_p}\).  Their fidelity (Hellinger affinity) is
\begin{equation}\label{eq:label-gaussian-affinity}
  F\!\left(\mathsf N_{0,\gamma\Sigma_p},
           \mathsf N_{0,\Sigma_p}\right)
  =\left(\frac{2\sqrt\gamma}{1+\gamma}\right)^{(d-1)/2}
  =\Fcl(\gamma,d).
\end{equation}
For typical input labels, every rounded output is compatible for all
sufficiently large \(n\).  Fallback therefore changes only atypical inputs,
whose total probability vanishes uniformly on compact spectral sets.
The full proof is given in
Appendix~\ref{app:young-rounding}.

\subsection{Cloning fidelity}
\label{sec:achievability-completion}

We use the following continuity estimate for positive trace-class
operators \(A,C,A',C'\) with trace at most one:
\begin{equation}\label{eq:fidelity-continuity}
  |F(A,C)-F(A',C')|
  \le \|A-A'\|_1^{1/2}+\|C-C'\|_1^{1/2}.
\end{equation}
Indeed, the Powers--St\o rmer inequality bounds
\(\|\sqrt A-\sqrt{A'}\|_2\) by \(\|A-A'\|_1^{1/2}\), and similarly for
\(C,C'\).  Adding and subtracting \(\sqrt{A'}\sqrt C\), the reverse triangle
inequality and Schatten H\"older give \eqref{eq:fidelity-continuity}, since
\(\|\sqrt C\|_2,\|\sqrt{A'}\|_2\le1\).

The sector-mixture estimate gives a common factorization for both cloners.

\begin{lemma}[Factorization of the lifted fidelity]
\label{lem:lifted-fidelity-factorization}
Fix \(K\Subset\Delta_d^\circ\).  For each \(p\in K\), choose a transition
rule \(\mathsf K_n^{(p)}:\mathsf Y_n\rightsquigarrow\mathsf Y_{m_n}\).
Use \(\mathsf J_{n,p},Q_{m_n,p}\) from \eqref{eq:lifted-label-laws}
and \(\Xi_{n,p,\nu}\) from \eqref{eq:lifted-output-block} with
\(\mathsf K_n=\mathsf K_n^{(p)}\).
\(P_{m_n,p}\) is the target Young-label law from \eqref{eq:schur-state}.
Suppose that,
for some \(\delta_n\to0\), the set of close label pairs
\[
 G_{n,p}
 =\left\{(\mu,\nu)\in\mathsf Y_n\times\mathsf Y_{m_n}:
   \left\|\frac\mu n-p\right\|_\infty\le\delta_n,\
   \left\|\frac\nu{m_n}-p\right\|_\infty\le\delta_n\right\}
\]
satisfies \(\sup_{p\in K}\mathsf J_{n,p}(G_{n,p}^c)\to0\).
Then, uniformly for \(p\in K\) and \(g\in\U(d)\),
\begin{equation}\label{eq:lifted-fidelity-factorization}
 F\!\left(\mathcal Q_n[\mathsf K_n^{(p)}](\rho_{p,g}^{\ot n}),
                         \rho_{p,g}^{\ot m_n}\right)
 =\Forb(\gamma,p)F(Q_{m_n,p},P_{m_n,p})+o(1).
\end{equation}
\end{lemma}

\begin{proof}
By covariance, take \(g=I\).  Write \(P=P_{m_n,p}\), \(Q=Q_{m_n,p}\),
and \(G=G_{n,p}\).  Omitting the maximally mixed Specht factors,
the output and target are \(O=\bigoplus_\nu\Xi_{n,p,\nu}\) and
\(T=\bigoplus_\nu P(\nu)\rho_{p,\nu}\).  For large \(n\), all pairs in
\(G\) are compatible by \cref{prop:sector-fidelity}.  Retain only these pairs:
\[
 O_G=\bigoplus_\nu\Xi_\nu^G,
 \qquad
 \Xi_\nu^G=\sum_{\mu:(\mu,\nu)\in G}
      \mathsf J_{n,p}(\mu,\nu)\cC_{\mu,\nu}(\rho_{p,\mu}),
 \qquad Q_G(\nu)=\Tr\Xi_\nu^G.
\]
Whenever \(Q_G(\nu)>0\), the mixture statement of
\cref{prop:sector-fidelity} applies to \(\Xi_\nu^G/Q_G(\nu)\).
Labels with \(Q_G(\nu)=0\) contribute zero to both fidelities. Thus the
direct-sum identity \eqref{eq:lifted-exact-block-fidelity} and
\(F(Q_G,P)\le1\) give
\[
 \bigl|F(O_G,T)-\Forb(\gamma,p)F(Q_G,P)\bigr|\le\eta_n,
 \qquad \eta_n\longrightarrow0,
\]
uniformly on \(K\).  Positivity and trace preservation give
\(\|O-O_G\|_1=\|Q-Q_G\|_{\ell^1}=\mathsf J_{n,p}(G^c)=:\epsilon_{n,p}\).
Applying \eqref{eq:fidelity-continuity} once to the states and once to
the label laws yields
\[
 \bigl|F(O,T)-\Forb(\gamma,p)F(Q,P)\bigr|
 \le\eta_n+2\sqrt{\epsilon_{n,p}}=o(1),
\]
which proves \eqref{eq:lifted-fidelity-factorization}.
\end{proof}

We now complete the achievability proofs of \cref{thm:known-optimum,thm:unknown-optimum}.
\paragraph{Known spectrum.}
Fix any compact \(K\Subset\Delta_d^\circ\).  For each \(p\in K\), take
the known-spectrum rule \(\mathsf K_n^{(p)}=\mathsf K_n^p\) from
\eqref{eq:known-target-label-kernel}.
The joint law is \(P_{n,p}\ot P_{m_n,p}\) and
\(Q_{m_n,p}=P_{m_n,p}\).  With \(\delta_n=n^{-1/3}\ge m_n^{-1/3}\),
\(G_{n,p}\) contains \(\mathsf T_n(p)\times\mathsf T_{m_n}(p)\), so Young
concentration makes the discarded mass vanish uniformly on \(K\), while
the label fidelity is exactly one.  Thus
\eqref{eq:lifted-fidelity-factorization} gives
\[
 \sup_{\substack{p\in K\\g\in\U(d)}}
 \left|F\!\left(\cC_n^p(\rho_{p,g}^{\ot n}),
                      \rho_{p,g}^{\ot m_n}\right)
             -\Forb(\gamma,p)\right|\longrightarrow0,
\]
proving achievability in \cref{thm:known-optimum}.

\paragraph{Unknown spectrum.}
For every \(p\), take \(\mathsf K_n^{(p)}=\mathsf K_n^{\mathrm{rnd}}\) and
\(\delta_n=n^{-1/3}+c_{\mathrm{rnd}}/m_n\).  By
\eqref{eq:rounding-support}, every pair in the support of
\(\mathsf J_{n,p}\) with \(\mu\in\mathsf T_n(p)\) lies in \(G_{n,p}\), so the
discarded mass is at most \(P_{n,p}(\mathsf T_n(p)^c)=o(1)\) uniformly on
\(K\).  Combining
\eqref{eq:lifted-fidelity-factorization} and
\eqref{eq:rounding-label-fidelity} yields
\[
 F\!\left(\cC_n^{\mathrm{univ}}(\rho_{p,g}^{\ot n}),
                                    \rho_{p,g}^{\ot m_n}\right)
 =\Fcl(\gamma,d)\Forb(\gamma,p)+o(1)
\]
uniformly for \(p\in K\) and \(g\in\U(d)\), proving the achievability
assertion of \cref{thm:unknown-optimum}.

\section{Optimality from LAN and Gaussian converses}\label{sec:gaussian-LAN}

We use two-way LAN to transport arbitrary cloners on local neighborhoods to
channels between Gaussian models.  After establishing the Gaussian bounds,
we apply the transfer in \cref{fig:LAN-transfer} to complete the minimax
proofs.

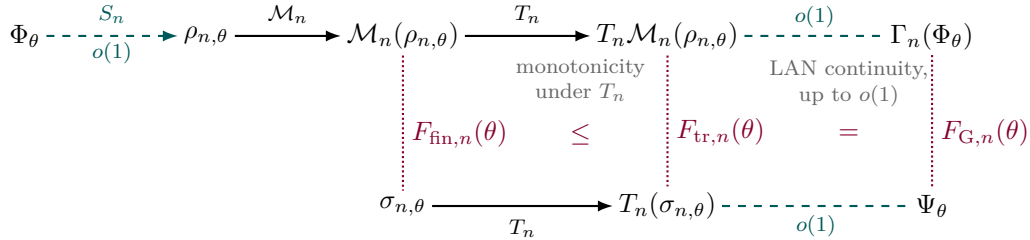
\begin{figure}[!htbp]
\centering
\begin{tikzpicture}[
  exact/.style={-{Latex[length=2mm]},thick,black},
  approxmap/.style={-{Latex[length=2mm]},thick,dashed,teal!70!black},
  approx/.style={thick,dashed,teal!70!black},
  fidpair/.style={thick,densely dotted,purple!70!black},
  state/.style={font=\small,inner sep=2pt},
  maplabel/.style={font=\scriptsize},
  fidlabel/.style={font=\small,text=purple!70!black,inner sep=1pt,yshift=-4pt},
  relation/.style={font=\small,text=purple!70!black},
  reason/.style={font=\scriptsize,text=black!60,align=center},
  error/.style={font=\scriptsize,text=teal!70!black,inner sep=1pt}
]
  \node[state] (gin) at (-6.6,1.3) {\(\Phi_\theta\)};
  \node[state] (finin) at (-4.2,1.3) {\(\rho_{n,\theta}\)};
  \node[state] (finout) at (-1.6,1.3) {\(\cM_n(\rho_{n,\theta})\)};
  \node[state] (trout) at (1.9,1.3) {\(T_n\cM_n(\rho_{n,\theta})\)};
  \node[state] (gout) at (5.4,1.3) {\(\Gamma_n(\Phi_\theta)\)};

  \node[state] (fintarget) at (-1.6,-0.95) {\(\sigma_{n,\theta}\)};
  \node[state] (trtarget) at (1.9,-0.95) {\(T_n(\sigma_{n,\theta})\)};
  \node[state] (gtarget) at (5.4,-0.95) {\(\Psi_\theta\)};

  \draw[approxmap] (gin.east) --
    node[above=1pt,maplabel] {\(S_n\)}
    node[below=1pt,error] {\(o(1)\)} (finin.west);
  \draw[exact] (finin.east) --
    node[above=1pt,maplabel] {\(\cM_n\)} (finout.west);
  \draw[exact] (finout.east) --
    node[above=1pt,maplabel] {\(T_n\)} (trout.west);
  \draw[exact] (fintarget.east) --
    node[below=1pt,maplabel] {\(T_n\)} (trtarget.west);

  \draw[approx] (trout.east) --
    node[above=2pt,error] {\(o(1)\)} (gout.west);
  \draw[approx] (trtarget.east) --
    node[below=2pt,error] {\(o(1)\)} (gtarget.west);

  \draw[fidpair] (finout.south) --
    node[right=2pt,fidlabel] {\(F_{\mathrm{fin},n}(\theta)\)}
    (fintarget.north);
  \draw[fidpair] (trout.south) --
    node[right=2pt,fidlabel] {\(F_{\mathrm{tr},n}(\theta)\)}
    (trtarget.north);
  \draw[fidpair] (gout.south) --
    node[right=2pt,fidlabel] {\(F_{\mathrm G,n}(\theta)\)}
    (gtarget.north);

  \node[relation] (leq) at (0.75,0) {\(\le\)};
  \node[reason,anchor=south] at ([yshift=1pt]leq.north)
    {monotonicity\\ under \(T_n\)};
  \node[relation] (eq) at (4.3,0) {\(=\)};
  \node[reason,anchor=south] at ([yshift=1pt]eq.north)
    {LAN continuity,\\ up to \(o(1)\)};
\end{tikzpicture}
\caption{LAN transfer for \(\Gamma_n=T_n\circ\cM_n\circ S_n\) on a fixed
  compact local parameter set \(C\).  The top and bottom rows track output
  and target. Dotted segments indicate fidelity pairs, solid arrows exact
  channels, and dashed links uniform \(o(1)\) trace-norm errors.
  Monotonicity and continuity give
  \(F_{\mathrm{fin},n}(\theta)\le F_{\mathrm{tr},n}(\theta)
    =F_{\mathrm G,n}(\theta)+o(1)\) uniformly on \(C\), as proved in
  \cref{lem:LAN-fidelity-transfer}.  The LAN maps
  need not be inverses, and \(\Gamma_n\) need not be Gaussian or converge.}
\label{fig:LAN-transfer}
\end{figure}

\subsection{Local parametrization and Gaussian amplification}
\label{sec:gaussian-amplification}

Fix a base spectrum \(p\in\Delta_d^\circ\).  Let
\[
  \mathsf H=\left\{h\in\R^d:\sum_i h_i=0\right\},
  \qquad
  Y_p(z)=\sum_{j<k}
    \frac{z_{jk}E_{kj}-\bar z_{jk}E_{jk}}{\sqrt{p_j-p_k}},
  \quad z\in\C^{r_d},\quad r_d=\binom d2,
\]
where \(E_{jk}\) are the standard matrix units, and define
\[
  g_{n,z}^{(p)}=\exp\!\left(\frac{Y_p(z)}{\sqrt n}\right).
\]
In \(\rho_{p+h/\sqrt n,g_{n,z}^{(p)}}\), the coordinates \(h\) and \(z\)
describe local changes of eigenvalues and eigenbasis, respectively.
The denominators in \(Y_p\) use the fixed base spectrum \(p\), even when the
local spectrum is \(p+h/\sqrt n\).  For \(L>0\), set
\[
  \Omega_L
  =\left\{(h,z)\in\mathsf H\times\C^{r_d}:
  \|h\|_\infty\le L,\
  |\Re z_{ij}|,|\Im z_{ij}|\le L\ \text{for all }i<j\right\}.
\]
For fixed \(p,L\), \(p+h/\sqrt n\in\Delta_d^\circ\) throughout
\(\Omega_L\) for large \(n\).  Under LAN, \(h\) becomes a classical Gaussian
shift and \(z\) an oscillator displacement.

Let \(\cF\) be the one-mode Fock space with number basis
\(\{\ket{k}:k\ge0\}\), creation and annihilation operators \(a^*,a\), and
number operator \(\widehat N=a^*a\).  We use the Weyl convention
\begin{equation}\label{eq:Weyl-convention}
  D(z)=\exp(z a^*-\bar z a).
\end{equation}
For \(0\le q<1\), define the number-diagonal thermal state
\[
  \tau_q=(1-q)\sum_{k\ge0}q^k\proj{k}.
\]
Since these states commute,
\begin{equation}\label{eq:thermal-root-fidelity}
  F(\tau_q,\tau_{q'})
  =\frac{\sqrt{(1-q)(1-q')}}{1-\sqrt{qq'}}.
\end{equation}
For \(G\ge1\), let \(q^{(G)}=1-(1-q)/G\).  The quantum-limited
phase-insensitive amplifier \(\cA_G\) satisfies
\begin{equation}\label{eq:amplifier-action}
  \cA_G\!\left(D(z)\tau_qD(z)^*\right)
  =D(\sqrt G\,z)\tau_{q^{(G)}}D(\sqrt G\,z)^*.
\end{equation}
Equivalently, it is implemented by a two-mode squeezer with a vacuum idler.

On \(\cF_{\mathrm{orb}}=\cF^{\ot r_d}\), let
\(q_{ij}^{(\gamma)}=1-(1-q_{ij})/\gamma\) and define
\begin{equation}\label{eq:sector-thermal-products}
  \tau_p=\bigotimes_{i<j}\tau_{q_{ij}},
  \qquad
  \tau_p^{(\gamma)}=\bigotimes_{i<j}\tau_{q_{ij}^{(\gamma)}}.
\end{equation}
Let \(D(z)=\bigotimes_{i<j}D(z_{ij})\).  The three relevant states are
\begin{equation}\label{eq:gaussian-amplification-model}
\begin{aligned}
  \text{input:}\quad
  &\Phi_z^{(p)}=D(z)\tau_pD(z)^*,\\
  \text{ideal gain-\(\sqrt\gamma\) target:}\quad
  &\Psi_z^{(p,\gamma)}=D(\sqrt\gamma z)\tau_pD(\sqrt\gamma z)^*,\\
  \text{amplifier output:}\quad
  &\cA_\gamma^{\ot r_d}(\Phi_z^{(p)})
    =D(\sqrt\gamma z)\tau_p^{(\gamma)}D(\sqrt\gamma z)^*.
\end{aligned}
\end{equation}
By \eqref{eq:thermal-root-fidelity},
\begin{equation}\label{eq:one-mode-f}
  F(\tau_{q^{(\gamma)}},\tau_q)=f_\gamma(q).
\end{equation}
Multiplicativity then gives, for every \(z\in\C^{r_d}\),
\begin{equation}\label{eq:product-amplifier-fidelity}
  F\!\left(\cA_\gamma^{\ot r_d}(\Phi_z^{(p)}),
           \Psi_z^{(p,\gamma)}\right)
  =\Forb(\gamma,p).
\end{equation}

Equip \(\mathsf H\) with its induced Euclidean measure and recall
\(\Sigma_p\) from \eqref{eq:young-covariance}.  Write
\(\mathsf N_{h,\Sigma_p}\) for the Gaussian law on \(\mathsf H\) with mean
\(h\) and covariance \(\Sigma_p\).  A classical--quantum state is a
measurable field \(R(x)\) of positive trace-class operators on
\(\cF_{\mathrm{orb}}\), with \(\int_{\mathsf H}\Tr R(x)\,\dd x=1\).
A classical--quantum channel \(\Lambda\) is a
completely positive map, in the sense specified in
Appendix~\ref{app:gaussian-converse}, that preserves
\(\int\Tr R(x)\,\dd x\).  For two such states \(R,S\),
\[
  F(R,S)=\int_{\mathsf H}F(R(x),S(x))\,\dd x.
\]
This is the continuous direct-sum identity.  With
\(\|R\|_1:=\int_{\mathsf H}\|R(x)\|_1\,\dd x\), fidelity retains its
monotonicity and joint-concavity properties.

Combining the classical spectral fluctuations with the quantum oscillator
modes, define the input and ideal target states of the full Gaussian model by
\begin{equation}\label{eq:full-gaussian-model}
  \Theta_{h,z}^{\mathrm{in}}
  =\mathsf N_{h,\Sigma_p}\ot\Phi_z^{(p)},
  \qquad
  \Theta_{h,z}^{\mathrm{tar}}
  =\mathsf N_{\sqrt\gamma h,\Sigma_p}\ot\Psi_z^{(p,\gamma)}.
\end{equation}
The target scales both local parameters \((h,z)\) by
\(\sqrt\gamma\), while retaining the input covariance and thermal noise.
For the orbital model, let \(\mathsf Z_L\) be the box of oscillator
displacements over which we average the fidelity, with Lebesgue measure
on its real and imaginary coordinates:
\[
  \mathsf Z_L
  =\{z\in\C^{r_d}:|\Re z_{ij}|,|\Im z_{ij}|\le L\ \text{for all }i<j\}.
\]

\begin{theorem}[Gaussian amplification optimum]
\label{thm:gaussian-amplification}
For every \(p\in\Delta_d^\circ\),
\begin{equation}\label{eq:gaussian-orbital-optimum}
  \lim_{L\to\infty}\sup_\Lambda
  \frac1{|\mathsf Z_L|}\int_{\mathsf Z_L}
  F\!\left(\Lambda(\Phi_z^{(p)}),\Psi_z^{(p,\gamma)}\right)\,\dd z
  =\Forb(\gamma,p),
\end{equation}
and
\begin{equation}\label{eq:gaussian-full-optimum}
  \lim_{L\to\infty}\sup_\Lambda
  \frac1{|\Omega_L|}\int_{\Omega_L}
  F\!\left(\Lambda(\Theta_{h,z}^{\mathrm{in}}),
            \Theta_{h,z}^{\mathrm{tar}}\right)\,\dd h\,\dd z
  =\Funiv(\gamma,p).
\end{equation}
The first supremum is over all quantum channels on
\(\cF_{\mathrm{orb}}\). The second is over all classical--quantum channels.
No Gaussianity, covariance, or product-structure assumption is imposed on
the competing channels in either supremum.
\end{theorem}

The product amplifier \(\cA_\gamma^{\ot r_d}\) attains
\(\Forb(\gamma,p)\) in the orbital model.  Tensoring it with the classical
pushforward \(x\mapsto\sqrt\gamma x\) supplies the Gaussian affinity
\(\Fcl(\gamma,d)\), attaining \(\Funiv(\gamma,p)\) in the full model.

For the converse, average a near-optimal channel over translates of the
expanding flat-prior boxes and use a weighted-fidelity witness to reduce
the problem to an oscillator moment.  The weighted oscillator maps have a
subsequential limit against compact observables that is completely positive
and displacement covariant.  Its output trace on normalized inputs is at
most \(1\) in the orbital model and, by a Gaussian translation estimate,
\(\Fcl(\gamma,d)\) in the full model.  The least-noise inequality then
gives the respective bounds \(\Forb(\gamma,p)\) and \(\Funiv(\gamma,p)\).
See \cref{prop:flat-prior-converse} in Appendix~\ref{app:gaussian-converse}.

\subsection{Local asymptotic normality}
\label{sec:local-asymptotic-normality}

Here LAN means that, uniformly for \((h,z)\) in compact sets,
the state families \(\rho_{p+h/\sqrt n,g_{n,z}^{(p)}}^{\ot n}\) and
\(\mathsf N_{h,\Sigma_p}\ot\Phi_z^{(p)}\) are asymptotically interconvertible
by channels with vanishing trace-norm error.  The Gaussian and oscillator
factors represent spectral and eigenbasis fluctuations, respectively. Cloning
scales \((h,z)\) by \(\sqrt\gamma\).

We fix the oscillator normalization directly from the local tangent and
the collective fluctuation operators.  Let \(\Delta_{ij}=p_i-p_j\).  The
chart defined above satisfies
\begin{equation}\label{eq:LAN-coordinate-map}
  \frac{[Y_p(z),\rho_p]_{ji}}{\sqrt{\Delta_{ij}}}=z_{ij},
  \qquad i<j.
\end{equation}
For \(a_{ij,n}:=(n\Delta_{ij})^{-1/2}\sum_{k=1}^n E_{ij}^{(k)}\), where
\(E_{ij}^{(k)}\) acts on the \(k\)-th tensor factor, the expectations read
\[
  \left\langle[a_{ij,n},a_{ij,n}^*]\right\rangle_{\rho_p^{\ot n}}=1,
  \qquad
  \left\langle a_{ij,n}^*a_{ij,n}\right\rangle_{\rho_p^{\ot n}}
  =\frac{p_j}{\Delta_{ij}}=\frac{q_{ij}}{1-q_{ij}}.
\]
For the local product state
\(\rho_{p+h/\sqrt n,g_{n,z}^{(p)}}^{\ot n}\),
\eqref{eq:LAN-coordinate-map} gives
\(\langle a_{ij,n}\rangle=z_{ij}+o(1)\), uniformly on compact
\((h,z)\)-sets.  The \((i,j)\) oscillator factor of \(\Phi_z^{(p)}\),
defined in \eqref{eq:gaussian-amplification-model}, is
\(D(z_{ij})\tau_{q_{ij}}D(z_{ij})^*\).  With the convention
\eqref{eq:Weyl-convention}, its mean is \(z_{ij}\), and its occupation
at zero displacement is \(q_{ij}/(1-q_{ij})\).  These match the
finite-dimensional moments above, fixing the displacement and thermal
occupation parameters of the Gaussian model.
The following two-way LAN statement follows from Kahn--Gu\c{t}\u{a}'s
channel construction \cite[Theorem~4.3, Lemma~6.4, and
Section~7.3]{KahnGuta2009LANFiniteDimensional}, in the oscillator
coordinates fixed above and with the target-sample rescaling below.\footnote{We
use the displacement normalization of the constructive Lemma~6.4 and
Section~7.3 of the cited work.  The displayed Eq.~(4.19) in its arXiv
version has a different factor and adjoint placement.  The rotation chart,
fluctuation calculation, and qubit check below fix our convention.}

\begin{samepage}
\begin{theorem}[Finite-dimensional LAN and target-sample scaling]
\label{thm:LAN-full}
Fix \(p\in\Delta_d^\circ\).  There are channels
\[
  T_N^{\mathrm{full}}:
  \cT_1((\C^d)^{\ot N})
  \longrightarrow L^1(\mathsf H;\cT_1(\cF_{\mathrm{orb}})),
  \qquad
  S_N^{\mathrm{full}}:
  L^1(\mathsf H;\cT_1(\cF_{\mathrm{orb}}))
  \longrightarrow\cT_1((\C^d)^{\ot N})
\]
such that, for every compact
\(C\subset\mathsf H\times\C^{r_d}\),
\begin{equation}\label{eq:LAN-full-forward}
  \sup_{(h,z)\in C}
  \left\|
    T_n^{\mathrm{full}}\!\left(\rho_{p+h/\sqrt n,g_{n,z}^{(p)}}^{\ot n}\right)
    -\mathsf N_{h,\Sigma_p}\ot\Phi_z^{(p)}
  \right\|_1\longrightarrow0,
\end{equation}
\begin{equation}\label{eq:LAN-full-reverse}
  \sup_{(h,z)\in C}
  \left\|
    S_n^{\mathrm{full}}\!\left(\mathsf N_{h,\Sigma_p}\ot\Phi_z^{(p)}\right)
    -\rho_{p+h/\sqrt n,g_{n,z}^{(p)}}^{\ot n}
  \right\|_1\longrightarrow0,
\end{equation}
and
\begin{equation}\label{eq:LAN-full-target}
  \sup_{(h,z)\in C}
  \left\|
    T_{m_n}^{\mathrm{full}}\!\left(
      \rho_{p+h/\sqrt n,g_{n,z}^{(p)}}^{\ot m_n}\right)
    -\mathsf N_{\sqrt\gamma h,\Sigma_p}\ot\Psi_z^{(p,\gamma)}
  \right\|_1\longrightarrow0.
\end{equation}
\end{theorem}
\end{samepage}

\paragraph{Chart independence of the converses.}
The converses of \cref{thm:known-optimum} and
\cref{thm:unknown-optimum}\textup{(b)} are chart-independent: an invertible
real-linear reparametrization of \((h,z)\) sends the local boxes to dilates
of bounded polytopes, and the averaging in
\cref{prop:flat-prior-converse} uses only
\(|(\Xi_L+\xi)\mathbin\triangle\Xi_L|/|\Xi_L|\to0\).
The resulting Gaussian optimum is determined by \(\Sigma_p\) and the
\(q_{ij}\).  The explicit displacement scale enters
\cref{thm:pct-comparison} through \(\mathbb E|\beta_{ij}|^2\) in
\(q_{ij}^{\mathrm{PCT}}\), and the fluctuation calculation above fixes it.

\paragraph{Coordinate identification and target scaling.}
For \eqref{eq:LAN-full-forward}--\eqref{eq:LAN-full-reverse}, identify the
spectral coordinate \(u\in\R^{d-1}\) of the cited theorem with
\(h=(u_1,\ldots,u_{d-1},-\sum_{i<d}u_i)\in\mathsf H\).
Its covariance \(V(p)_{ij}=\delta_{ij}p_i-p_ip_j\), \(i,j<d\)
\cite[Eq.~(3.3)]{KahnGuta2009LANFiniteDimensional}, becomes
\(\Sigma_p=\diag(p)-pp^{\mathsf T}\) under this embedding, giving
\(\mathsf N_{h,\Sigma_p}\). The normalization
\eqref{eq:LAN-coordinate-map} identifies the oscillator factor with
\(\Phi_z^{(p)}\).

For \eqref{eq:LAN-full-target}, write \(\gamma_n=m_n/n\) and use
\[
  \rho_{p+h/\sqrt n,g_{n,z}^{(p)}}
  =\rho_{p+\sqrt{\gamma_n}h/\sqrt{m_n},
          g_{m_n,\sqrt{\gamma_n}z}^{(p)}}
\]
with the same base \(p\) at both sample sizes.  Apply
\eqref{eq:LAN-full-forward} at sample size \(m_n\) on a common compact
set containing \((\sqrt{\gamma_n}h,\sqrt{\gamma_n}z)\) for all sufficiently
large \(n\).  Since \(\gamma_n\to\gamma\), uniform trace-norm continuity
of Gaussian translations and Weyl conjugation on that set gives the
claimed target limit.

\begin{remark}[Qubit normalization check]
For \(d=2\), let \(\Delta=p_1-p_2>0\), \(q=p_2/p_1\), and take
\(z=t\in\R\).  The exact qubit fidelity identity gives
\[
  F\!\left(\rho_p^{\ot n},
    \rho_{p,g_{n,t}^{(p)}}^{\ot n}\right)
  =\left(1-\Delta^2\sin^2\frac{t}{\sqrt{n\Delta}}\right)^{n/2}
  \longrightarrow e^{-\Delta t^2/2}.
\]
This equals
\(F(\tau_q,D(t)\tau_qD(t)^*)
=\exp\left[-\frac{(1-q)t^2}{2(1+q)}\right]\), confirming the displacement normalization we use.
\end{remark}

\begin{corollary}[Fixed-spectrum LAN]\label{cor:LAN-orb}
Fix \(p\in\Delta_d^\circ\).  There are channels
\[
  T_N^{\mathrm{orb}}:\cT_1((\C^d)^{\ot N})\to\cT_1(\cF_{\mathrm{orb}}),
  \qquad
  S_N^{\mathrm{orb}}:\cT_1(\cF_{\mathrm{orb}})\to\cT_1((\C^d)^{\ot N})
\]
such that, uniformly for \(z\) in compact sets,
\begin{equation}\label{eq:LAN-known-forward}
  \|T_n^{\mathrm{orb}}(\rho_{p,g_{n,z}^{(p)}}^{\ot n})-\Phi_z^{(p)}\|_1
  \longrightarrow0,
\end{equation}
\begin{equation}\label{eq:LAN-known-reverse}
  \|S_n^{\mathrm{orb}}(\Phi_z^{(p)})-\rho_{p,g_{n,z}^{(p)}}^{\ot n}\|_1
  \longrightarrow0,
\end{equation}
and
\begin{equation}\label{eq:LAN-known-target}
  \|T_{m_n}^{\mathrm{orb}}(\rho_{p,g_{n,z}^{(p)}}^{\ot m_n})
       -\Psi_z^{(p,\gamma)}\|_1
  \longrightarrow0.
\end{equation}
\end{corollary}

\begin{proof}
Restrict \cref{thm:LAN-full} to \(h=0\) and set
\(T_N^{\mathrm{orb}}(X):=\int_{\mathsf H}
T_N^{\mathrm{full}}(X)(x)\,\dd x\) and
\(S_N^{\mathrm{orb}}(\sigma):=S_N^{\mathrm{full}}
(\mathsf N_{0,\Sigma_p}\ot\sigma)\).
Marginalization contracts trace norm, and tensoring with
\(\mathsf N_{0,\Sigma_p}\) preserves it, giving all three assertions.
\end{proof}

\begin{lemma}[LAN fidelity transfer]\label{lem:LAN-fidelity-transfer}
Let \(\{\rho_{n,\theta}\}_{\theta\in C}\) and
\(\{\sigma_{n,\theta}\}_{\theta\in C}\) be input and target models on a
compact parameter set \(C\).  Let \(\{\Phi_\theta\}\) and
\(\{\Psi_\theta\}\) be limiting models.  Suppose channels \(S_n,T_n\)
satisfy
\[
\begin{aligned}
  \delta_n&:=\sup_{\theta\in C}
    \|S_n(\Phi_\theta)-\rho_{n,\theta}\|_1\to0,\\
  \eta_n&:=\sup_{\theta\in C}
    \|T_n(\sigma_{n,\theta})-\Psi_\theta\|_1\to0.
\end{aligned}
\]
For every channel \(\cM_n\), let \(\Gamma_n=T_n\circ\cM_n\circ S_n\).
Then, for every \(\theta\in C\),
\begin{equation}\label{eq:LAN-transfer}
  F(\cM_n(\rho_{n,\theta}),\sigma_{n,\theta})
  \le F(\Gamma_n(\Phi_\theta),\Psi_\theta)
      +\sqrt{\delta_n}+\sqrt{\eta_n}.
\end{equation}
The error is uniform over both \(\theta\in C\) and \(\cM_n\).
\end{lemma}

\begin{proof}
Trace-norm contractivity of \(T_n\circ\cM_n\) gives
\(\|\Gamma_n(\Phi_\theta)-T_n\cM_n(\rho_{n,\theta})\|_1\le\delta_n\).
Monotonicity under \(T_n\), followed by
\eqref{eq:fidelity-continuity} and the definition of \(\eta_n\), therefore yields
\[
\begin{aligned}
  F(\cM_n(\rho_{n,\theta}),\sigma_{n,\theta})
  \le F(T_n\cM_n(\rho_{n,\theta}),T_n(\sigma_{n,\theta}))
  \le F(\Gamma_n(\Phi_\theta),\Psi_\theta)
       +\sqrt{\delta_n}+\sqrt{\eta_n}.
\end{aligned}
\]
\end{proof}

\subsection{Minimax converses}
\label{sec:minimax-completion}

\begin{proof}[Known-spectrum converse and completion of \cref{thm:known-optimum}]
For the known-spectrum problem, set
\(u_n(p):=\sup_{\cM_n}\inf_{g\in\U(d)}
F\!\bigl(\cM_n(\rho_{p,g}^{\ot n}),\rho_{p,g}^{\ot m_n}\bigr)\).
For every channel \(\cM_n\) and fixed \(L\), set
\[
  \Gamma_n=T_{m_n}^{\mathrm{orb}}\circ\cM_n\circ S_n^{\mathrm{orb}}.
\]
The local orbit \(\{g_{n,z}^{(p)}:z\in\mathsf Z_L\}\) is contained in the full
unitary orbit.  Apply \cref{lem:LAN-fidelity-transfer}, average over
\(\mathsf Z_L\), and take the supremum over \(\cM_n\).  Since the transfer
error is independent of \(\cM_n\),
\[
  u_n(p)
  \le \sup_\Lambda\frac1{|\mathsf Z_L|}
  \int_{\mathsf Z_L}
  F(\Lambda(\Phi_z^{(p)}),\Psi_z^{(p,\gamma)})\,\dd z+o(1),
\]
where the supremum is over all channels on the Gaussian model and the
error tends to zero for fixed \(L\).  Taking \(\limsup_n\), then
\(L\to\infty\), \eqref{eq:gaussian-orbital-optimum} gives
\[
  \limsup_{n\to\infty} u_n(p)\le\Forb(\gamma,p).
\]
Together with
\(\liminf_n u_n(p)\ge\Forb(\gamma,p)\), proved in
\cref{sec:achievability}, this proves \cref{thm:known-optimum}.
\end{proof}

Fixing \(p\) would miss the cost of an unknown spectrum.  Instead, we use
\(n^{-1/2}\)-neighborhoods of an interior spectrum and let the local window
grow after the large-\(n\) limit.

\begin{proof}[Proof of part~\textup{(b)} of \cref{thm:unknown-optimum}]
Let \(w_n(K)\) denote the quantity inside the limit on the left of
\eqref{eq:compact-global-value}.  Part~\textup{(a)} gives
\(\liminf_nw_n(K)\ge\inf_{p\in K}\Funiv(\gamma,p)\).
For the reverse inequality, fix
\(p_0\in\operatorname{int}_{\mathsf A_d}K\) and \(L\).  For large \(n\),
all local spectra \(p_0+h/\sqrt n\), \((h,z)\in\Omega_L\), lie in \(K\), so
the worst-case fidelity of \(\cM_n\) on \(K\) is at most its average over
this window.  Repeating the known-spectrum argument with the full LAN maps
at \(p_0\), \(\Gamma_n=T_{m_n}^{\mathrm{full}}\circ\cM_n\circ
S_n^{\mathrm{full}}\), and \(\Omega_L\) in place of \(\mathsf Z_L\) gives
\[
  w_n(K)
  \le\sup_\Lambda\frac1{|\Omega_L|}\int_{\Omega_L}
  F\!\left(\Lambda(\Theta_{h,z}^{\mathrm{in}}),
                 \Theta_{h,z}^{\mathrm{tar}}\right)\,\dd h\,\dd z+o(1),
\]
where both Gaussian models are based at \(p_0\) and the supremum is over
all classical--quantum channels.  Taking \(\limsup_n\) and then
\(L\to\infty\), \eqref{eq:gaussian-full-optimum} yields
\(\limsup_nw_n(K)\le\Funiv(\gamma,p_0)\).  Since
\(\operatorname{int}_{\mathsf A_d}K\) is dense in \(K\) and
\(\Funiv(\gamma,\cdot)\) is continuous, taking the infimum over \(p_0\)
completes the proof.
\end{proof}

\section{Cloning projector states}
\label{sec:grassmann}

A flat rank-\(r\) projector state has the form
\(\rho_P=P/r\), where \(P\) is a rank-\(r\) orthogonal projector.  These
projectors form the complex Grassmannian \(\Gr(r,d)\).  Rank deficiency,
and for \(r>1\) the degenerate positive eigenvalue, place these states
outside the simple full-rank theorems.
We construct \(\cC_n^{\mathrm{proj}}\) and prove
\cref{thm:grassmann} directly, without a general rank-deficient LAN theorem.
The compression lemma below holds for any positive spectrum on its
support, but does not by itself establish optimal cloning beyond the flat
case.

\subsection{Exact rank compression}
\label{sec:compression}

A rank-\(r\) input occupies only part of each Schur sector in ambient
dimension \(d\), while our Cartan channel acts on the full sector.  To evaluate
fidelity with a target on the same support, we compare this ambient
channel with the Cartan channel in dimension \(r\).  The compression
identity below separates the mass retained on the target support from
the fidelity within it. For flat projector states, the latter is one.

Superscripts \((s)\) indicate ambient dimension.  For
\(\ell(\lambda)\le r\), write \(D_s^\lambda=\dim\cH_\lambda^{(s)}\) and
define
\begin{equation}
  R_\lambda:=\frac{D_d^\lambda}{D_r^\lambda}
  =\prod_{\substack{1\le i\le r\\r<j\le d}}
    \frac{\lambda_i+j-i}{j-i}.
  \label{eq:R-lambda}
\end{equation}
Let \(J_\lambda:\cH_\lambda^{(r)}\to\cH_\lambda^{(d)}\) be the isometric
\(\U(r)\)-intertwiner mapping unit highest-weight vectors to each other,
and let \(P_0\) project onto the first \(r\) coordinates.  Write
\(\Pi_{\lambda,P_0}=J_\lambda J_\lambda^*\) for the corresponding
Schur-sector support projection.

For a compatible pair of \(r\)-row partitions \(\mu\vdash n\) and
\(\nu\vdash m\), let \(\omega=\nu-\mu\).  Write
\(\cC_{\mu,\nu}^{(s)}\) for the Cartan-sector channel
\eqref{eq:sector-channel-formula} in ambient dimension \(s\in\{r,d\}\).
For a probability vector \(x\in\R^s\), use the dimension-\(s\)
polynomial representation \(\pi_\lambda^{(s)}\) and Schur polynomial
\(s_\lambda^{(s)}\) to define
\[
  \rho_{x,\lambda}^{(s)}
  =\frac{\pi_\lambda^{(s)}(\diag(x))}{s_\lambda^{(s)}(x)}
\]
when \(\ell(\lambda)\le|\{i:x_i>0\}|\), exactly the condition
\(s_\lambda^{(s)}(x)>0\).

\begin{lemma}[Exact rank-compression identity]
\label{lem:exact-compression}
Let \(p=(p_1,\ldots,p_r)\) satisfy \(p_i>0\) and \(\sum_i p_i=1\),
and let \(\bar p=(p_1,\ldots,p_r,0,\ldots,0)\in\R^d\).  For every
\(r\)-row partition \(\lambda\),
\begin{equation}
  \rho_{\bar p,\lambda}^{(d)}
  =J_\lambda\rho_{p,\lambda}^{(r)}J_\lambda^*.
  \label{eq:conditional-support}
\end{equation}
For every compatible \(r\)-row pair \(\mu,\nu\) and every trace-class
operator \(X\) on \(\cH_\mu^{(r)}\),
\begin{equation}
  \Pi_{\nu,P_0}\cC_{\mu,\nu}^{(d)}(J_\mu XJ_\mu^*)\Pi_{\nu,P_0}
  =\frac{R_\mu}{R_\nu}
   J_\nu\cC_{\mu,\nu}^{(r)}(X)J_\nu^*.
  \label{eq:channel-compression}
\end{equation}
Consequently,
\begin{equation}\label{eq:exact-sector-factorization}
 F\!\left(\cC_{\mu,\nu}^{(d)}(\rho_{\bar p,\mu}^{(d)}),
   \rho_{\bar p,\nu}^{(d)}\right)
 =\sqrt{\frac{R_\mu}{R_\nu}}\,
 F\!\left(\cC_{\mu,\nu}^{(r)}(\rho_{p,\mu}^{(r)}),
   \rho_{p,\nu}^{(r)}\right).
\end{equation}
\end{lemma}

\begin{proof}
Let \(\iota:\C^r\hookrightarrow\C^d\) be the coordinate inclusion.
Schur--Weyl naturality identifies \(\iota^{\ot N}\) with
\(\bigoplus_{\ell(\lambda)\le r}(J_\lambda\ot I_{\cS_\lambda})\).
Comparing blocks in
\((\diag\bar p)^{\ot N}
 =\iota^{\ot N}(\diag p)^{\ot N}(\iota^*)^{\ot N}\) gives
\[
  \pi_\lambda^{(d)}
    (\diag\bar p)
  =J_\lambda\pi_\lambda^{(r)}(\diag p)J_\lambda^*.
\]
Taking traces gives
\(s_\lambda^{(d)}(\bar p)=s_\lambda^{(r)}(p)>0\). Division proves
\eqref{eq:conditional-support}.

The maps \(V_{\mu,\omega}^{(d)}J_\nu\) and
\((J_\mu\ot J_\omega)V_{\mu,\omega}^{(r)}\) are \(\U(r)\)-intertwiners
that both send the unit highest-weight vector to
\(v_\mu\ot v_\omega\).  Hence they agree:
\begin{equation}
  V_{\mu,\omega}^{(d)}J_\nu
  =(J_\mu\ot J_\omega)V_{\mu,\omega}^{(r)}.
  \label{eq:cartan-restriction}
\end{equation}
Inserting \eqref{eq:cartan-restriction} into the two sector-channel formulas
proves \eqref{eq:channel-compression}.

For \(A,B\ge0\) and an orthogonal projection \(\Pi\) with
\(\Pi B\Pi=B\), \(F(A,B)=F(\Pi A\Pi,B)\).  Apply this with
\(\Pi=\Pi_{\nu,P_0}\). \eqref{eq:conditional-support},
\eqref{eq:channel-compression}, isometric invariance, and
\(F(aA,B)=\sqrt a F(A,B)\) give \eqref{eq:exact-sector-factorization}.
\end{proof}

For \(P_0\), the flat-spectrum conditional state in the
\(\lambda\)-sector is
\begin{equation}\label{eq:projector-sector-state}
  \rho_{P_0,\lambda}
  =J_\lambda\frac{I_{\cH_\lambda^{(r)}}}{D_r^\lambda}J_\lambda^*.
\end{equation}
By \(\U(r)\)-covariance and irreducibility, the \(r\)-dimensional Cartan
channel maps the maximally mixed input to the maximally mixed output.
Thus \eqref{eq:exact-sector-factorization} gives, for every compatible pair,
\begin{equation}
  F\!\left(
    \cC_{\mu,\nu}^{(d)}(\rho_{P_0,\mu}),
    \rho_{P_0,\nu}
  \right)
  =\sqrt{\frac{R_\mu}{R_\nu}}.
  \label{eq:grassmann-exact-sector}
\end{equation}
For \(\mu_i/n\to1/r\) and \(\nu_i/m_n\to1/r\), this tends to the
right-hand side of \eqref{eq:grassmann-value}.

\subsection{Achievability}
\label{sec:grassmann-achievability}

We follow the Schur--Cartan lifting of
\cref{sec:cartan-sector-lifting}: sample an output Young label and apply
the Cartan-sector channel when the input and output labels are compatible.
To choose the transition kernel, we couple the exact flat-spectrum Schur
laws at sizes \(n\) and \(m_n\) so that their centered diagram shapes are
close, then condition on the input label.  This gives the target label law
exactly and makes compatible pairs overwhelmingly likely.

Let
\begin{equation}
  \mathsf P_N^{(r)}(\lambda)
  =\frac{\dim\cS_\lambda\,D_r^\lambda}{r^N},
  \qquad \lambda\vdash N,\quad\ell(\lambda)\le r,
  \label{eq:flat-schur-law}
\end{equation}
be the Schur--Weyl law of the maximally mixed state on \(\C^r\).  The
complete Schur decomposition is
\begin{equation}
  (P/r)^{\ot N}
  =\bigoplus_{\ell(\lambda)\le r}
   \mathsf P_N^{(r)}(\lambda)
   \frac{\Pi_{\lambda,P}}{D_r^\lambda}
   \ot\frac{I_{\cS_\lambda}}{\dim\cS_\lambda}.
  \label{eq:flat-schur-decomposition}
\end{equation}
Here \(\Pi_{\lambda,P}=\pi_\lambda^{(d)}(P)\) is the rank-\(D_r^\lambda\)
projection obtained by applying the polynomial representation to \(P\).
For \(P=P_0\), it is the support projection defined above.
Sample \(\lambda\sim\mathsf P_N^{(r)}\) and pad it with zero rows to length
\(r\).  The Schur--Weyl/random-word central limit theorem
\cite{Kuperberg2002,Johansson2001} gives
\begin{equation}
  X_N:=\frac{\lambda-(N/r)\mathbf1_r}{\sqrt N}
  \longrightarrow X
  \quad\text{in distribution},
  \label{eq:traceless-GUE-limit}
\end{equation}
where \(X\) is a fixed scaling of the ordered eigenvalue vector of a
traceless \(r\times r\) GUE matrix.  In particular,
\(\Pr(X_1>\cdots>X_r)=1\).

Both \(X_n\) and \(X_{m_n}\) converge to the same law.  By Strassen's
coupling theorem \cite{Strassen1965}, there is a sufficiently slow
sequence \(\varepsilon_n\downarrow0\) for which the minimum, over
couplings of the exact marginals \(\mathsf P_n^{(r)}\) and
\(\mathsf P_{m_n}^{(r)}\), of the probability that their centered
shapes differ by more than \(\varepsilon_n\) tends to zero.  Fix such a
sequence and let \(\mathsf J_n(\mu,\nu)\) be a minimizer of this finite
transportation linear program.  Under \(\mathsf J_n\),
\begin{equation}
  \left\|
    \frac{\mu-(n/r)\mathbf1_r}{\sqrt n}
    -\frac{\nu-(m_n/r)\mathbf1_r}{\sqrt{m_n}}
  \right\|_\infty
  \longrightarrow0
  \quad\text{in probability}.
  \label{eq:shape-coupling}
\end{equation}
This coupling is needed when \(r>1\): independent sampling of the two
labels leaves their \(O(\sqrt n)\) row gaps unrelated.  For
\(1\le i<r\), independent sampling would give
\[
  \frac{(\nu_i-\mu_i)-(\nu_{i+1}-\mu_{i+1})}{\sqrt n}
  \longrightarrow
  \sqrt\gamma(X'_i-X'_{i+1})-(X_i-X_{i+1})
  \quad\text{in distribution},
\]
where \(X'\) is an independent copy of \(X\). This limit is negative
with positive probability.  The chosen minimizer need not have a
closed-form rule.
Every partition \(\mu\vdash n\) with \(\ell(\mu)\le r\) has
\(\mathsf P_n^{(r)}(\mu)>0\).  We therefore condition the coupling on
its input coordinate and define
\begin{equation}\label{eq:projector-label-kernel}
  \mathsf K_n^{\mathrm{proj}}(\nu\mid\mu)
  :=\frac{\mathsf J_n(\mu,\nu)}{\mathsf P_n^{(r)}(\mu)},
  \qquad \ell(\mu)\le r.
\end{equation}
Thus sampling \(\mu\sim\mathsf P_n^{(r)}\) and then
\(\nu\sim \mathsf K_n^{\mathrm{proj}}(\cdot\mid\mu)\) has joint law \(\mathsf J_n\).  In particular,
the output label has exactly the target marginal:
\begin{equation}\label{eq:projector-exact-output-marginal}
  \sum_{\ell(\mu)\le r}
    \mathsf P_n^{(r)}(\mu)\mathsf K_n^{\mathrm{proj}}(\nu\mid\mu)
  =\sum_{\ell(\mu)\le r}\mathsf J_n(\mu,\nu)
  =\mathsf P_{m_n}^{(r)}(\nu).
\end{equation}
For input diagrams with more than \(r\) rows, choose
\(\mathsf K_n^{\mathrm{proj}}(\cdot\mid\mu)\) arbitrarily.  Such diagrams have zero probability
on rank-\(r\) projector inputs, but this extension makes the lifted map a
channel on the entire input space.  Set
\begin{equation}\label{eq:projector-cloner}
  \cC_n^{\mathrm{proj}}:=\mathcal Q_n[\mathsf K_n^{\mathrm{proj}}].
\end{equation}
Its dependence on the fixed rank \(r\) is suppressed in the notation.
The channel depends on the chosen minimizer, but
is independent of the unknown projector \(P\).

Let \(\omega=\nu-\mu\).  Its leading term is \((m_n-n)/r\) in every
coordinate, while its fluctuations are of order \(\sqrt n\).  Hence
\(\omega_i>0\) for every \(i\) with probability \(1-o(1)\).  For adjacent rows,
\[
  \frac{\omega_i-\omega_{i+1}}{\sqrt n}
  =\sqrt{\gamma_n}
    \left[\frac{\nu_i-\nu_{i+1}}{\sqrt{m_n}}\right]
    -\left[\frac{\mu_i-\mu_{i+1}}{\sqrt n}\right]
  \longrightarrow
    (\sqrt\gamma-1)(X_i-X_{i+1})
  \quad\text{in distribution}.
\]
The limit is strictly positive almost surely.  Thus the compatibility event
\[
  C_n=\{(\mu,\nu):\nu-\mu\vdash m_n-n\}
\]
has probability \(1-o(1)\).  The cloner \(\cC_n^{\mathrm{proj}}\) uses the
ambient Cartan-sector channel on \(C_n\) and, otherwise, the covariant
replacement channel into the same sampled \(\nu\)-block.
Write \(\mathbf1_{C_n}\) for the indicator of \(C_n\): it equals one on
compatible pairs and zero on fallback pairs.
Under the coupling \((\mu,\nu)\sim\mathsf J_n\), the central limit theorem also
gives \(\mu_i/n,\nu_i/m_n\to1/r\) in probability.  By
\eqref{eq:R-lambda}, the random variable
\(\mathbf1_{C_n}\sqrt{R_\mu/R_\nu}\) therefore satisfies
\begin{equation}\label{eq:grassmann-sector-factor-limit}
  \mathbf1_{C_n}\sqrt{\frac{R_\mu}{R_\nu}}
  \longrightarrow\gamma^{-r(d-r)/2}
  \quad\text{in probability}.
\end{equation}
On \(C_n\), compatibility gives \(\nu_i\ge\mu_i\), hence
\(0\le\mathbf1_{C_n}\sqrt{R_\mu/R_\nu}\le1\).  This uniform bound
upgrades the convergence in probability in
\eqref{eq:grassmann-sector-factor-limit} to convergence of expectations.
Since the output-label marginal is exact, the block identity
\eqref{eq:lifted-exact-block-fidelity} (valid also for \(r\)-row label
laws), monotonicity and concavity of \(F\) in its first argument, and
\eqref{eq:grassmann-exact-sector} give
\begin{equation}\label{eq:grassmann-achievability-limit}
  F\!\left(
    \cC_n^{\mathrm{proj}}((P/r)^{\ot n}),
    (P/r)^{\ot m_n}
  \right)
  \ge
  \mathbb E_{\mathsf J_n}\!\left[
    \mathbf1_{C_n}\sqrt{\frac{R_\mu}{R_\nu}}\right]
  \longrightarrow\gamma^{-r(d-r)/2}.
\end{equation}
This proves achievability in \cref{thm:grassmann}.

\subsection{Converse}

The unknown projector \(P\) selects an \(r\)-dimensional support inside
\(\C^d\).  Averaging a cloner over unitaries and output permutations
cannot worsen its worst-case fidelity, so we may work with a symmetric
channel.  Replacing each input support projection by the identity gives
one \(P\)-independent operator \(B_n\) that dominates every projector
input. In the \(\mu\)-block, this multiplies its trace by \(R_\mu\).
The symmetric output on \(B_n\) is maximally mixed within each Schur
block, while the target support occupies only a fraction \(1/R_\nu\) of
the \(\nu\)-block.  Cauchy--Schwarz combines these factors into the bound
below.

\begin{proposition}[Projector converse]
\label{prop:grassmann-finite-converse}
For \(1\le r<d\) and every \(n\) with \(m_n\ge n\),
\begin{equation}\label{eq:grassmann-finite-converse}
 \sup_{\cM_n}\inf_{P\in\Gr(r,d)}
 F\!\left(
   \cM_n((P/r)^{\ot n}),
   (P/r)^{\ot m_n}
 \right)
 \le\left[
   \left(\sum_{\substack{\mu\vdash n\\\ell(\mu)\le r}}
     \mathsf P_n^{(r)}(\mu)R_\mu\right)
   \left(\sum_{\substack{\nu\vdash m_n\\\ell(\nu)\le r}}
     \frac{\mathsf P_{m_n}^{(r)}(\nu)}{R_\nu}\right)
 \right]^{1/2}.
\end{equation}
Here \(\mathsf P_N^{(r)}\) is the flat Schur law from
\eqref{eq:flat-schur-law}, and \(R_\lambda\) is the dimension ratio from
\eqref{eq:R-lambda}.  The supremum is over channels from \(n\) to
\(m_n\) qudits.
\end{proposition}

\begin{proof}
Covariantize an arbitrary cloner over \(\U(d)\) and twirl its output over
permutations.  Neither decreases its worst-case fidelity: the first follows
from joint concavity and unitary invariance, the second from monotonicity
under the output twirl, which fixes every target.  Write \(\cM_n\) for the
resulting channel.

Replace each support projection in \eqref{eq:flat-schur-decomposition}
by the identity to obtain the invariant positive operator
\[
  B_n=\bigoplus_{\substack{\mu\vdash n\\\ell(\mu)\le r}}
   \mathsf P_n^{(r)}(\mu)
   \frac{I_{\cH_\mu^{(d)}}}{D_r^\mu}
   \ot\frac{I_{\cS_\mu}}{\dim\cS_\mu}.
\]
For every \(P\),
\[
  (P/r)^{\ot n}\le B_n,
  \qquad
  \Tr B_n=\sum_\mu\mathsf P_n^{(r)}(\mu)R_\mu.
\]
Covariance and the output twirl imply that
\[
  \cM_n(B_n)=\bigoplus_{\substack{\nu\vdash m_n\\\ell(\nu)\le d}}
   z_\nu\frac{I_{\cH_\nu^{(d)}}}{D_d^\nu}
   \ot\frac{I_{\cS_\nu}}{\dim\cS_\nu},
  \qquad z_\nu\ge0,\quad \sum_\nu z_\nu=\Tr B_n.
\]
Fidelity increases in either positive argument:
\(A\le B\) implies
\(\Tr\sqrt{\sqrt T A\sqrt T}\le\Tr\sqrt{\sqrt T B\sqrt T}\).
Positivity of \(\cM_n\) and the direct-sum formula therefore give
\[
 F\!\left(\cM_n((P/r)^{\ot n}),(P/r)^{\ot m_n}\right)
 \le F\!\left(\cM_n(B_n),(P/r)^{\ot m_n}\right)
 =\sum_{\ell(\nu)\le r}
   \sqrt{\frac{z_\nu\mathsf P_{m_n}^{(r)}(\nu)}{R_\nu}}.
\]
By Cauchy--Schwarz and \(\sum_{\ell(\nu)\le r}z_\nu\le\Tr B_n\),
\[
\begin{aligned}
 \left(\sum_{\ell(\nu)\le r}
   \sqrt{\frac{z_\nu\mathsf P_{m_n}^{(r)}(\nu)}{R_\nu}}\right)^2
 \le (\Tr B_n)
   \sum_{\ell(\nu)\le r}\frac{\mathsf P_{m_n}^{(r)}(\nu)}{R_\nu}=\left(\sum_{\ell(\mu)\le r}\mathsf P_n^{(r)}(\mu)R_\mu\right)
   \left(\sum_{\ell(\nu)\le r}
     \frac{\mathsf P_{m_n}^{(r)}(\nu)}{R_\nu}\right).
\end{aligned}
\]
This gives \eqref{eq:grassmann-finite-converse} for the
covariantized channel, hence the minimax bound for the original one.
\end{proof}

For the asymptotics, let \(a=r(d-r)\) and
\(c_{r,d}=\prod_{i\le r<j\le d}[r(j-i)]^{-1}\).  Since
\(\mathsf P_N^{(r)}\) is the Young law of the flat spectrum
\((1/r,\ldots,1/r)\), \cref{lem:young-concentration} in dimension \(r\)
shows that \(\|\lambda/N-\mathbf1_r/r\|_\infty\le N^{-1/3}\) outside an
event of probability at most \(e^{-N^{1/3}/4}\).  Write
\(u_i=\lambda_i-N/r\), so \(\sum_i u_i=0\).  On this typical set,
the leading linear terms in \(u_i\) cancel in the product
\eqref{eq:R-lambda} because \(\sum_i u_i=0\), and Taylor expansion gives
\[
  R_\lambda=c_{r,d}N^a
    \left[1+O_{r,d}\!\left(N^{-1}+\frac{\|u\|_2^2}{N^2}\right)\right],
  \qquad
  \frac1{R_\lambda}=c_{r,d}^{-1}N^{-a}
    \left[1+O_{r,d}\!\left(N^{-1}+\frac{\|u\|_2^2}{N^2}\right)\right].
\]
The flat Young law has \(\mathbb E\|u\|_2^2=O_r(N)\).  Indeed, the
quadratic \(\U(r)\) Casimir has eigenvalue
\(\sum_i\lambda_i(\lambda_i+r+1-2i)\).  On \((\C^r)^{\ot N}\), the
Casimir operator equals
\(NrI+2\sum_{1\le j<k\le N}\operatorname{Swap}_{jk}\), where
\(\operatorname{Swap}_{jk}\) exchanges tensor factors \(j,k\).  Since
\(\Tr[\operatorname{Swap}_{12}(I_r/r)^{\ot 2}]=1/r\), its expectation under
\(\mathsf P_N^{(r)}\) is \(Nr+N(N-1)/r\).  Moreover,
\[
  \sum_i\lambda_i(\lambda_i+r+1-2i)
  =\frac{N^2}{r}+\|u\|_2^2
   +\sum_{i<j}(\lambda_i-\lambda_j),
\]
so \(\mathbb E\|u\|_2^2\le N(r-1/r)\).  Globally
\(1\le R_\lambda\le(N+1)^a\), so the atypical
event contributes less than any inverse power of \(N\).  Hence
\begin{equation}\label{eq:R-moment-limits}
  \mathbb E_{\mathsf P_N^{(r)}}R_\lambda
  =c_{r,d}N^a(1+O_{r,d}(N^{-1})),
  \qquad
  \mathbb E_{\mathsf P_N^{(r)}}\frac1{R_\lambda}
  =c_{r,d}^{-1}N^{-a}(1+O_{r,d}(N^{-1})).
\end{equation}
Applying \eqref{eq:R-moment-limits} at \(N=n\) and \(N=m_n\), the product
under the square root in \eqref{eq:grassmann-finite-converse} is
\[
\begin{aligned}
\mathbb E_{\mathsf P_n^{(r)}}R_\mu\,
\mathbb E_{\mathsf P_{m_n}^{(r)}}R_\nu^{-1}
=n^a m_n^{-a}
  (1+O_{r,d}(n^{-1}))(1+O_{r,d}(m_n^{-1}))
=\gamma_n^{-a}(1+O_{r,d}(n^{-1})).
\end{aligned}
\]
Taking its square root, the right-hand side of
\eqref{eq:grassmann-finite-converse} is
\(\gamma_n^{-a/2}(1+O_{r,d}(n^{-1}))\).  If
\(m_n=\gamma n+O(1)\), this equals
\(\gamma^{-a/2}+O_{r,d,\gamma}(n^{-1})\). The assumption
\(\gamma_n\to\gamma\) alone gives no rate relative to that limit.
Substitution yields
\begin{equation}\label{eq:grassmann-converse-limit}
  \limsup_{n\to\infty}\sup_{\cM_n}\inf_{P\in\Gr(r,d)}
  F\!\left(\cM_n((P/r)^{\ot n}),(P/r)^{\ot m_n}\right)
  \le\gamma^{-r(d-r)/2},
\end{equation}
which completes the converse and the proof of \cref{thm:grassmann}.

\section{Fixed-gain comparison with PCT}\label{sec:pct-comparison}

For the PCT channel defined before \cref{thm:pct-comparison}, the published
data-processing guarantee
\cite{LiTheilHarrowChuang2026,FanizzaGrinkoScharnhorstSpilecki2026,
JeonSohnOh2026} is, for every \(\rho\) of rank at most \(r\),
\begin{equation}\label{eq:pct-data-processing-lower-bound}
  F\!\left(\cC_{n,m}^{\mathrm{PCT},r}(\rho^{\ot n}),
             \rho^{\ot m}\right)
  \ge
  \left(\frac{\binom{n+dr-1}{dr-1}}
              {\binom{m+dr-1}{dr-1}}\right)^{1/2}
  \longrightarrow\gamma^{-(dr-1)/2}
  \qquad(m=m_n).
\end{equation}
The bound is exact for each purified component before averaging and
discarding the purifying registers, but is only a lower bound for the
final output.

Here and in Appendix~\ref{app:pct-comparison}, we analyze
the PCT cloning fidelity directly. \Cref{fig:pct-comparisons} illustrates
the comparisons.  At every fixed gain \(\gamma>1\), our cloners
asymptotically outperform PCT on simple full-rank states and non-pure
flat projector states, while both agree on pure states.

\paragraph{Simple full-rank spectra.}
Fix \(p_1>\cdots>p_d>0\).  PCT uses rank bound \(r=d\) and purification
dimension \(d^2\). Neither it nor \(\cC_n^{\mathrm{univ}}\) is supplied
\(p\).  \Cref{thm:pct-comparison}\textup{(a)} gives PCT's exact limiting
fidelity and its strict gap from \(\Funiv(\gamma,p)\).  The loss for PCT
comes from larger spectral fluctuations and less accurate reproduction
of the eigenbasis.  Appendix~\ref{app:pct-full-rank} gives its exact
fidelity and proves the strict comparison.

\paragraph{Projector states.}
For \(\rho=P/r\) with \(P\in\Gr(r,d)\) and \(1<r<d\), apply PCT
with rank bound \(r\).  \Cref{thm:pct-comparison}\textup{(b)} bounds its
limiting fidelity strictly below the projector optimum
\(\gamma^{-r(d-r)/2}\).  The extra factor is visible even after conditioning
on every output lying in \(\operatorname{ran}P\): counts of outcomes in a basis of
that space fluctuate more under PCT than under the target
\((P/r)^{\ot m_n}\).  Appendix~\ref{app:pct-projector} proves the bound.

\section{Discussion}
\label{sec:discussions}

\subsection{Dependence on spectral mixedness}\label{sec:spectral-mixedness}
For fixed \(\gamma>1\), the product formula shows how fidelity changes
with spectral mixedness:
\[
  f_\gamma'(q)
  =\frac{(\gamma-1)^2}
  {2\sqrt{q(\gamma-1+q)}
   \bigl(\sqrt{\gamma(\gamma-1+q)}+\sqrt q\bigr)^2}>0,
  \qquad 0<q<1.
\]
Consequently, if \(p,p'\in\Delta_d^\circ\) satisfy
\begin{equation}\label{eq:ratio-wise-flattening}
  \frac{p'_j}{p'_i}\ge \frac{p_j}{p_i}
  \qquad\text{for every }i<j,
\end{equation}
that is, if $p'$ is more mixed than $p$ in that every logarithmic spectral gap contracts in passing from \(p\)
to \(p'\), then
\begin{equation}\label{eq:ratio-wise-fidelity-monotonicity}
  \Forb(\gamma,p')\ge \Forb(\gamma,p),
  \qquad
  \Funiv(\gamma,p')\ge \Funiv(\gamma,p),
\end{equation}
with strict inequality if any ratio is strictly larger.\footnote{Condition \eqref{eq:ratio-wise-flattening} implies that \(p'\) is
majorized by \(p\), but majorization alone does not ensure larger
fidelity when \(d\ge3\).  For example,
\(p'=(0.8,0.16,0.04)\) is majorized by \(p=(0.9,0.06,0.04)\), yet
\(\Forb(10,p')\approx0.0772<0.0791\approx\Forb(10,p)\): the decrease
in the \((2,3)\) factor outweighs the increases in the other two.}

\subsection{Worst-case cloning over all states}
The optimal worst-case fidelity over all \(d\)-dimensional states remains
unknown.  At fixed gain \(\gamma>1\), a single channel must work without
spectral or rank information. Write its optimal guaranteed fidelity as
\begin{equation}\label{eq:uniform-state-minimax-value}
  v_n:=\sup_{\cM_n}\inf_{\substack{\rho\ge0\\\Tr\rho=1}}
  F\!\left(\cM_n(\rho^{\ot n}),\rho^{\ot m_n}\right).
\end{equation}
The PCT guarantee gives
\(\liminf_n v_n\ge\gamma^{-(d^2-1)/2}\).  Our
compact-spectrum optimum \eqref{eq:compact-global-value} gives
\begin{equation}\label{eq:uniform-state-minimax-upper-bound}
  \limsup_{n\to\infty}v_n
  \le\inf_{p\in\Delta_d^\circ}\Funiv(\gamma,p)
  =\Fcl(\gamma,d)\,\gamma^{-d(d-1)/4}.
\end{equation}
The equality follows by taking simple spectra with all ratios
\(p_j/p_i\), \(i<j\), tending to zero.
At large gain, these lower and upper bounds scale respectively as
\(\gamma^{-(d^2-1)/2}\) and \(\gamma^{-(d^2-1)/4}\), up to
constants, leaving a factor of two between the polynomial exponents.
Whether PCT attains the asymptotic all-state minimax value also remains
open.

\subsection{High-fidelity sample complexity}
\label{sec:small-error}
In this subsection we follow the sample-complexity literature, using
squared fidelity, \(N\) inputs, and \(M\) \emph{additional} outputs, so that
\begin{equation}\label{eq:pct-copy-count-parameters}
  m=N+M,\qquad \delta=\frac{M}{N},\qquad \gamma=1+\delta.
\end{equation}
The asymptotic family-specific thresholds below use the iterated limit: first
\(N,M\to\infty\) at fixed \(\delta>0\), then \(\delta\downarrow0\).
We have not controlled approximation errors uniformly as
\(\gamma_n\downarrow1\), so these thresholds are not finite-sample bounds
for a joint high-fidelity regime and do not improve the rank-bounded
worst-case \(1/\epsilon\) scaling.  The PCT sufficient condition
\eqref{eq:pct-certified-sample-ratio} is the finite-sample statement here.
The limiting fidelities give two high-fidelity behaviors.  At a fixed
simple full-rank spectrum, the ratio \(N/M\) needed for squared
infidelity \(\epsilon\) scales as \(\epsilon^{-1/2}\), with a smaller
leading constant for our cloner than for PCT.  On non-pure projector
states, both have the leading ratio \(r(d-r)/\epsilon\). Our strict
fidelity advantage first appears at order \(\delta^2\).  The published
PCT guarantee below gives only a looser sufficient ratio.

For \(1\le r\le d\) and sufficiently small \(\epsilon>0\), the concurrent
papers \cite{FanizzaGrinkoScharnhorstSpilecki2026,JeonSohnOh2026} prove that
cloning with squared fidelity at least \(1-\epsilon\) for every state of
rank at most \(r\) has optimal worst-case sample complexity
\begin{equation}\label{eq:pct-uniform-copy-complexity}
  N=\Theta\!\left(\frac{Mrd}{\epsilon}\right).
\end{equation}
Their upper bounds use PCT, and their lower bounds use projectors
of rank \(\ell=\Theta(r)\), possibly \(\ell<r\).

\paragraph{Simple full-rank spectra.}
For fixed \(p\) and \(\bullet\in\{\mathrm{univ},\mathrm{PCT}\}\), expand
the squared limiting fidelity at \(\gamma=1+\delta\) and invert it at
\(1-\epsilon\).  The resulting asymptotic threshold satisfies
\begin{equation}\label{eq:full-rank-small-error}
  F_\bullet(1+\delta,p)^2
  =1-A_\bullet(p)\delta^2+O(\delta^3),
  \qquad
  \frac{N}{M}\sim\sqrt{\frac{A_\bullet(p)}{\epsilon}},
\end{equation}
where
\begin{equation}
  A_{\mathrm{univ}}(p)
  =\frac{d-1}{8}+\frac14\sum_{i<j}\frac{p_i}{p_j},\quad
  A_{\mathrm{PCT}}(p)
  =\frac{d-1}{2}+\frac14\sum_{i<j}\frac{(p_i+p_j)^2}{p_ip_j}.
\end{equation}
For each pair, \((p_i+p_j)^2/(p_ip_j)=p_i/p_j+2+p_j/p_i>p_i/p_j\).
The constant term in \(A_{\mathrm{PCT}}\) is larger as well.  Thus
\(A_{\mathrm{PCT}}(p)>A_{\mathrm{univ}}(p)\).  Since \(\delta=M/N\), the
squared infidelity is asymptotic to
\(A_\bullet(p)(M/N)^2\).  To reach squared infidelity \(\epsilon\), PCT
therefore needs a larger leading ratio \(N/M\) than our cloner.\footnote{The
constant hidden in \(O(\delta^3)\) depends on \(p\), and expanding
\(f_{1+\delta}(q_{ij})\) requires \(\delta\ll q_{ij}\), hence
\(\delta\ll p_d/p_1\) for all pairs.  As the smallest eigenvalue approaches
zero, this window shrinks and both \(A_{\mathrm{univ}}(p)\) and
\(A_{\mathrm{PCT}}(p)\) diverge.  When \(q_{ij}\ll\delta\), instead
\(f_{1+\delta}(q_{ij})^2=(1+\delta)^{-1}+o(\delta)\), so that pair
contributes infidelity of order \(\delta\).  This helps explain why the
fixed-spectrum \(1/\sqrt\epsilon\) scaling does not hold uniformly over
the rank-bounded family, whose worst-case scaling is \(1/\epsilon\).}

\paragraph{Projector states.}
For \(1<r<d\), inverting our limiting squared fidelity gives
\begin{equation}\label{eq:projector-optimal-sample-ratio}
  \frac{N}{M}
  =\frac{1}{(1-\epsilon)^{-1/[r(d-r)]}-1}
  \sim\frac{r(d-r)}{\epsilon}.
\end{equation}
The published guarantee \eqref{eq:pct-data-processing-lower-bound} certifies
PCT squared fidelity at least \(1-\epsilon\) whenever
\begin{equation}\label{eq:pct-certified-sample-ratio}
  \frac{N}{M}
  \ge\frac{1}{(1-\epsilon)^{-1/(rd-1)}-1}
  \sim\frac{rd-1}{\epsilon},
\end{equation}
since, with \(n=N\) and \(m=N+M\), its binomial ratio is
\(\prod_{j=1}^{rd-1}(N+j)/(N+M+j)\ge[N/(N+M)]^{rd-1}\).
The coefficient \(rd-1\) describes only this sufficient condition.
The two thresholds have the subleading expansion
\begin{equation}\label{eq:power-fidelity-threshold-expansion}
  \frac1{(1-\epsilon)^{-1/a}-1}
  =\frac a\epsilon-\frac{a+1}{2}+O(\epsilon),
\end{equation}
with \(a=r(d-r)\) in \eqref{eq:projector-optimal-sample-ratio} and
\(a=rd-1\) in \eqref{eq:pct-certified-sample-ratio}.  The latter expansion
does not determine PCT's actual subleading threshold.
Both the \(\liminf\) and the \(\limsup\) of PCT's actual squared
infidelity equal \(r(d-r)\delta+O_{r,d}(\delta^2)\)
(\cref{prop:pct-projector-small-error}), so PCT shares the leading sample
ratio \(N/M\sim r(d-r)/\epsilon\).  With \(\gamma=1+\delta\), squaring
\eqref{eq:grassmann-pct-upper-bound} bounds the \(\limsup\) of PCT's
squared fidelity by our value \((1+\delta)^{-r(d-r)}\) times
\[
  \left(\frac{\sqrt{1+2\delta}}{1+\delta}\right)^{r-1}
  =1-\frac{r-1}{2}\delta^2+O_r(\delta^3).
\]
Thus the asymptotic squared-fidelity gap is at least
\((r-1)\delta^2/2+O_{r,d}(\delta^3)\), while
\cref{prop:pct-projector-small-error} bounds it above by
\(O_{r,d}(\delta^2)\).  Our advantage appears at second order,
with its exact coefficient undetermined.

\subsection{Conjectures for degenerate and rank-deficient spectra}
\label{sec:conjectures}
Beyond the flat projector family, optimal cloning remains open for
nonflat rank-deficient states and spectra with degeneracies.  We conjecture separate optimal cloning fidelities for two families: on a known-spectrum orbit only the eigenbasis varies, while in
a local fixed-rank family the positive eigenvalues may vary as well.
Suppose a rank-\(r\) state \(\rho\) has distinct positive eigenvalues
\(u_1>\cdots>u_s>0\), with multiplicities \(r_1,\ldots,r_s\), where
\(\sum_a r_a=r\) and \(\sum_a r_a u_a=1\).  We conjecture that the
known-spectrum optimum on its unitary orbit, defined as in
\eqref{eq:known-optimum}, is
\begin{equation}\label{eq:conjectured-degenerate-orbit}
  \Forb(\gamma,\rho)
  =\gamma^{-r(d-r)/2}
   \prod_{a<b}f_\gamma\!\left(\frac{u_b}{u_a}\right)^{r_ar_b}.
\end{equation}
Rotations within degenerate eigenspaces are stabilizers and do not
contribute.  Since
\(f_\gamma(0)=\gamma^{-1/2}\), the first factor accounts for the
\(r(d-r)\) support--kernel directions.  The simple full-rank case
\(r=d\), \(s=d\), \(r_a=1\) recovers \eqref{eq:main-values}. The flat
projector case \(s=1\) recovers \cref{thm:grassmann}, including Werner's
pure-state value at \(r=1\).  The exact compression identity in
\cref{lem:exact-compression} establishes this support--kernel factorization
for our sector channel, but does not by itself prove the conjecture for the
remaining spectra.

The known-spectrum optimum need not be continuous at a degeneracy,
because the orbit loses dimensions.  For qubits, let \(q=p_2/p_1\), so
\(\Forb(\gamma,p)=f_\gamma(q)\).  As \(p\to(1/2,1/2)\) through simple
spectra, \(q\uparrow1\), and
\[
  \lim_{\substack{p\to(1/2,1/2)\\p_1>p_2}}
  \Forb(\gamma,p)
  =\lim_{q\uparrow1}
    \frac{\sqrt\gamma+\sqrt{q(\gamma-1+q)}}{\gamma+q}
  =\frac{2\sqrt\gamma}{1+\gamma}<1,
\]
whereas the known-spectrum family at \(p=(1/2,1/2)\) contains only
\(I_2/2\) and can be cloned perfectly.  For each fixed \(n\), the constant
channel preparing \((I_2/2)^{\ot m_n}\) has fidelity tending to one as
\(p\to(1/2,1/2)\), uniformly over the orbit.  Thus the degeneracy and
large-$n$ limits do not commute.

For the local unknown-state problem, consider the smooth manifold of
rank-\(r\) density matrices, denoted as $ \mathcal D_r^\circ$.
If \(P\) is the support projection of \(\rho\), its tangent space is
\begin{equation}\label{eq:fixed-rank-tangent-space}
  T_\rho\mathcal D_r^\circ
  =\{X=X^*:\Tr X=0,\ (I-P)X(I-P)=0\}.
\end{equation}

Let \(\operatorname{Exp}_\rho\) be the exponential map for the induced
Hilbert--Schmidt metric.  For fixed \(L\) and sufficiently large \(n\), set
\(\rho_{n,X}=\operatorname{Exp}_\rho(X/\sqrt n)\) and define
\begin{equation}\label{eq:fixed-rank-local-minimax}
  v_{n,L}^{(r)}(\rho)
  :=\sup_{\cM_n}\inf_{\substack{X\in T_\rho\mathcal D_r^\circ\\
                                \|X\|_2\le L}}
  F\!\left(\cM_n(\rho_{n,X}^{\ot n}),\rho_{n,X}^{\ot m_n}\right).
\end{equation}

The channel may depend on the fixed center \(\rho\), but not on the
unknown displacement \(X\).  We conjecture that the \(\limsup\) and \(\liminf\) of
\(v_{n,L}^{(r)}(\rho)\), first as \(n\to\infty\) and then as
\(L\to\infty\), agree.  Writing their common value as
\(\Funiv^{(r)}(\gamma,\rho)\), the conjectured formula is
\begin{equation}\label{eq:conjectured-fixed-rank-local}
  \Funiv^{(r)}(\gamma,\rho)
  =\left(\frac{2\sqrt\gamma}{1+\gamma}\right)^{(\sum_{a=1}^s r_a^2-1)/2}
   \Forb(\gamma,\rho).
\end{equation}
Unlike \(v_n\) in \eqref{eq:uniform-state-minimax-value}, this problem
tests only \(n^{-1/2}\)-scale perturbations of a fixed rank-\(r\) center
\(\rho\).  Thus \eqref{eq:conjectured-fixed-rank-local} is a conjectured
local value, whereas \eqref{eq:uniform-state-minimax-upper-bound} is a
proved upper bound for the all-state minimax value.
The classical exponent in \eqref{eq:conjectured-fixed-rank-local} counts
perturbations within the positive-eigenvalue blocks: an
\(r_a\times r_a\) Hermitian block has \(r_a^2\) real parameters, and the
trace constraint removes one overall.  A classical Gaussian shift with
\(k\) real parameters contributes the fidelity factor
\((2\sqrt\gamma/(1+\gamma))^{k/2}\).  At fixed rank, the zero block has no
independent parameters. Support--kernel rotations are already counted in
\(\Forb(\gamma,\rho)\).  The GUE limit for flat-spectrum Young-label
fluctuations \cite{Kuperberg2002,Johansson2001} motivates treating the
within-block directions as classical Gaussian variables, but does not
establish the conjecture.

When the \(r\) positive eigenvalues are distinct (every \(r_a=1\)), the
fixed-rank LAN theorem of \cite[Theorem~3.3]{LahiryNussbaum2024}
identifies \(r-1\) classical Gaussian coordinates, \(\binom r2\) shifted
thermal modes (one per pair of positive eigenvectors), and \(r(d-r)\)
shifted pure modes (one per support--kernel pair).  The classical
coordinates give the exponent \((r-1)/2\) in
\eqref{eq:conjectured-fixed-rank-local}. The thermal and pure modes
correspond to the \(f_\gamma(u_b/u_a)\) factors and
\(\gamma^{-r(d-r)/2}\), respectively, in
\eqref{eq:conjectured-degenerate-orbit}.  These mode counts support the
conjecture but do not prove the optimal cloning fidelities.

The formula is consistent under positive-eigenvalue collisions at fixed
rank: each colliding pair contributes
\(f_\gamma(1)=2\sqrt\gamma/(1+\gamma)\), and
\((r-1)/2+\sum_a\binom{r_a}{2}=(\sum_a r_a^2-1)/2\).
When also \(r=d\) and the spectrum is simple, the formula reduces to the
full-rank value \(\Funiv(\gamma,p)\) in \eqref{eq:main-values}.

At a flat projector state, \(s=1\), the local model still allows the
positive eigenvalues to vary, unlike the projector family in
\cref{thm:grassmann}, where they are fixed at \(1/r\).  This uncertainty
gives the conjectured classical factor for \(r>1\). The two values coincide
for \(r=1\).  Neither model allows rank changes, which remains an open question.

We further conjecture that, at every fixed \(\gamma>1\) and with the same
prior spectral or rank information available to both protocols,
Schur--Cartan cloners with transition rules adapted to degenerate or
rank-deficient spectra outperform PCT on the corresponding families of mixed states, and
that a single spectrum-independent extension strictly beats PCT with rank
bound \(d\) in the all-state problem
\eqref{eq:uniform-state-minimax-value}.  The latter would require control
near the boundary of the spectral simplex, beyond the pointwise
comparisons proved here.

\appendix
\section{Exact-weight PBW calculation of the sector fidelity}
\label[appendix]{app:sector-fidelity}

We prove \cref{prop:sector-fidelity} in three steps.
On fixed-height weight spaces, the Gram matrices of normalized PBW
vectors approach the identity.  The Cartan inclusion then splits
lowering operators with asymptotically binomial weights.  Uniform
character normalization and a bound on the target mass above the cutoff
yield the thermal fidelity \(\Forb(\gamma,p)\).

Throughout, \(O_L(\cdot)\) may depend on \(L,d,c,C\) in
\(cN\le\Delta_\alpha^\lambda\le CN\) for \(\alpha\in\mathsf R_+\), but is
uniform in \(N\) and all dominant \(\lambda\) satisfying these inequalities.

\subsection{PBW vectors on complete weight spaces}

Identify the positive roots of \(\U(d)\) with
\[
  \mathsf R_+=\{e_i-e_j:1\le i<j\le d\}.
\]
For \(\alpha=e_i-e_j\), set
\[
  F_\alpha=E_{ji},
  \qquad E_\alpha=E_{ij},
  \qquad H_\alpha=[E_\alpha,F_\alpha]=E_{ii}-E_{jj},
  \qquad \Delta_\alpha^\lambda=\lambda_i-\lambda_j.
\]
Fix a total order on \(\mathsf R_+\).  For
\(\mathbf r\in\mathbb Z_{\ge0}^{\mathsf R_+}\), write
\[
  F^{\mathbf r}=\prod_{\alpha\in\mathsf R_+}^{\prec}F_\alpha^{r_\alpha},
  \qquad |\mathbf r|=\sum_\alpha r_\alpha.
\]
For a nonnegative simple-root combination \(\eta\), let
\(\operatorname{ht}(\eta)\) be the sum of its coefficients in the simple
roots \(e_i-e_{i+1}\), \(i<d\), and let
\[
  \mathsf P(\eta)
  =\left\{\mathbf r:\sum_{\alpha\in\mathsf R_+}r_\alpha\alpha=\eta\right\}.
\]
The set \(\mathsf P(\eta)\) is finite, and the PBW theorem implies that
\(\{F^{\mathbf r}v_\lambda:\mathbf r\in\mathsf P(\eta)\}\) spans
\(\cH_\lambda[\lambda-\eta]\) \cite{FultonHarris1991}.
Here \(\cH_\lambda[\lambda-\eta]\) denotes the weight space of weight
\(\lambda-\eta\).  Define
\begin{equation}\label{eq:PBW-height-subspace}
  \mathcal K_{\lambda,L}
  =\bigoplus_{\operatorname{ht}(\eta)\le L}
    \cH_\lambda[\lambda-\eta].
\end{equation}
Let \(P_{\lambda,L}\) denote the orthogonal projection onto
\(\mathcal K_{\lambda,L}\).

\begin{lemma}[PBW Gram matrices at fixed height]
\label{lem:PBW-fixed-height}
Let \(\lambda=\lambda^{(N)}\) range over dominant weights satisfying
\[
  cN\le\Delta_\alpha^\lambda\le CN
  \qquad(\alpha\in\mathsf R_+)
\]
for fixed \(0<c<C<\infty\).  Define
\begin{equation}\label{eq:normalized-PBW-vectors}
  \widetilde e_{\mathbf r}^\lambda
  =\prod_{\alpha\in\mathsf R_+}^{\prec}
   \frac{F_\alpha^{r_\alpha}}
        {\sqrt{r_\alpha!}(\Delta_\alpha^\lambda)^{r_\alpha/2}}
   v_\lambda.
\end{equation}
For every fixed \(L\), uniformly over
\(\operatorname{ht}(\eta)\le L\) and
\(\mathbf r,\mathbf s\in\mathsf P(\eta)\),
\begin{equation}\label{eq:PBW-Gram-fixed-height}
  \langle\widetilde e_{\mathbf r}^\lambda,
         \widetilde e_{\mathbf s}^\lambda\rangle
  =\delta_{\mathbf r,\mathbf s}+O_L(N^{-1/2}).
\end{equation}
Consequently, for all sufficiently large \(N\), the full Gram matrix
\(G_{\lambda,\eta}\) on every such exact-weight block is invertible.  If
\begin{equation}\label{eq:exact-block-orthonormalization}
  (e_{\mathbf r}^\lambda)_{\mathbf r\in\mathsf P(\eta)}
  =
  (\widetilde e_{\mathbf r}^\lambda)_{\mathbf r\in\mathsf P(\eta)}
  G_{\lambda,\eta}^{-1/2},
\end{equation}
then these vectors form an orthonormal basis of
\(\cH_\lambda[\lambda-\eta]\), and
\begin{equation}\label{eq:PBW-orthonormal-close}
  e_{\mathbf r}^\lambda
  =\widetilde e_{\mathbf r}^\lambda+O_L(N^{-1/2})
\end{equation}
uniformly for \(\operatorname{ht}(\eta)\le L\) and
\(\mathbf r\in\mathsf P(\eta)\).
\end{lemma}

\begin{proof}
For \(\mathbf r\in\mathsf P(\eta)\), every positive root has positive
simple-root height, so
\(|\mathbf r|\le\operatorname{ht}(\eta)\).  Let
\[
  \widehat F_\alpha=\frac{F_\alpha}{\sqrt{\Delta_\alpha^\lambda}},
  \qquad
  \widehat E_\alpha=\frac{E_\alpha}{\sqrt{\Delta_\alpha^\lambda}}.
\]
The word inside each unnormalized Gram entry
\(\langle v_\lambda,(\widehat F^{\mathbf r})^*
\widehat F^{\mathbf s}v_\lambda\rangle\) contains at most \(2L\)
normalized root operators.  Reordering them does not increase the number
of factors.  Since each root has height at most \(d-1\), every
intermediate vector lies in
\[
  \mathcal K_{\lambda,R_L},
  \qquad R_L:=2L(d-1).
\]

The normalized root operators are uniformly bounded on each fixed
cutoff.  Indeed, let \(h=\operatorname{ht}(\alpha)\), set
\(\mathcal K_{\lambda,R}=\{0\}\) for \(R<0\), and write
\[
  M_R=\bigl\|\widehat F_\alpha|_{\mathcal K_{\lambda,R}}\bigr\|^2.
\]
Because each cutoff \(\mathcal K_{\lambda,R}\)
contains complete weight spaces, \(\widehat E_\alpha\) maps
\(\mathcal K_{\lambda,R}\) into
\(\mathcal K_{\lambda,R-h}\).  It is the adjoint of
\(\widehat F_\alpha:\mathcal K_{\lambda,R-h}\to
\mathcal K_{\lambda,R}\), so its squared operator norm is \(M_{R-h}\).
On \(\mathcal K_{\lambda,R}\),
\([\widehat E_\alpha,\widehat F_\alpha]
=H_\alpha/\Delta_\alpha^\lambda=I+O_R(N^{-1})\).  Applying this identity
to unit vectors gives
\[
  M_R\le M_{R-h}+1+O_R(N^{-1})
  \le \lfloor R/h\rfloor+1+O_R(N^{-1})=O_R(1),
\]
where the second inequality follows by induction from \(M_R=0\) for
\(R<0\).  This also bounds \(\widehat E_\alpha\).  All operator norms below
are taken after restriction to
\(\mathcal K_{\lambda,R_L}\), with the codomain norm inherited from
\(\cH_\lambda\).  For \(0\le R\le R_L=2L(d-1)\),
the \(O_R\) bounds above can be written \(O_L\).

The matrix-unit commutator formula now gives
\begin{equation}\label{eq:normalized-root-commutator}
  [\widehat E_\alpha,\widehat F_\beta]
  =\delta_{\alpha\beta}I+O_L(N^{-1/2}).
\end{equation}
For \(\alpha=\beta\), the left side is
\(H_\alpha/\Delta_\alpha^\lambda=I+O_L(N^{-1})\) on the cutoff.
For \(\alpha\ne\beta\), the unnormalized commutator is zero or,
up to sign, a root operator \(E_\gamma\) or \(F_\gamma\).  In the nonzero
case, the preceding bound gives
\(\|[\widehat E_\alpha,\widehat F_\beta]\|
=O_L\bigl(\sqrt{\Delta_\gamma^\lambda}/
\sqrt{\Delta_\alpha^\lambda\Delta_\beta^\lambda}\bigr)
=O_L(\sqrt N/N)=O_L(N^{-1/2})\).

Commute every \(\widehat E_\alpha\) to the right in
\(\langle v_\lambda,(\widehat F^{\mathbf r})^*
\widehat F^{\mathbf s}v_\lambda\rangle\).  The scalar commutators in
\eqref{eq:normalized-root-commutator} contribute only when
\(\mathbf r=\mathbf s\), in which case they give
\(\prod_\alpha r_\alpha!\).  Every other term vanishes or contains an
\(O_L(N^{-1/2})\) remainder.  Only \(O_L(1)\) terms occur, and all remaining
operator factors have norm \(O_L(1)\). Consequently,
\[
  \langle v_\lambda,(\widehat F^{\mathbf r})^*
  \widehat F^{\mathbf s}v_\lambda\rangle
  =\delta_{\mathbf r,\mathbf s}\prod_\alpha r_\alpha!
   +O_L(N^{-1/2}).
\]
Dividing by the factorials in \eqref{eq:normalized-PBW-vectors} proves
\eqref{eq:PBW-Gram-fixed-height}.

The number and sizes of the exact-weight blocks are bounded in terms of \(L\), so
\(G_{\lambda,\eta}=I+O_L(N^{-1/2})\) in operator norm.  It is therefore
invertible for large \(N\), and the PBW spanning vectors form a basis.
Functional calculus gives
\(G_{\lambda,\eta}^{-1/2}=I+O_L(N^{-1/2})\), proving
\eqref{eq:exact-block-orthonormalization} and
\eqref{eq:PBW-orthonormal-close}.
\end{proof}

\subsection{Character normalization and fixed-height matrices}

For \(a=(a_1,\ldots,a_d)\), write \(p^a=\prod_i p_i^{a_i}\).

\begin{lemma}[Uniform character normalization]
\label{lem:uniform-character-normalization}
Let \(K\Subset\Delta_d^\circ\), let \(\delta_N\to0\), and suppose
\(p\in K\), \(\lambda\in\mathsf Y_N\), and
\(\|\lambda/N-p\|_\infty\le\delta_N\).  Then, uniformly over this family,
\begin{equation}\label{eq:uniform-character-normalization}
  \frac{p^\lambda}{s_\lambda(p)}
  =\prod_{i<j}\left(1-\frac{p_j}{p_i}\right)+o(1).
\end{equation}
\end{lemma}

\begin{proof}
Write the numerator in the Weyl character formula as
\[
  \sum_{\sigma\in\mathfrak S_d}\operatorname{sgn}(\sigma)A_\sigma,
  \qquad
  A_\sigma=\prod_i p_i^{\lambda_{\sigma(i)}+d-\sigma(i)}.
\]
For each nonidentity permutation, strict rearrangement
makes \(\sum_i(p_i-p_{\sigma(i)})\log p_i\) positive.  Compactness of
\(K\) and finiteness of \(\mathfrak S_d\) make this bound uniform for all
sufficiently large \(N\):
\[
  \frac1N\log\frac{A_{\mathrm{id}}}{A_\sigma}
  =\sum_i(p_i-p_{\sigma(i)})\log p_i
   +O_K(\delta_N+N^{-1})\ge c_K>0.
\]
Thus the numerator is \(A_{\mathrm{id}}(1+O_K(e^{-c_KN}))\).
Cancelling \(\prod_i p_i^{d-i}\) against the Weyl denominator gives
\[
  s_\lambda(p)
  =p^\lambda
   \prod_{i<j}\left(1-\frac{p_j}{p_i}\right)^{-1}
   \bigl(1+O_K(e^{-c_KN})\bigr).
\]
This proves \eqref{eq:uniform-character-normalization}.
\end{proof}

Index the number basis of \(\cF_{\mathrm{orb}}\) by
\(\mathbf r\in\mathbb Z_{\ge0}^{\mathsf R_+}\), writing
\(\{\ket{\mathbf r}\}\), and define the finite-rank projection
\begin{equation}\label{eq:Fock-height-projection}
  P_L^{\cF}
  =\sum_{\operatorname{ht}(\sum_\alpha r_\alpha\alpha)\le L}
    \proj{\mathbf r}.
\end{equation}
For every fixed \(L\) and all sufficiently large \(N\), each highest
weight \(\lambda\) satisfying the gap bounds in
\cref{lem:PBW-fixed-height} has the orthonormal basis
\(e_{\mathbf r}^\lambda\) of \(\mathcal K_{\lambda,L}\) constructed there.
Define
\begin{equation}\label{eq:finite-coordinate-unitary}
  J_{\lambda,L}:\mathcal K_{\lambda,L}\longrightarrow
  P_L^{\cF}\cF_{\mathrm{orb}},
  \qquad
  J_{\lambda,L}e_{\mathbf r}^\lambda=\ket{\mathbf r}.
\end{equation}
The common index set makes \(J_{\lambda,L}\) unitary on its stated domain.

\begin{proof}[Proof of \cref{prop:sector-fidelity}]
Fix \(K\) and \(\delta_n\) as in the proposition, write \(m=m_n\), and let
\(\omega=\nu-\mu\).  Since \(m/n=\gamma_n\to\gamma>1\), we have
\(m/n-1\ge(\gamma-1)/2\) for all sufficiently large \(n\).
Uniformly over the allowed parameters,
\[
  \left\|\frac\omega{m-n}-p\right\|_\infty
  \le\frac{m/n+1}{m/n-1}\,\delta_n=o(1).
\]
Every coordinate and every adjacent gap is uniformly positive on \(K\).
Hence \(\omega\vdash m-n\) for all sufficiently large \(n\), and
\(\nu=\mu+\omega\in\mathsf Y_m\).  For
\(\alpha=e_i-e_j\), set
\begin{equation}\label{eq:splitting-fraction}
  s_{\alpha,n}
  :=\frac{\Delta_\alpha^\mu}{\Delta_\alpha^\nu}
    =\frac{n}{m}
      \frac{\mu_i/n-\mu_j/n}{\nu_i/m-\nu_j/m}
    =\frac1\gamma+o(1),
  \qquad
  1-s_{\alpha,n}=\frac{\Delta_\alpha^\omega}{\Delta_\alpha^\nu}.
\end{equation}
Consequently, there are constants
\(0<c_K<C_K<\infty\), depending only on \(K\) and
\(\gamma\), such that, for all large \(n\),
\begin{equation}\label{eq:uniform-sector-gaps}
  c_Kn
  \le\Delta_\alpha^\mu,\Delta_\alpha^\nu,\Delta_\alpha^\omega
  \le C_Kn
  \qquad(\alpha\in\mathsf R_+).
\end{equation}

Let
\(V=V_{\mu,\omega}:\cH_\nu\to\cH_\mu\ot\cH_\omega\), and write
\(F_\alpha^{(\lambda)}=\dd\pi_\lambda(F_\alpha)\) for the action of the root
operator on \(\cH_\lambda\).  Intertwining and the highest-weight
normalization give
\[
  VF_\alpha^{(\nu)}
  =\bigl(F_\alpha^{(\mu)}\ot I+I\ot F_\alpha^{(\omega)}\bigr)V.
\]
Expanding each
\(\bigl(F_\alpha^{(\mu)}\ot I+I\ot F_\alpha^{(\omega)}\bigr)^{r_\alpha}\)
binomially gives the exact identity
\begin{equation}\label{eq:cartan-PBW-exact}
  V\widetilde e_{\mathbf r}^{\nu}
  =\sum_{\mathbf a\le\mathbf r}
    \beta_{\mathbf r,\mathbf a}^{(n)}
    \widetilde e_{\mathbf a}^{\mu}
    \ot\widetilde e_{\mathbf r-\mathbf a}^{\omega},
\end{equation}
where
\begin{equation}\label{eq:cartan-coefficients}
  \beta_{\mathbf r,\mathbf a}^{(n)}
  =\prod_{\alpha\in\mathsf R_+}
   \sqrt{\binom{r_\alpha}{a_\alpha}
   s_{\alpha,n}^{a_\alpha}
   (1-s_{\alpha,n})^{r_\alpha-a_\alpha}}.
\end{equation}
If \(\operatorname{ht}(\sum_\alpha r_\alpha\alpha)\le L\),
then \eqref{eq:cartan-PBW-exact} has \(O_L(1)\) terms, all of height at
most \(L\), with
\(0\le\beta_{\mathbf r,\mathbf a}^{(n)}\le1\).
Since \(n,m,m-n\) are uniformly comparable, the gap bounds
\eqref{eq:uniform-sector-gaps} and \cref{lem:PBW-fixed-height} give
\(e_{\mathbf b}^{\lambda}=\widetilde e_{\mathbf b}^{\lambda}
+O_L(n^{-1/2})\) for \(\lambda\in\{\mu,\nu,\omega\}\) and all relevant
\(\mathbf b\).  Substituting these relations into
\eqref{eq:cartan-PBW-exact} and using \(\|Vx\|=\|x\|\) gives
\begin{equation}\label{eq:cartan-orthonormal-fixed-height}
  Ve_{\mathbf r}^{\nu}
  =\sum_{\mathbf a\le\mathbf r}
    \beta_{\mathbf r,\mathbf a}^{(n)}
    e_{\mathbf a}^{\mu}\ot e_{\mathbf r-\mathbf a}^{\omega}
    +O_L(n^{-1/2})
\end{equation}
uniformly in norm for
\(\operatorname{ht}(\sum_\alpha r_\alpha\alpha)\le L\).

Let
\[
  \sigma_{\mu,\nu}=\cC_{\mu,\nu}(\rho_{p,\mu}),
  \qquad
  c_\lambda(p)=\frac{p^\lambda}{s_\lambda(p)},
  \qquad
  q_\alpha=\frac{p_j}{p_i}\quad(\alpha=e_i-e_j).
\]
For \(\lambda\in\{\mu,\nu\}\), the weight of
\(e_{\mathbf b}^\lambda\) is
\(\lambda-\sum_\alpha b_\alpha\alpha\), and hence
\[
  \rho_{p,\lambda}e_{\mathbf b}^\lambda
  =c_\lambda(p)\prod_\alpha q_\alpha^{b_\alpha}e_{\mathbf b}^\lambda.
\]
The sector channel gives the local formula
\[
  \sigma_{\mu,\nu}
  =\frac{\dim\cH_\mu}{\dim\cH_\nu}
   V^*(\rho_{p,\mu}\ot I_{\cH_\omega})V.
\]
All estimates below are uniform over the allowed family.
By \cref{lem:uniform-character-normalization},
\[
  c_\lambda(p)=\prod_\alpha(1-q_\alpha)+o(1)
  \quad(\lambda=\mu,\nu),
\]
while the Weyl dimension formula, \eqref{eq:uniform-sector-gaps}, and
\eqref{eq:splitting-fraction} give
\[
  \frac{\dim\cH_\mu}{\dim\cH_\nu}
  =\prod_{\alpha=e_i-e_j}
    \frac{\Delta_\alpha^\mu+j-i}{\Delta_\alpha^\nu+j-i}
  =\prod_{\alpha\in\mathsf R_+}
    \bigl(s_{\alpha,n}+O(n^{-1})\bigr)
  =\gamma^{-r_d}+o(1).
\]
Substitute \eqref{eq:cartan-orthonormal-fixed-height} into the sector
formula.  The binomial theorem
and \(\|\rho_{p,\mu}\|\le1\) therefore give, for fixed \(L\),
\begin{align}
  \langle e_{\mathbf r}^\nu,
    \sigma_{\mu,\nu}e_{\mathbf s}^\nu\rangle
  &={}\frac{\dim\cH_\mu}{\dim\cH_\nu}
  \sum_{\substack{\mathbf a\le\mathbf r\\
                   \mathbf b\le\mathbf s}}
  \beta_{\mathbf r,\mathbf a}^{(n)}
  \beta_{\mathbf s,\mathbf b}^{(n)}
  \langle e_{\mathbf a}^\mu,
    \rho_{p,\mu}e_{\mathbf b}^\mu\rangle
  \langle e_{\mathbf r-\mathbf a}^\omega,
    e_{\mathbf s-\mathbf b}^\omega\rangle
  +O_L(n^{-1/2})\notag\\
  &={}
  \delta_{\mathbf r,\mathbf s}
  \frac{\dim\cH_\mu}{\dim\cH_\nu}c_\mu(p)
  \prod_{\alpha\in\mathsf R_+}
  \bigl(1-s_{\alpha,n}+s_{\alpha,n}q_\alpha\bigr)^{r_\alpha}
  +O_L(n^{-1/2}),
  \label{eq:sector-output-fixed-matrix}\\
  \langle e_{\mathbf r}^\nu,
    \rho_{p,\nu}e_{\mathbf s}^\nu\rangle
  &=\delta_{\mathbf r,\mathbf s}
    c_\nu(p)\prod_{\alpha\in\mathsf R_+}q_\alpha^{r_\alpha}.
  \label{eq:sector-target-fixed-matrix}
\end{align}

With \(s_{\alpha,n}=\gamma^{-1}+o(1)\) and
\(q_\alpha^{(\gamma)}=1-(1-q_\alpha)/\gamma\), the output entries in
\eqref{eq:sector-output-fixed-matrix} equal
\(\delta_{\mathbf r,\mathbf s}\prod_\alpha
(1-q_\alpha^{(\gamma)})(q_\alpha^{(\gamma)})^{r_\alpha}+o(1)\)
uniformly for fixed \(L\).  The target entries in
\eqref{eq:sector-target-fixed-matrix} have the same form with
\(q_\alpha\) in place of \(q_\alpha^{(\gamma)}\).  These are the
number-basis entries of \(\tau_p^{(\gamma)}\) and \(\tau_p\),
respectively.  The cutoff space has
fixed finite dimension, so this uniform entrywise convergence gives
\begin{equation}\label{eq:output-fixed-height-convergence}
  \left\|
  J_{\nu,L}P_{\nu,L}\sigma_{\mu,\nu}P_{\nu,L}J_{\nu,L}^*
  -P_L^{\cF}\tau_p^{(\gamma)}P_L^{\cF}
  \right\|_1\longrightarrow0.
\end{equation}
\begin{equation}\label{eq:target-fixed-height-convergence}
  \left\|
  J_{\nu,L}P_{\nu,L}\rho_{p,\nu}P_{\nu,L}J_{\nu,L}^*
  -P_L^{\cF}\tau_pP_L^{\cF}
  \right\|_1\longrightarrow0.
\end{equation}

To obtain the fidelity assertions, let \((a_\mu)_\mu\) be a
probability distribution supported on labels \(\mu\in\mathsf Y_n\) with
\(\|\mu/n-p\|_\infty\le\delta_n\).
All such pairs \((\mu,\nu)\) are compatible for large \(n\), by the first
part of the proof.  Let
\[
  \overline\sigma_{p,\nu}
  =\sum_\mu a_\mu\,\sigma_{\mu,\nu}.
\]
For fixed \(L\), let \(\varepsilon_{n,L}\to0\) bound both trace-norm
errors in \eqref{eq:output-fixed-height-convergence} and
\eqref{eq:target-fixed-height-convergence}, uniformly over the allowed
parameters.  Since the output frame \(J_{\nu,L}P_{\nu,L}\) is independent
of \(\mu\), convexity gives the same bound for \(\overline\sigma_{p,\nu}\),
uniformly over all mixture weights.

Both \(\overline\sigma_{p,\nu}\) and \(\rho_{p,\nu}\)
commute with the diagonal torus.  The second assertion is immediate.
The first follows from torus invariance of each \(\rho_{p,\mu}\) and
covariance of \(\cC_{\mu,\nu}\), and persists under the mixture.
Consequently, \(P_{\nu,L}\) reduces both states, and the direct-sum
identity gives
\[
  F(\overline\sigma_{p,\nu},\rho_{p,\nu})
  =F(P_{\nu,L}\overline\sigma_{p,\nu}P_{\nu,L},
        P_{\nu,L}\rho_{p,\nu}P_{\nu,L})
  +F((I-P_{\nu,L})\overline\sigma_{p,\nu}(I-P_{\nu,L}),
       (I-P_{\nu,L})\rho_{p,\nu}(I-P_{\nu,L})).
\]
The trace bound
\(F(A,B)\le\sqrt{\Tr A\,\Tr B}\) shows that the second term is at most
\(\sqrt{\Tr[(I-P_{\nu,L})\rho_{p,\nu}]}\), because
\(\Tr[(I-P_{\nu,L})\overline\sigma_{p,\nu}]\le1\).

Compactness of \(K\Subset\Delta_d^\circ\) gives
\(b_K:=\sup_{p\in K,\alpha\in\mathsf R_+}q_\alpha<1\).
Let \(k_L:=\lfloor L/((d-1)r_d)\rfloor+1\).
Since there are \(r_d\) roots, each of height at most \(d-1\), a
number vector outside \(P_L^{\cF}\) has \(r_\alpha\ge k_L\) for some
\(\alpha\).  Under \(\tau_p\), each \(r_\alpha\) has the geometric law
\(\Pr_{\tau_p}(r_\alpha=j)=(1-q_\alpha)q_\alpha^j\), so
\(\Pr_{\tau_p}(r_\alpha\ge k_L)=q_\alpha^{k_L}\le b_K^{k_L}\).
A union bound gives
\begin{equation}\label{eq:uniform-thermal-height-tail}
  t_L:=\sup_{p\in K}\Tr[(I-P_L^{\cF})\tau_p]
  \le
  r_d\,b_K^{k_L}
  \longrightarrow0\qquad(L\to\infty).
\end{equation}
Taking traces in
\eqref{eq:target-fixed-height-convergence} now gives
\(\Tr[(I-P_{\nu,L})\rho_{p,\nu}]\le t_L+\varepsilon_{n,L}\)
uniformly for fixed \(L\).  The states \(\tau_p^{(\gamma)}\) and
\(\tau_p\) are number diagonal, so \(P_L^{\cF}\) reduces them. The same
trace bound makes their
complementary fidelity at most \(\sqrt{t_L}\).  Consequently,
\[
\begin{aligned}
 \left|F(\overline\sigma_{p,\nu},\rho_{p,\nu})
       -F(\tau_p^{(\gamma)},\tau_p)\right|
 &\le
 \left|F(P_{\nu,L}\overline\sigma_{p,\nu}P_{\nu,L},
         P_{\nu,L}\rho_{p,\nu}P_{\nu,L})
       -F(P_L^{\cF}\tau_p^{(\gamma)}P_L^{\cF},
          P_L^{\cF}\tau_pP_L^{\cF})\right|\\
 &\qquad+
 \sqrt{\Tr[(I-P_{\nu,L})\rho_{p,\nu}]}
 +\sqrt{\Tr[(I-P_L^{\cF})\tau_p]}.
\end{aligned}
\]
By \eqref{eq:fidelity-continuity} and
\eqref{eq:output-fixed-height-convergence}--\eqref{eq:target-fixed-height-convergence},
together with the mixture bound above, for every fixed \(L\) the
cutoff-fidelity difference on the right is at most
\(2\sqrt{\varepsilon_{n,L}}\) uniformly.  The preceding tail estimates
therefore give
\[
  \left|F(\overline\sigma_{p,\nu},\rho_{p,\nu})
  -F(\tau_p^{(\gamma)},\tau_p)\right|
  \le 2\sqrt{\varepsilon_{n,L}}
     +\sqrt{t_L+\varepsilon_{n,L}}+\sqrt{t_L}.
\]
First let \(n\to\infty\) at fixed \(L\), and then let
\(L\to\infty\), using \(\varepsilon_{n,L}\to0\) uniformly and
\eqref{eq:uniform-thermal-height-tail}.  This proves uniform
convergence.  Finally,
\[
  F(\tau_p^{(\gamma)},\tau_p)
  =\prod_{i<j}F(\tau_{q_{ij}^{(\gamma)}},\tau_{q_{ij}})
  =\Forb(\gamma,p)
\]
by multiplicativity of the trace norm under tensor
products, \eqref{eq:thermal-root-fidelity}, and \eqref{eq:one-mode-f}.
This proves the mixture assertion of \cref{prop:sector-fidelity}.
Taking a point mass gives the individual-sector assertion.
\end{proof}

\section{Young local limit and randomized rounding}
\label[appendix]{app:young-rounding}

To prove \cref{prop:young-rounding}, we represent each
Young-label law as a density constant on lattice cells.  In this
representation, randomized rounding is a dilation followed by cell
averaging.  A uniform Gaussian local limit then yields
\(\Fcl(\gamma,d)\).

Define the half-open fundamental cell
\[
  \mathsf C
  =\left\{\sum_{i=1}^{d-1}t_i(e_i-e_d):-\frac12\le t_i<\frac12\right\}
  \subset\mathsf H.
\]
Volumes are taken with the induced \((d-1)\)-dimensional Lebesgue measure.
Except on cell boundaries, which have probability zero,
the rounding rule \eqref{eq:randomized-rounding} selects the unique
\(\widetilde\nu\in\mathsf L_{m_n}\) for which
\[
  \gamma_n(\mu+y)\in\widetilde\nu+\mathsf C,
  \qquad
  y=\sum_{i=1}^{d-1}y_i(e_i-e_d).
\]
The first \(d-1\) coordinates follow the nearest-integer rule. The final
coordinate is determined by the required sum.  Since \(y\) is uniform
on \(\mathsf C\),
\begin{equation}\label{eq:cell-kernel}
  \mathsf R_n(\widetilde\nu\mid\mu)
  =\frac{|\gamma_n(\mu+\mathsf C)\cap(\widetilde\nu+\mathsf C)|}
         {\gamma_n^{d-1}|\mathsf C|}.
\end{equation}
Extend every label law on \(\mathsf Y_N\) by zero to \(\mathsf L_N\).

For \(\lambda\in\mathsf L_N\), define the recentered and
rescaled lattice cell
\[
  \mathsf C_{N,p}(\lambda)=\frac{\lambda-Np+\mathsf C}{\sqrt N}.
\]
For any \(R:\mathsf L_N\to\R\) satisfying
\(\sum_{\lambda\in\mathsf L_N}|R(\lambda)|<\infty\), define its
piecewise-constant interpolation by
\[
  (\mathcal I_{N,p}R)(x)
  =\frac{N^{(d-1)/2}}{|\mathsf C|}R(\lambda)
  \quad(x\in\mathsf C_{N,p}(\lambda)).
\]
These cells tile \(\mathsf H\) up to boundaries of zero
Lebesgue measure.  Each has volume
\(|\mathsf C|N^{-(d-1)/2}\), which cancels the density prefactor
\(N^{(d-1)/2}/|\mathsf C|\) in \(\mathcal I_{N,p}\).
Hence \(\mathcal I_{N,p}\) preserves \(\ell^1\)-distance,
\(\|\mathcal I_{N,p}R-\mathcal I_{N,p}S\|_1=\|R-S\|_{\ell^1}\), and, for
nonnegative \(R,S\), Hellinger affinity,
\(\int_{\mathsf H}\sqrt{(\mathcal I_{N,p}R)(\mathcal I_{N,p}S)}\,\dd x
=\sum_{\lambda\in\mathsf L_N}\sqrt{R(\lambda)S(\lambda)}\).

\begin{lemma}[Uniform Young local limit]\label{lem:young-local-limit}
For every compact \(K\Subset\Delta_d^\circ\),
\begin{equation}\label{eq:young-local-limit}
  \sup_{p\in K}
  \left\|\mathcal I_{N,p}P_{N,p}-\varphi_{0,\Sigma_p}\right\|_1
  \longrightarrow0.
\end{equation}
\end{lemma}

\begin{proof}
Write \(f_{N,p}=\mathcal I_{N,p}P_{N,p}\), and fix a ball
\(B_R=\{x\in\mathsf H:\|x\|\le R\}\).
For every cell \(\mathsf C_{N,p}(\lambda)\) meeting
\(B_R\), its center \(x_\lambda=(\lambda-Np)/\sqrt N\) lies in a
bounded set independent of \(N,p\).  Compactness of
\(K\Subset\Delta_d^\circ\) gives positive lower bounds on \(p_d\)
and all adjacent gaps, so these \(\lambda\) belong to
\(\mathsf Y_N\) for all sufficiently large \(N\).
The Frobenius dimension formula and
\cref{lem:uniform-character-normalization} give, uniformly on these cells,
\[
  P_{N,p}(\lambda)
  =\frac{N!\prod_i p_i^{\lambda_i}}{\prod_i\lambda_i!}
   \prod_{i<j}\frac{\lambda_i-\lambda_j+j-i}{N(p_i-p_j)}
   \prod_i\prod_{a=1}^{d-i}\frac{Np_i}{\lambda_i+a}
   \bigl(1+o(1)\bigr).
\]
Each factor in the two last products is \(1+O_{K,R}(N^{-1/2})\).
Stirling's formula and a Taylor expansion of the first factor therefore
yield
\[
  P_{N,p}(\lambda)
  =\frac{\exp\!\left[-\frac12\sum_i x_{\lambda,i}^2/p_i\right]}
  {(2\pi N)^{(d-1)/2}\sqrt{\prod_i p_i}}\bigl(1+o(1)\bigr).
\]
The normalization agrees with the Gaussian density on \(\mathsf H\):
\[
  |\mathsf C|=\sqrt d,\qquad
  \det_{\mathsf H}\Sigma_p=d\prod_i p_i,\qquad
  \langle x,\Sigma_p^{-1}x\rangle=\sum_i\frac{x_i^2}{p_i}
  \quad(x\in\mathsf H).
\]
Indeed, the Gram matrix of \((e_i-e_d)_{i<d}\) has determinant \(d\).
The ambient covariance \(\diag(p)-pp^{\mathsf T}\) has kernel
\(\operatorname{span}\{(1,\ldots,1)\}\).  Deleting row and column \(i\)
leaves a principal minor of determinant
\((\prod_{j\ne i}p_j)(1-\sum_{j\ne i}p_j)=\prod_jp_j\).
The product of the nonzero eigenvalues is the sum of these \(d\) principal
minors, giving \(\det_{\mathsf H}\Sigma_p=d\prod_i p_i\).  Its inverse on
\(\mathsf H\) sends \(x\) to \(\diag(p)^{-1}x\) up to a multiple of
\((1,\ldots,1)\), giving the quadratic identity.
For cells meeting \(B_R\), the Gaussian log-density
has gradient norm
\(\|\nabla\log\varphi_{0,\Sigma_p}(z)\|
=\|\Sigma_p^{-1}z\|=O_{K,R}(1)\), while each cell has diameter
\(O(N^{-1/2})\).
Thus its density changes by a factor \(1+O_{K,R}(N^{-1/2})\) within a
cell.  Together with the cell-center estimate for \(P_{N,p}(\lambda)\),
this gives
\[
  \sup_{p\in K}\sup_{x\in B_R}
  |f_{N,p}(x)-\varphi_{0,\Sigma_p}(x)|\longrightarrow0.
\]
Both densities have integral one.  Consequently,
\[
  \|f_{N,p}-\varphi_{0,\Sigma_p}\|_1
  \le 2\int_{B_R}|f_{N,p}-\varphi_{0,\Sigma_p}|\,\dd x
     +2\int_{B_R^c}\varphi_{0,\Sigma_p}(x)\,\dd x.
\]
Let \(N\to\infty\), then \(R\to\infty\).  The Gaussian tails
vanish uniformly on compact \(K\), proving \eqref{eq:young-local-limit}.
\end{proof}

\begin{lemma}[Continuous limit of randomized dilation]
\label{lem:randomized-dilation}
Let \(\widetilde Q_{m_n,p}=P_{n,p}\mathsf R_n\) be the law
obtained by applying \(\mathsf R_n\) to an input drawn from \(P_{n,p}\), before
the fallback replacement in \eqref{eq:rounding-kernel}.  For every compact
\(K\Subset\Delta_d^\circ\),
\begin{equation}\label{eq:rounded-density-limit}
  \sup_{p\in K}
  \left\|\mathcal I_{m_n,p}\widetilde Q_{m_n,p}
  -\varphi_{0,\gamma\Sigma_p}\right\|_1\longrightarrow0.
\end{equation}
\begin{equation}\label{eq:raw-rounded-affinity}
  \sup_{p\in K}
  \left|F(\widetilde Q_{m_n,p},P_{m_n,p})-\Fcl(\gamma,d)\right|
  \longrightarrow0.
\end{equation}
\end{lemma}

\begin{proof}
For \(a>0\) and \(f\in L^1(\mathsf H)\), set
\[
  (\operatorname{dil}_a f)(z)=a^{-(d-1)/2}f(z/\sqrt a),
\]
and define the target-cell averaging operator by
\[
  (\operatorname{Av}_{m_n,p}f)(x)
  :=\frac{m_n^{(d-1)/2}}{|\mathsf C|}
   \int_{\mathsf C_{m_n,p}(\nu)}f(z)\,\dd z,
  \qquad x\in\mathsf C_{m_n,p}(\nu).
\]
If \(X=(\mu-np+y)/\sqrt n\), where \(\mu\sim P_{n,p}\) and
\(y\) is independent of \(\mu\) and uniform on \(\mathsf C\), then
\(X\) has density \(\mathcal I_{n,p}P_{n,p}\).  Because
\(m_n=\gamma_n n\),
\(\sqrt{\gamma_n}X=
[\gamma_n(\mu+y)-m_np]/\sqrt{m_n}\).  Thus the condition
\(\sqrt{\gamma_n}X\in\mathsf C_{m_n,p}(\nu)\) is equivalent to the
rounding event \(\gamma_n(\mu+y)\in\nu+\mathsf C\).  Hence the dilated
density assigns mass \(\widetilde Q_{m_n,p}(\nu)\) to that target cell.
Averaging over each target cell gives
\begin{equation}\label{eq:rounding-intertwining}
 \operatorname{Av}_{m_n,p}\operatorname{dil}_{\gamma_n}
   (\mathcal I_{n,p}P_{n,p})
 =\mathcal I_{m_n,p}\widetilde Q_{m_n,p}.
\end{equation}
Cell averaging is an \(L^1\)-contraction and dilation an \(L^1\)-isometry.
Since
\(\operatorname{dil}_{\gamma_n}\varphi_{0,\Sigma_p}
=\varphi_{0,\gamma_n\Sigma_p}\), \eqref{eq:rounding-intertwining} yields
\[
  \|\mathcal I_{m_n,p}\widetilde Q_{m_n,p}
       -\varphi_{0,\gamma\Sigma_p}\|_1
  \le \|\mathcal I_{n,p}P_{n,p}-\varphi_{0,\Sigma_p}\|_1
  +\|\operatorname{Av}_{m_n,p}\varphi_{0,\gamma_n\Sigma_p}
               -\varphi_{0,\gamma_n\Sigma_p}\|_1
  +\|\varphi_{0,\gamma_n\Sigma_p}
               -\varphi_{0,\gamma\Sigma_p}\|_1.
\]
The first term tends uniformly to zero by \cref{lem:young-local-limit},
and the last does so because \(\gamma_n\to\gamma\).  For the middle
term, the cellwise \(L^1\) Poincar\'e inequality gives, for smooth \(f\)
with integrable gradient,
\[
  \|\operatorname{Av}_{m_n,p}f-f\|_1
  \le C_d m_n^{-1/2}\|\nabla f\|_1.
\]
Gaussian gradients have uniformly bounded \(L^1\)-norm on the compact
covariance family \(\gamma_n\Sigma_p\), so this term is
\(O_K(m_n^{-1/2})\).  This proves \eqref{eq:rounded-density-limit}.

Applying \cref{lem:young-local-limit} at sample size \(m_n\) gives
\[
  \mathcal I_{m_n,p}P_{m_n,p}\longrightarrow\varphi_{0,\Sigma_p}
\]
uniformly in \(L^1\).  The embedding preserves Hellinger affinity, and
\eqref{eq:fidelity-continuity} also holds for nonnegative \(L^1\) densities.
Therefore
\[
  F(\widetilde Q_{m_n,p},P_{m_n,p})
  \longrightarrow
  F(\mathsf N_{0,\gamma\Sigma_p},\mathsf N_{0,\Sigma_p})
  =\Fcl(\gamma,d)
\]
uniformly on \(K\).
\end{proof}

\begin{proof}[Proof of \cref{prop:young-rounding}]
Write \(e=\widetilde\nu-\gamma_n\mu\).  The nearest-integer rule gives
\(|e_i|\le(\gamma_n+1)/2\) for \(i<d\), and \(e_d=-\sum_{i<d}e_i\).
Thus \(\|e\|_\infty\le c_{\mathrm{rnd}}\), where
\[
  c_{\mathrm{rnd}}
  =\frac{d-1}{2}\left(1+\sup_n\gamma_n\right)<\infty,
\]
where finiteness follows from \(\gamma_n\to\gamma\).

For \(p\in K\) and a typical input label
\(\mu\in\mathsf T_n(p)\), the identity
\(\widetilde\nu/m_n=\mu/n+e/m_n\) gives
\[
  \left\|\frac{\widetilde\nu}{m_n}-p\right\|_\infty
  \le n^{-1/3}+\frac{c_{\mathrm{rnd}}}{m_n}.
\]
Apply the compatibility assertion of
\cref{prop:sector-fidelity} with
\(\delta_n=n^{-1/3}+c_{\mathrm{rnd}}/m_n\).  For all sufficiently large
\(n\), it yields \(\widetilde\nu-\mu\vdash m_n-n\), so the fallback
rule leaves every typical rounded label unchanged.  This proves
\eqref{eq:rounding-support}.

Coupling the raw and final labels by the fallback rule therefore gives
\[
  \|Q_{m_n,p}-\widetilde Q_{m_n,p}\|_{\ell^1}
  \le 2P_{n,p}(\mathsf T_n(p)^c)=o(1)
\]
uniformly on \(K\), by \cref{lem:young-concentration}.  Combining
\eqref{eq:raw-rounded-affinity} with \eqref{eq:fidelity-continuity} proves
\eqref{eq:rounding-label-fidelity}.
\end{proof}

\section{Unified Gaussian flat-prior converse}
\label[appendix]{app:gaussian-converse}

We prove the two converses in
\cref{thm:gaussian-amplification} through the common flat-prior bound in
\cref{prop:flat-prior-converse}.  Its proof combines weighted fidelity,
limits tested on compact observables, a least-noise oscillator moment,
and Gaussian translation.

The classical parameter space has dimension \(k=0\) in the orbital model
and \(k=d-1\) in the full model.
All trace-class-valued integrals below are Bochner integrals in trace norm.
For classical--quantum maps, complete positivity means positivity of
\(\operatorname{id}_{M_a}\otimes\Lambda\) for every \(a\ge1\), with the
matrix factor acting on the quantum trace-class fibers. The classical
algebra is abelian.

\subsection{Weighted fidelity and compact limits}

\begin{lemma}[Weighted fidelity bound]\label{lem:weighted-fidelity}
Let \(R,T\) be positive trace-class operators, and let \(W\) be bounded,
positive, and injective.  Define
\[
  \Tr(TW^{-1})
  :=\lim_{\varepsilon\downarrow0}
  \Tr\bigl(T(W+\varepsilon I)^{-1}\bigr)\in[0,\infty].
\]
If this quantity is finite, then
\begin{equation}\label{eq:weighted-quantum-fidelity}
  F(R,T)^2\le\Tr(RW)\Tr(TW^{-1}).
\end{equation}

Let \(\mathsf E\) be a finite-dimensional Euclidean space,
let \(R\in L^1(\mathsf E;\cT_1(\cF_{\mathrm{orb}}))\) be positive, let
\(g\) be a probability density on \(\mathsf E\), and let \(\chi>0\) be
measurable.  For positive \(R,S\) in this space, write
\(F(R,S)=\int_{\mathsf E}F(R(y),S(y))\,\dd y\), in agreement with the
hybrid-state convention in \cref{sec:gaussian-LAN}.
The notation \(gT\) means the density \(y\mapsto g(y)T\).
Whenever the right-hand side below is finite,
\begin{equation}\label{eq:weighted-hybrid-fidelity}
  F(R,gT)^2
  \le
  \left(\int_{\mathsf E}\chi(y)\Tr[R(y)W] \,\dd y\right)
  \left(\int_{\mathsf E}\frac{g(y)}{\chi(y)}\,\dd y\right)
  \Tr(TW^{-1}).
\end{equation}
\end{lemma}

\begin{proof}
For \(W_\varepsilon=W+\varepsilon I\), the Schatten H\"older inequality gives
\[
  F(R,T)
  =\|(R^{1/2}W_\varepsilon^{1/2})(W_\varepsilon^{-1/2}T^{1/2})\|_1
  \le\sqrt{\Tr(RW_\varepsilon)}\sqrt{\Tr(TW_\varepsilon^{-1})}.
\]
Squaring and letting \(\varepsilon\downarrow0\), using
\(\Tr(RW_\varepsilon)=\Tr(RW)+\varepsilon\Tr R\) and monotone convergence
for the second trace, proves \eqref{eq:weighted-quantum-fidelity}.  For the
hybrid state,
\[
  F(R,gT)
  =\int_{\mathsf E}\sqrt{g(y)}\,F(R(y),T)\,\dd y.
\]
Apply the first inequality pointwise and then Cauchy--Schwarz to the factors
\(\sqrt{\chi(y)\Tr[R(y)W]}\) and
\(\sqrt{g(y)/\chi(y)}\).  This gives
\eqref{eq:weighted-hybrid-fidelity}.
\end{proof}

\begin{lemma}[Compact-observable limits of completely positive maps]
\label{lem:compact-cp-limit}
Let \(\cH_0,\cH_1\) be separable Hilbert spaces and let
\(\Gamma_j:\cT_1(\cH_0)\to\cT_1(\cH_1)\) be completely positive maps with
\(\sup_j\|\Gamma_j\|_{1\to1}<\infty\).  There are a subsequence and a
completely positive map
\(\Gamma:\cT_1(\cH_0)\to\cT_1(\cH_1)\) such that
\begin{equation}\label{eq:compact-map-convergence}
  \Tr[\Gamma_j(X)O]\longrightarrow\Tr[\Gamma(X)O]
\end{equation}
for every \(X\in\cT_1(\cH_0)\) and every compact operator
\(O\) on \(\cH_1\).  If
\(\limsup_j\Tr\Gamma_j(X)\le c\Tr X\) for every \(X\ge0\), then
\(\Tr\Gamma(X)\le c\Tr X\).

Suppose in addition that
\(U_g\in\U(\cH_0)\) and \(V_g\in\U(\cH_1)\) are families of
unitaries indexed by a set \(G\), and that, for every fixed \(g\in G\) and
\(X\in\cT_1(\cH_0)\),
\[
  \|\Gamma_j(U_g X U_g^*)-V_g\Gamma_j(X)V_g^*\|_1
  \longrightarrow0.
\]
Then
\begin{equation}\label{eq:compact-limit-covariance}
  \Gamma(U_g X U_g^*)=V_g\Gamma(X)V_g^*
\end{equation}
for every \(g\in G\) and \(X\in\cT_1(\cH_0)\).
\end{lemma}

\begin{proof}
Let \(M=\sup_j\|\Gamma_j\|_{1\to1}\), and let
\(\mathcal D\) be a countable norm-dense set of compact operators on
\(\cH_1\).
Since \(\cT_1(\cH_0)\) is separable, a diagonal weak-*
compactness argument yields a subsequence along which
\(\Gamma_j^*(O)\) converges weak-* in \(\mathcal B(\cH_0)\) for every
\(O\in\mathcal D\).  The bound
\(\|\Gamma_j^*(O)\|\le M\|O\|\) extends both the limit and the convergence
to every compact \(O\), producing a bounded complex-linear map
\(\Gamma^*\) from compact operators on \(\cH_1\) to
\(\mathcal B(\cH_0)\).

For every positive block matrix \([O_{ab}]\) of compact operators,
\([\Gamma_j^*(O_{ab})]\ge0\), so its entrywise weak-* limit
\([\Gamma^*(O_{ab})]\) is positive.  Thus \(\Gamma^*\) is completely positive.

Since the dual of the compact operators on \(\cH_1\) is
\(\cT_1(\cH_1)\), define \(\Gamma(X)\) by
\[
  \Tr[\Gamma(X)O]=\Tr[X\Gamma^*(O)]
  \qquad(O\text{ compact}).
\]
Then \(\Gamma\) is completely positive and
\eqref{eq:compact-map-convergence} holds.  If \(P_N\uparrow I\) are
finite-rank projections and \(X\ge0\), then
\[
  \Tr\Gamma(X)
  =\sup_N\lim_j\Tr[\Gamma_j(X)P_N]
  \le\limsup_j\Tr\Gamma_j(X)
  \le c\Tr X.
\]
Finally, for fixed \(g,X\) and compact \(O\), the operator
\(V_g^*OV_g\) is compact.  Passing to the limit gives
\[
  \Tr[\Gamma(U_gXU_g^*)O]
  =\Tr[V_g\Gamma(X)V_g^*O].
\]
Compact operators separate trace-class operators, so the limiting map is
covariant.
\end{proof}

\subsection{The least-noise oscillator moment}

Fix \(\gamma>1\).  For \(s\ge1\) and
\(q_1,\ldots,q_s\in(0,1)\), let
\[
  q_\ell^{(\gamma)}=1-\frac{1-q_\ell}{\gamma},
  \qquad
  \tau=\bigotimes_{\ell=1}^s\tau_{q_\ell},
  \qquad
  \tau^{(\gamma)}
  =\bigotimes_{\ell=1}^s\tau_{q_\ell^{(\gamma)}}.
\]

\begin{lemma}[Covariant-amplifier representation]
\label{lem:covariant-amplifier-representation}
Let \(G_1,\ldots,G_s>1\), and write
\[
  \sqrt{\mathbf G}z=(\sqrt{G_\ell}z_\ell)_{\ell=1}^s.
\]
Every \(s\)-mode channel \(\Lambda\) satisfying
\[
  \Lambda(D(z)XD(z)^*)
  =D(\sqrt{\mathbf G}z)\Lambda(X)D(\sqrt{\mathbf G}z)^*
\]
admits an idler realization
\begin{equation}\label{eq:covariant-amplifier-idler}
  a_{\ell,\mathrm{out}}
  =\sqrt{G_\ell}\,a_{\ell,\mathrm{in}}
   +\sqrt{G_\ell-1}\,b_\ell^*,
  \qquad \ell=1,\ldots,s.
\end{equation}
Here \(a_{\ell,\mathrm{in}}\) and \(b_\ell\) annihilate the
input and idler modes, respectively, and \(b_\ell^*\) creates an idler
excitation.  For some joint idler state \(\sigma\), a product \(U\) of
two-mode squeezers implements
\(\Lambda(X)=\Tr_{\mathrm{idler}}[U(X\otimes\sigma)U^*]\), with
\(a_{\ell,\mathrm{out}}=U^*(a_{\ell,\mathrm{in}}\otimes I)U\) in the
Heisenberg picture.
The channel uniquely determines the joint idler state,
which may be non-Gaussian, nondiagonal in the number basis, and
correlated across modes.
\end{lemma}

\begin{proof}
Covariance and Weyl irreducibility give
\(\Lambda^*(D(w))=f(w)D(\sqrt{\mathbf G}w)\).  Vacuum matrix elements
show that \(f\) is continuous at zero, so \(\Lambda\) is a linear
bosonic channel. See \cite[Ch.~12]{Holevo2013QuantumSystems} for the
bosonic-channel formalism.  Its defect form is
\(\bigoplus_{\ell=1}^s(1-G_\ell)\Omega_2\), where
\(\Omega_2=\left(\begin{smallmatrix}0&1\\-1&0\end{smallmatrix}\right)\).
This form is invertible, so \cite[Lemma~7]{LamiSabapathyWinter2018}
gives an exact dilation with \(s\) idler modes.  Choose the noise matrix
\(\bigoplus_\ell\sqrt{G_\ell-1}\diag(1,-1)\) and complete these
blocks with a product of two-mode squeezers, yielding
\eqref{eq:covariant-amplifier-idler}.  The noise matrix is invertible,
so \(f\) determines the idler characteristic function and hence its state.
\end{proof}

\begin{proposition}[Least-noise moment]
\label{prop:least-noise-moment}
Let \(\Gamma\) be a completely positive trace-nonincreasing map on the
\(s\)-mode trace class satisfying
\[
  \Gamma(D(z)XD(z)^*)
  =D(\sqrt\gamma\,z)\Gamma(X)D(\sqrt\gamma\,z)^*
  \qquad(z\in\mathbb C^s).
\]
Let
\(
  W(\widehat{\mathbf N}):=\prod_{\ell=1}^s w_\ell(\widehat N_\ell),
\)
where every \(w_\ell\) is bounded, nonnegative, and nonincreasing on
\(\mathbb Z_{\ge0}\).  Then
\begin{equation}\label{eq:least-noise-moment}
  \Tr[\Gamma(\tau)W(\widehat{\mathbf N})]
  \le\Tr[\Gamma(\tau)]
       \Tr[\tau^{(\gamma)}W(\widehat{\mathbf N})].
\end{equation}
\end{proposition}

\begin{proof}
Displacement covariance makes \(\Gamma^*(I)\) commute with every Weyl
displacement.  Irreducibility therefore gives \(\Gamma^*(I)=cI\), with
\(c=\Tr\Gamma(\tau)\in[0,1]\).  The case \(c=0\) is immediate. Otherwise
\(\Lambda=c^{-1}\Gamma\) is a channel with the same covariance.

The quantum-limited amplifier \(\cA_G\) sends vacuum to
\(\tau_{1-1/G}\).  Thus
\begin{equation}\label{eq:thermal-as-amplified-vacuum}
  \tau=\left(\bigotimes_{\ell=1}^s\cA_{1/(1-q_\ell)}\right)
        (\proj{\mathbf0}).
\end{equation}
The composite
\(\mathcal B=\Lambda\circ\bigotimes_\ell\cA_{1/(1-q_\ell)}\)
therefore sends joint vacuum to \(\Lambda(\tau)\), with modewise gains
\[
  G_\ell=\frac{\gamma}{1-q_\ell},\qquad
  \lambda_\ell=1-G_\ell^{-1}=q_\ell^{(\gamma)}.
\]
Apply \cref{lem:covariant-amplifier-representation} to \(\mathcal B\),
and let \(\sigma\) be its joint idler state.

For a single two-mode squeezer, normal-ordered disentangling
\cite{BarnettRadmore1997} gives, with \(\lambda=1-G^{-1}\),
\begin{equation}\label{eq:seeded-amplifier-expansion}
  U_G\ket{0,l}
  =(1-\lambda)^{(l+1)/2}
   \sum_{k\ge0}\sqrt{\binom{k+l}{l}}\,
   \lambda^{k/2}\ket{k,k+l}.
\end{equation}
For fixed \(\mathbf k\), tracing out the idler requires
\(\mathbf k+\mathbf l=\mathbf k+\mathbf l'\), hence
\(\mathbf l=\mathbf l'\).  Thus only
\(a_{\mathbf l}=\langle\mathbf l|\sigma|\mathbf l\rangle\)
enters the output number probabilities:
\begin{equation}\label{eq:multimode-negative-binomial-law}
  \langle\mathbf k|\Lambda(\tau)|\mathbf k\rangle
  =\sum_{\mathbf l\in\mathbb Z_{\ge0}^s}a_{\mathbf l}
    \prod_{\ell=1}^s
    \binom{k_\ell+l_\ell}{l_\ell}
    (1-\lambda_\ell)^{l_\ell+1}\lambda_\ell^{k_\ell}.
\end{equation}
Conditional on \(\mathbf l\), the coordinates are independent
\(\mathrm{NB}(l_\ell+1,\lambda_\ell)\) variables. The mixing law
\(a_{\mathbf l}\) may be correlated.  No product or phase-invariance
assumption on the idler is needed.

Let \(\mathbf L=(L_1,\ldots,L_s)\) satisfy
\(\Pr(\mathbf L=\mathbf l)=a_{\mathbf l}\).
Independently, draw geometric variables \(X_{\ell,j}\), mutually
independent over \(\ell\) and \(j\), with
\(\Pr(X_{\ell,j}=k)=(1-\lambda_\ell)\lambda_\ell^k\) for
\(j,k\ge0\).
The coupling
\[
  N_\ell=\sum_{j=0}^{L_\ell}X_{\ell,j}\ge X_{\ell,0}
  \qquad(\ell=1,\ldots,s)
\]
has the output law \eqref{eq:multimode-negative-binomial-law}, while
\((X_{\ell,0})_\ell\) has the number law of \(\tau^{(\gamma)}\).
Since each \(w_\ell\) is nonnegative and nonincreasing,
\[
  \Tr[\Lambda(\tau)W]
  =\mathbb E\prod_\ell w_\ell(N_\ell)
  \le\mathbb E\prod_\ell w_\ell(X_{\ell,0})
  =\Tr[\tau^{(\gamma)}W].
\]
Multiplying by \(c=\Tr\Gamma(\tau)\) proves
\eqref{eq:least-noise-moment}.
\end{proof}

Specialize to \(s=r_d\), \(q_\ell=q_{ij}\),
\(\tau=\tau_p\), and \(\tau^{(\gamma)}=\tau_p^{(\gamma)}\).
For
\(\mathbf k=(k_{ij})_{i<j}\in\mathbb Z_{\ge0}^{r_d}\), let
\(t_{\mathbf k}\) and \(t_{\mathbf k}^{(\gamma)}\) be the corresponding
number-basis eigenvalues, and define
\begin{equation}\label{eq:quantum-witness}
  W_\gamma
  =\sum_{\mathbf k}
  \sqrt{\frac{t_{\mathbf k}}{t_{\mathbf k}^{(\gamma)}}}
  \proj{\mathbf k}.
\end{equation}
Equivalently,
\[
  W_\gamma=\prod_{i<j}w_{ij}(\widehat N_{ij}),
  \qquad
  w_{ij}(n)=
  \sqrt{\frac{1-q_{ij}}{1-q_{ij}^{(\gamma)}}}
  \left(\frac{q_{ij}}{q_{ij}^{(\gamma)}}\right)^{n/2}.
\]
Here \(\widehat N_{ij}\) is the number operator of the mode indexed by
\(i<j\), and \(\lvert\mathbf k\rvert:=\sum_{i<j}k_{ij}\) is the total
occupation number.
Since \(q_{ij}<q_{ij}^{(\gamma)}\), each \(w_{ij}(n)\) is positive, bounded,
and strictly decreasing, and the eigenvalues of \(W_\gamma\) tend to zero
as \(\lvert\mathbf k\rvert\to\infty\).  Thus \(W_\gamma\) is bounded,
injective, and compact, with
\begin{equation}\label{eq:quantum-witness-identities}
  \Tr(\tau_p^{(\gamma)}W_\gamma)
  =\Tr(\tau_p W_\gamma^{-1})
  =\sum_{\mathbf k}\sqrt{t_{\mathbf k}t_{\mathbf k}^{(\gamma)}}
  =\Forb(\gamma,p),
\end{equation}
where the inverse is understood as in \cref{lem:weighted-fidelity}.

\subsection{One flat-prior converse for both Gaussian models}

\begin{proposition}[Flat-prior converse inequality]
\label{prop:flat-prior-converse}
Let \(k\in\{0,d-1\}\).  For \(k=0\), let
\(\mathsf E=\{0\}\), omit the classical component, and interpret \(h\)-integration
as evaluation at \(0\).  For \(k=d-1\), let \(\mathsf E=\mathsf H\) with its
induced Euclidean measure.  Let
\[
  \mathsf E_L=\{h\in\mathsf E:\|h\|_\infty\le L\},
  \qquad
  \Xi_L=\mathsf E_L\times\mathsf Z_L,
\]
with \(|\mathsf E_L|=1\) when \(k=0\).  Define
\[
  R_{0,z}^{(0)}=\Phi_z^{(p)},\ \ \
   S_{0,z}^{(0)}=\Psi_z^{(p,\gamma)},\ \ \
  R_{h,z}^{(d-1)}=\Theta_{h,z}^{\mathrm{in}},\ \ \
   S_{h,z}^{(d-1)}=\Theta_{h,z}^{\mathrm{tar}},\ \ \
  c_k=\left(\frac{2\sqrt\gamma}{1+\gamma}\right)^{k/2}.
\]
For \(k=0\), the supremum below is over quantum channels on
\(\cF_{\mathrm{orb}}\). For \(k=d-1\), it is over classical--quantum
channels.  Then
\begin{equation}\label{eq:unified-flat-prior-converse}
  \limsup_{L\to\infty}\sup_\Lambda
  \frac1{|\Xi_L|}\int_{\Xi_L}
  F\!\left(\Lambda(R_{h,z}^{(k)}),S_{h,z}^{(k)}\right)
  \,\dd h\,\dd z
  \le c_k\Forb(\gamma,p).
\end{equation}
\end{proposition}

\begin{proof}
The classical weight below captures the spectral
affinity, while the compact witness \(W_\gamma\) controls the oscillator
output.
For \(k=0\), set \(g_1=g_\gamma=\chi=1\) on the one-point space.  For
\(k=d-1\), let \(g_1\) and \(g_\gamma\) be the centered Gaussian densities
with covariances \(\Sigma_p\) and \(\gamma\Sigma_p\), and set
\[
  \chi(y)=\sqrt{\frac{g_1(y)}{g_\gamma(y)}}
  =\gamma^{k/4}
   \exp\!\left[-\frac{\gamma-1}{4\gamma}
     \langle y,\Sigma_p^{-1}y\rangle\right].
\]
This function is bounded and strictly positive.  In both cases,
\begin{equation}\label{eq:classical-witness-identity}
  \int_{\mathsf E}\frac{g_1(y)}{\chi(y)}\,\dd y
  =\int_{\mathsf E}\sqrt{g_1(y)g_\gamma(y)}\,\dd y
  =c_k.
\end{equation}
For \(k=d-1\), completing the square also gives
\begin{equation}\label{eq:classical-witness-translation}
  \int_{\mathsf E}g_1(h)\chi(v+\sqrt\gamma h)\,\dd h
  =c_k\exp\!\left[-\frac{\gamma-1}{2\gamma(\gamma+1)}
    \langle v,\Sigma_p^{-1}v\rangle\right]
  \le c_k.
\end{equation}

For \(\xi=(h,z)\), define
\[
  (\mathcal U_\xi^{\mathrm{in}}R)(x)
  =D(z)R(x-h)D(z)^*,
  \qquad
  (\mathcal U_\xi^{\mathrm{out}}R)(y)
  =D(\sqrt\gamma z)R(y-\sqrt\gamma h)D(\sqrt\gamma z)^*.
\]
Then
\(R_\xi^{(k)}=\mathcal U_\xi^{\mathrm{in}}(g_1\tau_p)\) and
\(S_\xi^{(k)}=\mathcal U_\xi^{\mathrm{out}}(g_1\tau_p)\).  For a channel
\(\Lambda\), define its average pointwise by
\[
  \overline\Lambda_L(R)
  =\frac1{|\Xi_L|}\int_{\Xi_L}
   \mathcal U_{-\eta}^{\mathrm{out}}
   \bigl(\Lambda(\mathcal U_\eta^{\mathrm{in}}R)\bigr)\,\dd\eta.
\]
For each fixed input \(R\), the integrand is trace-norm continuous and
bounded in norm by \(\|R\|_1\), so the Bochner integral is well-defined.
Fidelity is invariant under the output isometries
\(\mathcal U_\xi^{\mathrm{out}}\), and is concave and trace-norm continuous
in its first argument. The finite-average inequality therefore passes
to the Bochner average by Riemann sums:
\begin{equation}\label{eq:unified-folner-concavity}
  F_L(\Lambda)
  =\frac1{|\Xi_L|}\int_{\Xi_L}
  F\!\left(\Lambda(R_\xi^{(k)}),S_\xi^{(k)}\right)\,\dd\xi
  \le F(\overline\Lambda_L(g_1\tau_p),g_1\tau_p).
\end{equation}

Let \(F_*=\limsup_{L\to\infty}\sup_\Lambda F_L(\Lambda)\).  Choose
\(L_j\to\infty\) with \(\sup_\Lambda F_{L_j}(\Lambda)\to F_*\), and channels
\(\widetilde\Lambda_j\) satisfying
\(F_{L_j}(\widetilde\Lambda_j)\ge\sup_\Lambda F_{L_j}(\Lambda)-j^{-1}\).
Set
\[
  \Lambda_j=\overline{\widetilde\Lambda_j}_{L_j},
  \qquad R_j=\Lambda_j(g_1\tau_p).
\]
Then
\begin{equation}\label{eq:unified-maximizing-sequence}
  F_{L_j}(\widetilde\Lambda_j)\longrightarrow F_*,
  \qquad
  F_{L_j}(\widetilde\Lambda_j)\le F(R_j,g_1\tau_p).
\end{equation}
Write
\[
  A_{j,u}=\mathcal U_{-u}^{\mathrm{out}}\circ
    \widetilde\Lambda_j\circ\mathcal U_u^{\mathrm{in}}.
\]
The group law and the change of variables \(u=\eta+\xi\) give
\[
\begin{aligned}
  \Lambda_j(\mathcal U_\xi^{\mathrm{in}}R)
  &=\mathcal U_\xi^{\mathrm{out}}
    \frac1{|\Xi_{L_j}|}\int_{\Xi_{L_j}+\xi}A_{j,u}(R)\,\dd u,\\
  \mathcal U_\xi^{\mathrm{out}}\Lambda_j(R)
  &=\mathcal U_\xi^{\mathrm{out}}
    \frac1{|\Xi_{L_j}|}\int_{\Xi_{L_j}}A_{j,u}(R)\,\dd u.
\end{aligned}
\]
The integrals over the intersection of the two
averaging sets cancel.  On their symmetric difference, each integrand
has trace norm at most \(\|R\|_1\): the output translations are
isometries, and each \(A_{j,u}\) is a channel and hence a trace-norm
contraction.
\begin{equation}\label{eq:unified-approx-covariance}
  \|\Lambda_j(\mathcal U_\xi^{\mathrm{in}}R)
    -\mathcal U_\xi^{\mathrm{out}}\Lambda_j(R)\|_1
  \le b_j(\xi)\|R\|_1,
  \qquad
  b_j(\xi):=
  \frac{|(\Xi_{L_j}+\xi)\mathbin\triangle\Xi_{L_j}|}{|\Xi_{L_j}|}.
\end{equation}
Here \(0\le b_j\le2\), and \(b_j(\xi)=O_\xi(L_j^{-1})\to0\) for fixed
\(\xi\), since \(\Xi_1\) is a bounded polytope and \(\Xi_L=L\Xi_1\).

Weighting and integrating the classical output now
produces a quantum map. The Gaussian translation bound controls its
trace.  Define
\[
  \Gamma_j^\chi(X)
  =\int_{\mathsf E}\chi(y)[\Lambda_j(g_1X)](y)\,\dd y.
\]
These maps are completely positive and satisfy
\(\|\Gamma_j^\chi\|_{1\to1}\le\|\chi\|_\infty\).
We first show the sharper trace bound
\begin{equation}\label{eq:weighted-map-trace-bound}
  \Tr\Gamma_j^\chi(X)\le(c_k+e_j)\Tr X
  \quad(X\ge0),\qquad
  e_j:=\|\chi\|_\infty\int_{\mathsf E}g_1(h)b_j(h,0)\,\dd h
  \longrightarrow0.
\end{equation}
For \(k=0\), this is trace preservation, with \(c_0=1\) and \(e_j=0\).
For \(k=d-1\), take a nonnegative probability approximate identity
\(u_\varepsilon\) on \(\mathsf E\) and let
\(B_{j,\varepsilon}=\Lambda_j(u_\varepsilon X)\), whose total trace is
\(\Tr X\).  Expressing \(g_1*u_\varepsilon\) as an average of translates
of \(u_\varepsilon\), \eqref{eq:unified-approx-covariance} gives
\[
  \left\|\Lambda_j((g_1*u_\varepsilon)X)
  -\int_{\mathsf E}g_1(h)\mathcal U_{(h,0)}^{\mathrm{out}}
       B_{j,\varepsilon}\,\dd h\right\|_1
  \le \Tr X\int_{\mathsf E}g_1(h)b_j(h,0)\,\dd h.
\]
Pairing with \(\chi\), changing variables, and using
\eqref{eq:classical-witness-translation} therefore yields
\[
  \int_{\mathsf E}\chi(y)
    \Tr[\Lambda_j((g_1*u_\varepsilon)X)(y)]\,\dd y
  \le c_k\int_{\mathsf E}\Tr B_{j,\varepsilon}(v)\,\dd v+e_j\Tr X
  =(c_k+e_j)\Tr X.
\]
Letting \(\varepsilon\downarrow0\) proves
\eqref{eq:weighted-map-trace-bound}, since
\(g_1*u_\varepsilon\to g_1\) in \(L^1\).  Dominated convergence gives
\(e_j\to0\).

Taking \(\xi=(0,z)\) in \eqref{eq:unified-approx-covariance} and integrating
against \(\chi\) gives
\[
 \|\Gamma_j^\chi(D(z)XD(z)^*)
   -D(\sqrt\gamma z)\Gamma_j^\chi(X)D(\sqrt\gamma z)^*\|_1
 \le\|\chi\|_\infty b_j(0,z)\|X\|_1\longrightarrow0.
\]
Hence, after passing to a subsequence, \cref{lem:compact-cp-limit} and
\eqref{eq:weighted-map-trace-bound} give a completely positive,
displacement-covariant limit \(\Gamma^\chi\) with
\(\Tr\Gamma^\chi(X)\le c_k\Tr X\) for \(X\ge0\), such that
\begin{equation}\label{eq:weighted-map-limit}
  \Tr[\Gamma_j^\chi(X)K]
  \longrightarrow\Tr[\Gamma^\chi(X)K]
\end{equation}
for trace-class \(X\) and compact \(K\).  Since \(c_k\le1\), the limit is
trace nonincreasing.
Define \(M_j\) by
\[
  M_j
  =\int_{\mathsf E}\chi(y)\Tr[R_j(y)W_\gamma] \,\dd y
  =\Tr[\Gamma_j^\chi(\tau_p)W_\gamma].
\]
Because \(W_\gamma\) is compact,
\eqref{eq:weighted-map-limit} passes this moment to the limit:
\begin{equation}\label{eq:moment-limit}
  M_j\longrightarrow\Tr[\Gamma^\chi(\tau_p)W_\gamma].
\end{equation}
Moreover, \cref{lem:weighted-fidelity},
\eqref{eq:classical-witness-identity}, and
\eqref{eq:quantum-witness-identities} give
\begin{equation}\label{eq:unified-fidelity-to-moment}
  F(R_j,g_1\tau_p)^2
  \le M_j\,c_k\Forb(\gamma,p).
\end{equation}

By \cref{prop:least-noise-moment},
\[
  \lim_j M_j
  =\Tr[\Gamma^\chi(\tau_p)W_\gamma]
  \le\Tr\Gamma^\chi(\tau_p)\,\Forb(\gamma,p)
  \le c_k\Forb(\gamma,p).
\]
Combining \eqref{eq:unified-maximizing-sequence} with
\eqref{eq:unified-fidelity-to-moment} yields
\(F_*\le c_k\Forb(\gamma,p)\), proving both cases.
\end{proof}

\section{Comparison with purify--clone--trace}
\label[appendix]{app:pct-comparison}

We compute the exact asymptotic PCT fidelity at simple full-rank spectra
and derive a strict upper bound for flat projector states with
\(1<r<d\), proving the comparisons stated in
\cref{sec:pct-comparison}.

\subsection{Common preliminaries}
\label{app:pct-preliminaries}

For \(D\ge1\), write
\(\mathsf d_D(N)=\dim\operatorname{Sym}^N(\C^D)=\binom{N+D-1}{D-1}\),
and let \(\Pi_{\mathrm{sym},N}^{(D)}\) project onto this symmetric
subspace.  On inputs supported on \(\operatorname{Sym}^n(\C^D)\),
Werner's pure-state cloner \cite{Werner1998} is
\begin{equation}\label{eq:grassmann-werner-cloner}
  \mathcal W_{n,m}^{(D)}(X)
  =\frac{\mathsf d_D(n)}{\mathsf d_D(m)}
   \Pi_{\mathrm{sym},m}^{(D)}
   \bigl(X\ot I_D^{\ot(m-n)}\bigr)
   \Pi_{\mathrm{sym},m}^{(D)},\qquad m\ge n.
\end{equation}
If \(\rho\) has rank at most \(r\), choose a purification
\(\psi\in\C^d\ot E\), \(E=\C^r\).  The random-purification channel
acts on product inputs as \cite{TangWrightZhandry2025}
(see also \cite[Eq.~(50) and Theorem~22]{GirardiMeleLami2025})
\begin{equation}\label{eq:random-purification-channel}
  \mathcal P_n^{(r)}(\rho^{\ot n})
  =\int_{\U(r)}
    \bigl((I_d\ot U)\proj\psi(I_d\ot U)^*\bigr)^{\ot n}\,\dd U.
\end{equation}
Here \(\dd U\) is normalized Haar measure.  All purifications in this
space differ by a unitary on \(E\), so the average is independent of
the chosen purification.  Covariance of Werner's cloner and invariance
of partial trace remove this average:
\begin{equation}\label{eq:Haar-drops}
  \cC_{n,m}^{\mathrm{PCT},r}(\rho^{\ot n})
  =\Tr_{E^m}\mathcal W_{n,m}^{(dr)}(\proj\psi^{\ot n}).
\end{equation}

\paragraph{Occupation and coherent-state limits.}
Fix \(D\ge2\), a unit vector \(\psi\in\C^D\), and an orthonormal basis
\(e_1,\ldots,e_{D-1}\) of \(\psi^\perp\).  For
\(\mathbf k=(k_1,\ldots,k_{D-1})\), let \(\ket{\mathbf k}_m\) be the
normalized symmetric vector with \(k_j\) particles in \(e_j\) and the
remaining \(m-|\mathbf k|\) in \(\psi\).

\begin{lemma}[Werner occupation law]
\label{lem:Werner-Fock}
The output \(\mathcal W_{n,m}^{(D)}(\proj\psi^{\ot n})\) is diagonal
in this occupation basis, with eigenvalue
\begin{equation}\label{eq:occupation-law}
  \pi_{n,m}^{(D)}(\mathbf k)
  =\frac{\binom{m-|\mathbf k|}{n}}
         {\binom{m+D-1}{n+D-1}},
  \qquad |\mathbf k|\le m-n.
\end{equation}
Recall that \(\tau_q=(1-q)\sum_{j\ge0}q^j\proj{j}\) is the
one-mode thermal state.
If \(m_n/n\to\gamma>1\), the output's pullback \(\omega_{n,m_n}\) to the
symmetric Fock space \(\Gamma_s(\psi^\perp)\), with occupation basis
\(\ket{\mathbf k}\), satisfies
\begin{equation}\label{eq:Werner-thermal-limit}
  \left\|\omega_{n,m_n}
    -\tau_{(\gamma-1)/\gamma}^{\ot(D-1)}\right\|_1
  \longrightarrow0.
\end{equation}
\end{lemma}

\begin{proof}
On an occupation vector with \(t=|\mathbf k|\), the expectation of
\(\Pi_{\mathrm{sym},m}^{(D)}
 (\proj\psi^{\ot n}\ot I)
 \Pi_{\mathrm{sym},m}^{(D)}\)
is \(\binom{m-t}{n}/\binom mn\).  Multiplying by
\(\mathsf d_D(n)/\mathsf d_D(m)\) gives
\eqref{eq:occupation-law}.  The rotation \(\psi\mapsto e^{i\theta}\psi\)
fixes the input state and acts on the \(t\)-excitation subspace by
\(e^{i(m-t)\theta}\), so the output has no cross-\(t\) blocks.
The stabilizer \(\U(D-1)\) acts irreducibly on each
\(\operatorname{Sym}^t(\psi^\perp)\), making each block scalar.
For each fixed \(\mathbf k\),
\[
  \pi_{n,m_n}^{(D)}(\mathbf k)
  \longrightarrow
  \gamma^{-(D-1)}
  \left(\frac{\gamma-1}{\gamma}\right)^{|\mathbf k|}.
\]
The limit is a normalized product of geometric laws.  The discrete
Scheff\'e lemma therefore gives \(\ell^1\) convergence, which is
\eqref{eq:Werner-thermal-limit}.
\end{proof}

For the rest of this appendix, \(L\) denotes a total product length.
The \(M\) of \cref{sec:small-error} counts additional outputs.
Let \(\Pi_{\le L}\) project onto Fock occupations of total number at
most \(L\), and let \(J_L\ket{\mathbf k}=\ket{\mathbf k}_L\) on this
subspace.
The trace-preserving occupation embedding is
\begin{equation}\label{eq:occupation-channel}
  \mathcal J_L(X)
  =J_L\Pi_{\le L}X\Pi_{\le L}J_L^*
   +\Tr[(I-\Pi_{\le L})X]\proj\psi^{\ot L}.
\end{equation}
For \(Z\in\psi^\perp\), distinguish its coherent Fock vector
\(\ket{\operatorname{coh}(Z)}
  =e^{-\|Z\|^2/2}\bigoplus_{t\ge0}Z^{\ot t}/\sqrt{t!}\)
from the one-particle vector \(\ket Z\),
and let
\[
  \ket{\psi_{Z,L}}
  =\frac{\ket\psi+L^{-1/2}\ket Z}
         {\sqrt{1+\|Z\|^2/L}}.
\]
For fixed \(Z\),
\begin{equation}\label{eq:coherent-product-limit}
  \left\|
    \mathcal J_L\bigl(\proj{\operatorname{coh}(Z)}\bigr)
    -\proj{\psi_{Z,L}}^{\ot L}
  \right\|_1\longrightarrow0.
\end{equation}
For \(Z\ne0\), both vectors have excitations only in the direction
\(Z/\|Z\|\) (the case \(Z=0\) is immediate).  With \(t=\|Z\|^2\),
their number amplitudes after pulling
back the product vector by \(J_L^*\) are the square roots of
\[
  B_L(j)=\binom Lj\frac{(t/L)^j}{(1+t/L)^L},
  \qquad
  \mathrm{Pois}_t(j)=e^{-t}\frac{t^j}{j!}.
\]
Pointwise convergence and the discrete Scheff\'e lemma give
\[
  \sum_{j\ge0}\bigl(\sqrt{B_L(j)}-\sqrt{\mathrm{Pois}_t(j)}\bigr)^2
  \le\sum_{j\ge0}|B_L(j)-\mathrm{Pois}_t(j)|\longrightarrow0.
\]
Thus the pure states converge in trace norm. Contractivity of
\(\mathcal J_L\) proves \eqref{eq:coherent-product-limit}.

For \(s>0\), let \(\mu_s\) be the circular complex Gaussian measure
on \(\psi^\perp\) with covariance \(sI\), so each orthonormal coordinate has
\(\mathbb E|Z_j|^2=s\).  Gaussian integration in the number basis gives
the positive coherent-state representation
\begin{equation}\label{eq:pct-coherent-mixture}
  \tau_{s/(1+s)}^{\ot(D-1)}
  =\int_{\psi^\perp}\proj{\operatorname{coh}(Z)}\,\dd\mu_s(Z),
  \qquad
  \dd\mu_s(Z)=\frac{e^{-\|Z\|^2/s}}{(\pi s)^{D-1}}\,\dd Z.
\end{equation}
Specifically, the integral is diagonal with eigenvalue
\(s^{|\mathbf k|}/(1+s)^{|\mathbf k|+D-1}\).
Combining \eqref{eq:Haar-drops}, \eqref{eq:Werner-thermal-limit},
and \eqref{eq:coherent-product-limit} therefore gives, for any fixed
purification \(\psi\in\C^d\ot E\),
\begin{equation}\label{eq:pct-product-mixture}
  \left\|
    \cC_{n,m_n}^{\mathrm{PCT},r}(\rho^{\ot n})
    -\int_{\psi^\perp}\rho_{Z,m_n}^{\ot m_n}\,\dd\mu_{\gamma-1}(Z)
  \right\|_1\longrightarrow0,
  \qquad
  \rho_{Z,L}=\Tr_E\proj{\psi_{Z,L}}.
\end{equation}
Here the finite Werner output equals
\(\mathcal J_{m_n}(\omega_{n,m_n})\).  Contractivity controls its
replacement by the thermal state, and dominated convergence controls
the coherent-state integral, whose trace-distance integrand is bounded
by \(2\).  No uniform estimate over unbounded \(Z\) is needed.

\subsection{Simple full-rank spectra}
\label{app:pct-full-rank}

Fix \(p\in\Delta_d^\circ\), so that
\(\rho_p=\diag(p_1,\ldots,p_d)\) with
\(p_1>\cdots>p_d>0\), and recall \(q_{ij}=p_j/p_i\).  Define
\begin{equation}\label{eq:q-pct}
  q_{ij}^{\mathrm{PCT}}
  =\frac{\gamma-1+\gamma q_{ij}}
         {\gamma+(\gamma-1)q_{ij}}.
\end{equation}
The asymptotic PCT fidelity is
\begin{equation}\label{eq:pct-fidelity}
  F_{\mathrm{PCT}}(\gamma,p)
  =\left(\frac{\sqrt{2\gamma-1}}{\gamma}\right)^{(d-1)/2}
   \prod_{i<j}F\!\left(\tau_{q_{ij}},
                       \tau_{q_{ij}^{\mathrm{PCT}}}\right).
\end{equation}
The equivalent closed form is displayed in
\eqref{eq:pct-fidelity-closed-form}.
PCT here uses rank bound \(d\), hence purification dimension \(d^2\).

\begin{proof}[Proof of \cref{thm:pct-comparison}\textup{(a)}]
The value \(\Funiv\) is attained by \cref{thm:unknown-optimum}\textup{(a)}.
For the PCT limit, apply interior LAN to the mixture
\eqref{eq:pct-product-mixture} and compare the resulting factors.

\emph{Step 1: The Gaussian output.}
Choose
\[
  \ket{\psi_p}=\sum_i\sqrt{p_i}\ket{ii}.
\]
Identify the tangent vector \(\ket Z=\sum_{i,j}Z_{ij}\ket{ij}\)
with its coefficient matrix.  Partial trace has differential
\begin{equation}\label{eq:partial-trace-differential}
  D_p(Z)=Z\sqrt{\rho_p}+\sqrt{\rho_p}Z^*,
  \qquad
  \rho_{Z,L}=\rho_p+L^{-1/2}D_p(Z)+O_Z(L^{-1}).
\end{equation}
Its diagonal and orbital LAN coordinates are
\begin{equation}\label{eq:h-Z}
  h_i(Z)=2\sqrt{p_i}\,\Re Z_{ii},
  \qquad \sum_i h_i(Z)=0,
\end{equation}
\begin{equation}\label{eq:beta-Z}
  \beta_{ij}(Z)
  =\frac{D_p(Z)_{ji}}{\sqrt{p_i-p_j}}
  =\frac{\sqrt{p_i}Z_{ji}
            +\sqrt{p_j}\,\overline{Z_{ij}}}
           {\sqrt{p_i-p_j}},
  \qquad i<j.
\end{equation}
The derivative of the local chart from \cref{sec:gaussian-amplification} is
\((u,w)\mapsto\diag(u)+[Y_p(w),\rho_p]\), a real linear isomorphism
onto the traceless Hermitian matrices.  The inverse function theorem
and \eqref{eq:partial-trace-differential} give exact LAN coordinates:
\[
  \rho_{Z,L}
  =g_{L,z_L}^{(p)}\diag(p+L^{-1/2}h_L)(g_{L,z_L}^{(p)})^*,
  \qquad
  (h_L,z_L)=(h(Z),\beta(Z))+O_Z(L^{-1/2}).
\]
Let \(T_L^p,S_L^p\) be the forward and reverse interior LAN channels
centered at \(\rho_p\), in the normalization of
\eqref{eq:LAN-coordinate-map}, and set
\begin{equation}\label{eq:pct-local-gaussian}
  \Theta(Z)
  =\mathsf N_{h(Z),\Sigma_p}\ot
   \bigotimes_{i<j}
   D(\beta_{ij}(Z))\tau_{q_{ij}}D(\beta_{ij}(Z))^*.
\end{equation}
The bounded-parameter LAN estimates
\eqref{eq:LAN-full-forward}--\eqref{eq:LAN-full-reverse} and
trace-norm continuity of Gaussian shifts imply, for each fixed \(Z\),
\[
  \left\|T_L^p(\rho_{Z,L}^{\ot L})-\Theta(Z)\right\|_1\to0,
  \qquad
  \left\|\rho_{Z,L}^{\ot L}-S_L^p(\Theta(Z))\right\|_1\to0.
\]

To average \(\Theta(Z)\) under \(Z\sim\mu_{\gamma-1}\), use the
orthogonal decomposition:
\[
  \psi_p^\perp=
  \left\{\sum_i z_i\ket{ii}:\sum_i\sqrt{p_i}z_i=0\right\}
  \oplus\bigoplus_{i<j}\operatorname{span}\{\ket{ij},\ket{ji}\}
\]
The diagonal and pair variables are independent, and
\eqref{eq:h-Z}--\eqref{eq:beta-Z} give
\begin{equation}\label{eq:pct-displacement-covariances}
  \operatorname{Cov}(h)=2(\gamma-1)\Sigma_p,\qquad
  \mathbb E|\beta_{ij}|^2
    =\frac{(\gamma-1)(p_i+p_j)}{p_i-p_j},
  \qquad
  \mathbb E\beta_{ij}^2=0.
\end{equation}
Here the diagonal Gaussian has covariance
\((\gamma-1)(I-\sqrt p\,\sqrt p^{\,\mathsf T})\). Multiplying its real
parts by \(2\sqrt{p_i}\) gives the first identity on the trace-zero
hyperplane.  Each \(\beta_{ij}\) is circular complex Gaussian.  Averaging
its displacement multiplies the Weyl characteristic function by
\(\exp(-\mathbb E|\beta_{ij}|^2|\xi|^2)\), leaving a thermal state
whose mean photon number increases by \(\mathbb E|\beta_{ij}|^2\):
\begin{equation}\label{eq:thermal-mean-out}
  \bar n_{ij}^{\mathrm{PCT}}
  =\frac{p_j+(\gamma-1)(p_i+p_j)}{p_i-p_j}
  =\frac{(\gamma-1)p_i+\gamma p_j}{p_i-p_j},
\end{equation}
whose geometric parameter is \eqref{eq:q-pct}.  Consequently,
\begin{equation}\label{eq:pct-lan-output}
  \Theta_p^{\mathrm{PCT}}
  :=\int_{\psi_p^\perp}\Theta(Z)\,\dd\mu_{\gamma-1}(Z)
  =\mathsf N_{0,(2\gamma-1)\Sigma_p}
   \ot\bigotimes_{i<j}\tau_{q_{ij}^{\mathrm{PCT}}}.
\end{equation}
Integrating both LAN approximations by dominated convergence and
using \eqref{eq:pct-product-mixture} gives
\begin{align}
  \left\|T_{m_n}^p(\cC_{n,m_n}^{\mathrm{PCT},d}(\rho_p^{\ot n}))
                   -\Theta_p^{\mathrm{PCT}}\right\|_1&\longrightarrow0,
  \label{eq:pct-forward-state-limit}\\
  \left\|\cC_{n,m_n}^{\mathrm{PCT},d}(\rho_p^{\ot n})
                   -S_{m_n}^p(\Theta_p^{\mathrm{PCT}})\right\|_1
                   &\longrightarrow0.
  \label{eq:pct-reverse-state-limit}
\end{align}
The target has the same two-sided approximations with
\begin{equation}\label{eq:target-lan-output}
  \Theta_p=\Theta(0)
  =\mathsf N_{0,\Sigma_p}\ot\bigotimes_{i<j}\tau_{q_{ij}}.
\end{equation}

\emph{Step 2: Fidelity and strict comparison.}
Monotonicity under \(T_{m_n}^p\) and \(S_{m_n}^p\), together with
trace-norm continuity, now gives the two-sided bound
\[
  F(\Theta_p^{\mathrm{PCT}},\Theta_p)-o(1)
  \le F\!\left(\cC_{n,m_n}^{\mathrm{PCT},d}(\rho_p^{\ot n}),\rho_p^{\ot m_n}\right)
  \le F(\Theta_p^{\mathrm{PCT}},\Theta_p)+o(1).
\]
The limiting fidelity factorizes, with
\begin{equation}\label{eq:pct-classical-affinity}
  F\!\left(\mathsf N_{0,(2\gamma-1)\Sigma_p},
            \mathsf N_{0,\Sigma_p}\right)
  =\left(\frac{\sqrt{2\gamma-1}}{\gamma}\right)^{(d-1)/2}
\end{equation}
and
\(F(\tau_q,\tau_x)=\sqrt{(1-q)(1-x)}/(1-\sqrt{qx})\).
This proves the exact value \eqref{eq:pct-fidelity}.

A thermal state \(\tau_x\) has mean photon number \(x/(1-x)\).
For one mode with \(q\in(0,1)\), the optimal amplifier has mean
\(\bar n^{\mathrm{opt}}=\gamma q/(1-q)+\gamma-1\), hence
\[
  \bar n^{\mathrm{opt}}>\frac q{1-q},
  \qquad
  \bar n^{\mathrm{PCT}}-\bar n^{\mathrm{opt}}
  =\frac{(\gamma-1)q}{1-q}>0.
\]
Thus \(q<q^{\mathrm{opt}}<q^{\mathrm{PCT}}\), where
\(q^{\mathrm{opt}}=(\gamma-1+q)/\gamma\) and
\(f_\gamma(q)=F(\tau_q,\tau_{q^{\mathrm{opt}}})\).  Since
\[
  \frac{\dd}{\dd x}\log F(\tau_q,\tau_x)
  =\frac{\sqrt q-\sqrt x}
         {2\sqrt x(1-x)(1-\sqrt{qx})}<0,\qquad x>q,
\]
every PCT quantum factor is strictly smaller.  PCT's classical factor
is also smaller: \(c\mapsto2\sqrt c/(1+c)\) decreases for \(c>1\),
and \(2\gamma-1>\gamma>1\).  Multiplying the factors proves
\(F_{\mathrm{PCT}}(\gamma,p)<\Funiv(\gamma,p)\).
\end{proof}

\subsection{Projector states}
\label{app:pct-projector}

For projector states, we upper-bound PCT's actual fidelity to compare it
with the optimum in \cref{thm:grassmann}. Its published data-processing
lower bound does not suffice for this comparison.

Fix \(1\le r<d\) and \(P\in\Gr(r,d)\), choose an isometry
\(V_P:\C^r\to\C^d\) with \(V_PV_P^*=P\), and let
\[
  \ket{\psi_P}
  =\frac1{\sqrt r}\sum_{i=1}^r V_P\ket i\ot\ket i,
  \qquad
  \ket{\psi_r}=\frac1{\sqrt r}\sum_{i=1}^r\ket{ii}.
\]
\begin{proposition}[Finite-sample PCT factorization for projector states]
\label{prop:grassmann-pct-comparison}
Define
\begin{equation}
  s_{n,m}^{(r,d)}
  =\frac{\mathsf d_{rd}(n)/\mathsf d_{rd}(m)}
         {\mathsf d_{r^2}(n)/\mathsf d_{r^2}(m)}
  =\prod_{j=r^2}^{rd-1}\frac{n+j}{m+j}.
  \label{eq:grassmann-pct-support-factor}
\end{equation}
There is an exact finite-sample factorization
\begin{equation}
  F\!\left(
    \cC_{n,m}^{\mathrm{PCT},r}((P/r)^{\ot n}),(P/r)^{\ot m}
  \right)
  =\sqrt{s_{n,m}^{(r,d)}}\,
   F\!\left(\zeta_{n,m}^{(r)},(I_r/r)^{\ot m}\right),
  \label{eq:grassmann-pct-factorization}
\end{equation}
where
\begin{equation}
  \zeta_{n,m}^{(r)}
  =\Tr_{E^m}\mathcal W_{n,m}^{(r^2)}
    \bigl(\proj{\psi_r}^{\ot n}\bigr).
  \label{eq:grassmann-pct-internal-state}
\end{equation}
\end{proposition}

\begin{proof}
By \eqref{eq:Haar-drops}, use the fixed purification \(\psi_P\).
Werner's channel is the one-row Cartan channel, so
\eqref{eq:channel-compression} applies with \(\mu=(n)\), \(\nu=(m)\),
and \(V_P\ot I_r:\C^{r^2}\to\C^{rd}\).  On the symmetric output, its
compression projection is \((P\ot I_r)^{\ot m}\) and its dimension factor is
\(s_{n,m}^{(r,d)}\) by \eqref{eq:grassmann-pct-support-factor}.
Tracing out \(E^m\) gives
\begin{equation}
  P^{\ot m}\cC_{n,m}^{\mathrm{PCT},r}((P/r)^{\ot n})P^{\ot m}
  =s_{n,m}^{(r,d)}
   V_P^{\ot m}\zeta_{n,m}^{(r)}(V_P^*)^{\ot m}.
  \label{eq:grassmann-pct-system-compression}
\end{equation}
Compressing to the target support \(P^{\ot m}\) leaves fidelity
unchanged. Homogeneity then yields \eqref{eq:grassmann-pct-factorization}.
\end{proof}

\begin{proof}[Proof of \cref{thm:pct-comparison}\textup{(b)}]
By \cref{prop:grassmann-pct-comparison}, it remains to bound the internal
fidelity.  Assume \(r>1\), let
\[
  \mathsf H_r=\left\{h\in\R^r:\sum_{i=1}^r h_i=0\right\},
  \qquad
  \Sigma_r
  =\left.\left(\frac1r I_r-\frac1{r^2}
    \mathbf1_r\mathbf1_r^{\mathsf T}\right)\right|_{\mathsf H_r},
\]
and measure every system output in the computational basis.  If
\(C_m=(C_{m,1},\ldots,C_{m,r})\) is the resulting count vector, its
centered and rescaled version belongs to \(\mathsf H_r\).
Under the target \((I_r/r)^{\ot m}\), the multinomial central limit theorem
gives
\begin{equation}
  \frac{C_m-(m/r)\mathbf1_r}{\sqrt m}
  \longrightarrow\mathsf N_{0,\Sigma_r}
  \quad\text{in distribution}.
  \label{eq:grassmann-target-count-limit}
\end{equation}
Under \(\zeta_{n,m_n}^{(r)}\), we will show
\begin{equation}
  \frac{C_{m_n}-(m_n/r)\mathbf1_r}{\sqrt{m_n}}
  \longrightarrow
  \mathsf N_{0,(2\gamma-1)\Sigma_r}
  \quad\text{in distribution}.
  \label{eq:grassmann-pct-count-limit}
\end{equation}

Apply \eqref{eq:pct-product-mixture} to the internal state, with
condensate \(\psi_r\), dimension \(r^2\), and
\(Z\sim\mu_{\gamma-1}\) on \(\psi_r^\perp\).  Identify \(Z\) with a
traceless complex \(r\times r\) matrix.  Conditional on \(Z\), measuring
the system part of the product of the normalized vector
\(\psi_r+m^{-1/2}Z\) gives multinomial probabilities
\begin{align}
  p_i^{(m)}(Z)
  &=\frac{r^{-1}+2\Re Z_{ii}/\sqrt{rm}
                +m^{-1}\sum_j|Z_{ij}|^2}
         {1+m^{-1}\|Z\|^2}\notag\\*
  &=\frac1r+\frac{h_i(Z)}{\sqrt m}+O_Z(m^{-1}),
  \qquad h_i(Z)=\frac{2\Re Z_{ii}}{\sqrt r}.
  \label{eq:grassmann-coherent-count-probabilities}
\end{align}
The multinomial central limit theorem therefore gives the conditional
limit \(\mathsf N_{h(Z),\Sigma_r}\).  Since \(\Tr Z=0\), the mean
\(h(Z)\) belongs to \(\mathsf H_r\). Under \(\mu_{\gamma-1}\) it is
centered Gaussian, with
\[
  \operatorname{Cov}(h(Z))_{ij}
  =\frac{2(\gamma-1)}r\left(\delta_{ij}-\frac1r\right)
  =2(\gamma-1)(\Sigma_r)_{ij}.
\]
Dominated convergence of the conditional characteristic functions gives
the Gaussian limit with covariance
\(\Sigma_r+2(\gamma-1)\Sigma_r=(2\gamma-1)\Sigma_r\).
The trace-norm approximation \eqref{eq:pct-product-mixture} and
contractivity of measurement transfer this limit to
\(\zeta_{n,m_n}^{(r)}\), proving
\eqref{eq:grassmann-pct-count-limit}.

Let \(L_n,T_n\) be the laws of the centered, rescaled count vectors,
converging weakly to their Gaussian
limits \(L,T\), respectively.  For any finite partition \(\mathcal P\) whose
boundaries have zero Gaussian measure, coarse-graining and weak
convergence give
\[
  \limsup_n F(L_n,T_n)
  \le\sum_{A\in\mathcal P}\sqrt{L(A)T(A)}.
\]
Along a refining sequence generating the Borel sets, the right-hand side
decreases to \(F(L,T)\).  This upper semicontinuity, fidelity monotonicity
under the count measurement, and
\eqref{eq:grassmann-target-count-limit}--\eqref{eq:grassmann-pct-count-limit}
give
\begin{equation}
  \limsup_{n\to\infty}
  F\!\left(\zeta_{n,m_n}^{(r)},(I_r/r)^{\ot m_n}\right)
  \le
  \left(\frac{\sqrt{2\gamma-1}}{\gamma}\right)^{(r-1)/2}.
  \label{eq:grassmann-pct-internal-bound}
\end{equation}
Finally, \eqref{eq:grassmann-pct-support-factor} gives
\[
  s_{n,m_n}^{(r,d)}\longrightarrow\gamma^{-r(d-r)}.
\]
Combining this with \eqref{eq:grassmann-pct-factorization} and
\eqref{eq:grassmann-pct-internal-bound} proves
\eqref{eq:grassmann-pct-upper-bound}.  Its right-hand side is strictly
smaller than \(\gamma^{-r(d-r)/2}\) when \(1<r<d\), because
\(2\gamma-1<\gamma^2\).
\end{proof}

\begin{proposition}[PCT's leading projector infidelity coefficient]
\label{prop:pct-projector-small-error}
Fix \(1<r<d\) and \(P\in\Gr(r,d)\), let \(a=r(d-r)\), and, for each
fixed \(0<\delta<1/(2r)\), take \(m_n/n\to1+\delta\).  Define
\[
  \mathcal E_n(\delta)
  =1-F\!\left(\cC_{n,m_n}^{\mathrm{PCT},r}((P/r)^{\ot n}),
                  (P/r)^{\ot m_n}\right)^2,
  \qquad
  C_r(\delta)=(1-4\delta^2)^{-(r^2-1)/2}-1.
\]
Then
\begin{equation}\label{eq:pct-projector-small-error-envelope}
  1-(1+\delta)^{-a}
  \le\liminf_n\mathcal E_n(\delta)
   \le\limsup_n\mathcal E_n(\delta)
  \le1-(1+\delta)^{-a}+(1+\delta)^{-a}C_r(\delta).
\end{equation}
The restriction \(\delta<1/(2r)\) is used only for domination in the proof.
Thus both envelopes equal \(a\delta+O_{r,d}(\delta^2)\) as
\(\delta\downarrow0\).  In this iterated limit, PCT has the same leading
coefficient of \(1-F^2\) as the optimal projector cloner.
\end{proposition}

\begin{proof}
It suffices to bound the internal squared infidelity in
\eqref{eq:grassmann-pct-factorization}.  Specialize the mixture
\eqref{eq:pct-product-mixture} to the condensate \(\psi_r\), and write
\[
  \widetilde\zeta_L
  =\int_{\psi_r^\perp}\rho_{Z,L}^{\ot L}\,\dd\mu_\delta(Z),
  \qquad
  \rho_{Z,L}=\frac{I_r}{r}+\frac{A_Z}{\sqrt L}+O_Z(L^{-1}),
  \qquad A_Z=\frac{Z+Z^*}{\sqrt r}.
\]
Here \(Z\) is a traceless complex \(r\times r\) matrix, and
\(\|\zeta_{n,m_n}^{(r)}-\widetilde\zeta_{m_n}\|_1\to0\).
For independent \(Z,W\sim\mu_\delta\), the scaled purity satisfies
\begin{align*}
  r^L\Tr\widetilde\zeta_L^2
  &=\mathbb E\!\left[
      \bigl(r\Tr(\rho_{Z,L}\rho_{W,L})\bigr)^L
    \right],\\*
  \bigl(r\Tr(\rho_{Z,L}\rho_{W,L})\bigr)^L
  &\longrightarrow e^{r\Tr(A_ZA_W)}.
\end{align*}
To justify exchanging limit and expectation, contraction under partial
trace and the pure-state trace-distance formula give
\[
  \left\|\rho_{Z,L}-\frac{I_r}{r}\right\|_2
  \le
  \left\|\proj{\psi_{Z,L}}-\proj{\psi_r}\right\|_1
  =2\sqrt{\frac{\|Z\|^2}{L+\|Z\|^2}}.
\]
Since both perturbations have trace zero, Cauchy--Schwarz and
\(1+x\le e^x\) yield
\[
  0\le\bigl(r\Tr(\rho_{Z,L}\rho_{W,L})\bigr)^L
  \le e^{4r\|Z\|\|W\|}
  \le e^{2r(\|Z\|^2+\|W\|^2)}.
\]
The last expression is integrable when \(\delta<1/(2r)\).
In a Hilbert--Schmidt orthonormal basis of traceless Hermitian matrices,
the real parts \(x_j,y_j\) of the coordinates of \(Z,W\) are independent
\(\mathsf N_{0,\delta/2}\) variables and
\(r\Tr(A_ZA_W)=4\sum_{j=1}^{r^2-1}x_jy_j\).
Dominated convergence and the Gaussian cross-moment therefore give
\begin{equation}\label{eq:pct-internal-purity-limit}
  \lim_{L\to\infty}r^L\Tr\widetilde\zeta_L^2
  =\prod_{j=1}^{r^2-1}\mathbb E e^{4x_jy_j}
  =(1-4\delta^2)^{-(r^2-1)/2}
  =1+C_r(\delta).
\end{equation}

For any state \(\sigma\) on \((\C^r)^{\ot L}\),
\[
  1-F\!\left(\sigma,(I_r/r)^{\ot L}\right)^2
  \le r^L\Tr\sigma^2-1.
\]
Indeed, apply \((\sqrt x-1)^2\le(x-1)^2\) to the relative eigenvalues
\(x=r^L\lambda_j(\sigma)\), average with weights \(r^{-L}\), and use
\(1-F^2\le2(1-F)\).
Apply this to \(\widetilde\zeta_L\) and use fidelity continuity to
transfer the bound to the PCT output:
\begin{equation}\label{eq:pct-projector-internal-small-error}
  \limsup_n
  \left[1-F\!\left(\zeta_{n,m_n}^{(r)},
                         (I_r/r)^{\ot m_n}\right)^2\right]
  \le C_r(\delta).
\end{equation}
Finally, the exact support factorization and
\(s_{n,m_n}^{(r,d)}\to(1+\delta)^{-a}\) imply
\eqref{eq:pct-projector-small-error-envelope}.
Since \(C_r(\delta)=2(r^2-1)\delta^2+O_r(\delta^4)\), the asserted
first-order coefficient follows.  For any fixed \(0<\eta<1\) and
sufficiently small \(\epsilon>0\), both envelopes lie below \(\epsilon\)
at \(\delta=(1-\eta)\epsilon/a\) and above it at
\(\delta=(1+\eta)\epsilon/a\).  Letting \(\eta\downarrow0\) gives the
leading sample coefficient \(a\) in \cref{sec:small-error}, without
assuming convergence of the fidelity at fixed \(\delta\).
\end{proof}


\begin{thebibliography}{BCFJS96}

\bibitem[BA06]{BaeAcin2006}
J. Bae and A. Ac\'{\i}n,
Asymptotic quantum cloning is state estimation,
\doi{10.1103/PhysRevLett.97.030402}{\emph{Phys. Rev. Lett.} \textbf{97}, 030402 (2006)},
\arxiv{quant-ph/0603078}.

\bibitem[BR97]{BarnettRadmore1997}
S. M. Barnett and P. M. Radmore,
\emph{Methods in Theoretical Quantum Optics},
Oxford University Press, Oxford (1997);
\doi{10.1093/acprof:oso/9780198563617.001.0001}{paperback ed. (2002)}.

\bibitem[BCFJS96]{BarnumCavesFuchsJozsaSchumacher1996}
H. Barnum, C. M. Caves, C. A. Fuchs, R. Jozsa, and B. Schumacher,
Noncommuting mixed states cannot be broadcast,
\doi{10.1103/PhysRevLett.76.2818}{\emph{Phys. Rev. Lett.} \textbf{76}, 2818--2821 (1996)},
\arxiv{quant-ph/9511010}.

\bibitem[BGA11]{BowlesGutaAdesso2011Purification}
P. Bowles, M. Gu{\c t}{\u a}, and G. Adesso,
Asymptotically optimal purification and dilution of mixed qubit and Gaussian states,
\doi{10.1103/PhysRevA.84.022320}{\emph{Phys. Rev. A} \textbf{84}, 022320 (2011)},
\arxiv{1102.2341}.

\bibitem[BGA12]{BowlesGutaAdesso2012Reversal}
P. Bowles, M. Gu{\c t}{\u a}, and G. Adesso,
Asymptotically optimal quantum channel reversal for qudit ensembles and multimode Gaussian states,
\doi{10.1088/1367-2630/14/11/113041}{\emph{New J. Phys.} \textbf{14}, 113041 (2012)},
\arxiv{1208.3569}.

\bibitem[BDE+98]{BrussDiVincenzoEkertFuchsMacchiavelloSmolin1998}
D. Bru{\ss}, D. P. DiVincenzo, A. Ekert, C. A. Fuchs, C. Macchiavello, and J. A. Smolin,
Optimal universal and state-dependent quantum cloning,
\doi{10.1103/PhysRevA.57.2368}{\emph{Phys. Rev. A} \textbf{57}, 2368--2378 (1998)},
\arxiv{quant-ph/9705038}.

\bibitem[BEM98]{BrussEkertMacchiavello1998}
D. Bru{\ss}, A. Ekert, and C. Macchiavello,
Optimal universal quantum cloning and state estimation,
\doi{10.1103/PhysRevLett.81.2598}{\emph{Phys. Rev. Lett.} \textbf{81}, 2598--2601 (1998)}.

\bibitem[BDMP05]{BuscemiDArianoMacchiavelloPerinotti2005Optimal}
F. Buscemi, G. M. D'Ariano, C. Macchiavello, and P. Perinotti,
Optimal superbroadcasting of mixed qubit states,
\href{https://wordpress.qubit.it/wp-content/uploads/publications-dariano/sendai_proc05-printed.pdf}{in
\emph{Proceedings of the 13th Quantum Information Technology Symposium
(QIT13)}, Tohoku University, Sendai, 149--155 (2005)},
\arxiv{quant-ph/0510155}.

\bibitem[BDMP06]{BuscemiDArianoMacchiavelloPerinotti2006}
F. Buscemi, G. M. D'Ariano, C. Macchiavello, and P. Perinotti,
Universal and phase-covariant superbroadcasting for mixed qubit states,
\doi{10.1103/PhysRevA.74.042309}{\emph{Phys. Rev. A} \textbf{74}, 042309 (2006)},
\arxiv{quant-ph/0602125}.

\bibitem[BH96]{BuzekHillery1996}
V. Bu\v{z}ek and M. Hillery,
Quantum copying: Beyond the no-cloning theorem,
\doi{10.1103/PhysRevA.54.1844}{\emph{Phys. Rev. A} \textbf{54}, 1844--1852 (1996)}.

\bibitem[Cav82]{Caves1982LinearAmplifiers}
C. M. Caves,
Quantum limits on noise in linear amplifiers,
\doi{10.1103/PhysRevD.26.1817}{\emph{Phys. Rev. D} \textbf{26}, 1817--1839 (1982)}.

\bibitem[CKNWW05]{CerfKrugerNavezWernerWolf2005}
N. J. Cerf, O. Kr\"uger, P. Navez, R. F. Werner, and M. M. Wolf,
Non-Gaussian cloning of quantum coherent states is optimal,
\doi{10.1103/PhysRevLett.95.070501}{\emph{Phys. Rev. Lett.} \textbf{95}, 070501 (2005)},
\arxiv{quant-ph/0410058}.

\bibitem[CB99]{CheflesBarnett1999}
A. Chefles and S. M. Barnett,
Strategies and networks for state-dependent quantum cloning,
\doi{10.1103/PhysRevA.60.136}{\emph{Phys. Rev. A} \textbf{60}, 136--144 (1999)},
\arxiv{quant-ph/9812035}.

\bibitem[Chi11]{Chiribella2011EstimationCloning}
G. Chiribella,
On quantum estimation, quantum cloning and finite quantum de Finetti theorems,
\doi{10.1007/978-3-642-18073-6_2}{in \emph{Theory of Quantum Computation,
Communication, and Cryptography}, Lecture Notes in Computer Science
\textbf{6519}, 9--25 (Springer, 2011)},
\arxiv{1010.1875}.

\bibitem[CD06]{ChiribellaDAriano2006}
G. Chiribella and G. M. D'Ariano,
Quantum information becomes classical when distributed to many users,
\doi{10.1103/PhysRevLett.97.250503}{\emph{Phys. Rev. Lett.} \textbf{97}, 250503 (2006)}.

\bibitem[CY14]{ChiribellaYang2014Optimal}
G. Chiribella and Y. Yang,
Optimal asymptotic cloning machines,
\doi{10.1088/1367-2630/16/6/063005}{\emph{New J. Phys.} \textbf{16}, 063005 (2014)},
\arxiv{1404.0990}.

\bibitem[CEM99]{CiracEkertMacchiavello1999}
J. I. Cirac, A. K. Ekert, and C. Macchiavello,
Optimal purification of single qubits,
\doi{10.1103/PhysRevLett.82.4344}{\emph{Phys. Rev. Lett.} \textbf{82}, 4344--4347 (1999)},
\arxiv{quant-ph/9812075}.

\bibitem[DF07]{DangFan2007}
G.-F. Dang and H. Fan,
Optimal broadcasting of mixed states,
\doi{10.1103/PhysRevA.76.022323}{\emph{Phys. Rev. A} \textbf{76}, 022323 (2007)},
\arxiv{quant-ph/0609182}.

\bibitem[DMP05]{DArianoMacchiavelloPerinotti2005}
G. M. D'Ariano, C. Macchiavello, and P. Perinotti,
Superbroadcasting of mixed states,
\doi{10.1103/PhysRevLett.95.060503}{\emph{Phys. Rev. Lett.} \textbf{95}, 060503 (2005)},
\arxiv{quant-ph/0506251}.

\bibitem[Die82]{Dieks1982}
D. Dieks,
Communication by EPR devices,
\doi{10.1016/0375-9601(82)90084-6}{\emph{Physics Letters A} \textbf{92}, 271--272 (1982)}.

\bibitem[Fan03]{Fan2003}
H. Fan,
Quantum cloning of mixed states in symmetric subspaces,
\doi{10.1103/PhysRevA.68.054301}{\emph{Phys. Rev. A} \textbf{68}, 054301 (2003)},
\arxiv{quant-ph/0308058}.

\bibitem[FLS07]{FanLiuShi2007}
H. Fan, B. Y. Liu, and K. J. Shi,
Quantum cloning of identical mixed qubits,
\doi{10.26421/QIC7.5-6-8}{\emph{Quantum Inf. Comput.} \textbf{7}, 551--558 (2007)},
\arxiv{quant-ph/0601017}.

\bibitem[FWJ+14]{FanWangJingYueShiZhangMu2014}
H. Fan, Y.-N. Wang, L. Jing, J.-D. Yue, H.-D. Shi, Y.-L. Zhang, and L.-Z. Mu,
Quantum cloning machines and the applications,
\doi{10.1016/j.physrep.2014.06.004}{\emph{Phys. Rep.} \textbf{544}, 241--322 (2014)},
\arxiv{1301.2956}.

\bibitem[FGSS26]{FanizzaGrinkoScharnhorstSpilecki2026}
M. Fanizza, D. Grinko, T. Scharnhorst, and J. Spilecki,
Optimal cloning of mixed states,
\arxiv{2608.27298} (2026).

\bibitem[FH91]{FultonHarris1991}
W. Fulton and J. Harris,
\doi{10.1007/978-1-4612-0979-9}{\emph{Representation Theory: A First Course}},
Graduate Texts in Mathematics 129, Springer (1991).

\bibitem[GML25]{GirardiMeleLami2025}
F. Girardi, F. A. Mele, and L. Lami,
Random purification channel made simple,
\arxiv{2511.23451} (2025).

\bibitem[GM97]{GisinMassar1997}
N. Gisin and S. Massar,
Optimal quantum cloning machines,
\doi{10.1103/PhysRevLett.79.2153}{\emph{Phys. Rev. Lett.} \textbf{79}, 2153--2156 (1997)}.

\bibitem[GBA10]{GutaBowlesAdesso2010Teleportation}
M. Gu{\c t}{\u a}, P. Bowles, and G. Adesso,
Quantum teleportation benchmarks for independent and identically-distributed
spin states and displaced thermal states,
\doi{10.1103/PhysRevA.82.042310}{\emph{Phys. Rev. A} \textbf{82}, 042310 (2010)},
\arxiv{1004.1843}.

\bibitem[GJ07]{GutaJencova2007LAN}
M. Gu{\c t}{\u a} and A. Jen{\v c}ov\'a,
Local asymptotic normality in quantum statistics,
\doi{10.1007/s00220-007-0340-1}{\emph{Commun. Math. Phys.} \textbf{276}, 341--379 (2007)},
\arxiv{quant-ph/0606213}.

\bibitem[GK06]{GutaKahn2006LANQubits}
M. Gu{\c t}{\u a} and J. Kahn,
Local asymptotic normality for qubit states,
\doi{10.1103/PhysRevA.73.052108}{\emph{Phys. Rev. A} \textbf{73}, 052108 (2006)},
\arxiv{quant-ph/0512075}.

\bibitem[GM06]{GutaMatsumoto2006MixedGaussianCloning}
M. Gu{\c t}{\u a} and K. Matsumoto,
Optimal cloning of mixed Gaussian states,
\doi{10.1103/PhysRevA.74.032305}{\emph{Phys. Rev. A} \textbf{74}, 032305 (2006)},
\arxiv{quant-ph/0605161}.

\bibitem[Hol13]{Holevo2013QuantumSystems}
A. S. Holevo,
\doi{10.1515/9783110273403}{\emph{Quantum Systems, Channels, Information:
A Mathematical Introduction}},
De Gruyter Studies in Mathematical Physics 16, De Gruyter, Berlin/Boston (2013).

\bibitem[JSO26]{JeonSohnOh2026}
S. Jeon, V. Sohn, and C. Oh,
Optimal copy complexity of quantum state cloning,
\arxiv{2608.24484} (2026).

\bibitem[Joh01]{Johansson2001}
K. Johansson,
Discrete orthogonal polynomial ensembles and the Plancherel measure,
\doi{10.2307/2661375}{\emph{Ann. of Math.} \textbf{153}, 259--296 (2001)},
\arxiv{math/9906120}.

\bibitem[KG09]{KahnGuta2009LANFiniteDimensional}
J. Kahn and M. Gu{\c t}{\u a},
Local asymptotic normality for finite dimensional quantum systems,
\doi{10.1007/s00220-009-0787-3}{\emph{Commun. Math. Phys.} \textbf{289}, 597--652 (2009)},
\arxiv{0804.3876}.

\bibitem[KH08]{KalevHen2008}
A. Kalev and I. Hen,
No-broadcasting theorem and its classical counterpart,
\doi{10.1103/PhysRevLett.100.210502}{\emph{Phys. Rev. Lett.} \textbf{100},
210502 (2008)}.

\bibitem[Kaz10]{Kazakov2007}
A. Ya. Kazakov,
``Partial'' quantum cloning and quantum cloning of mixed states,
\doi{10.1142/S0219749910006186}{\emph{Int. J. Quantum Inf.}
\textbf{8}, 435--442 (2010)},
\arxiv{0711.2860}.

\bibitem[KW99]{KeylWerner1999}
M. Keyl and R. F. Werner,
Optimal cloning of pure states, testing single clones,
\doi{10.1063/1.532887}{\emph{J. Math. Phys.} \textbf{40}, 3283--3299 (1999)},
\arxiv{quant-ph/9807010}.

\bibitem[KW01]{KeylWerner2001Spectrum}
M. Keyl and R. F. Werner,
Estimating the spectrum of a density operator,
\doi{10.1103/PhysRevA.64.052311}{\emph{Phys. Rev. A} \textbf{64}, 052311 (2001)},
\arxiv{quant-ph/0102027}.

\bibitem[KW01p]{KeylWerner2001Purification}
M. Keyl and R. F. Werner,
The rate of optimal purification procedures,
\doi{10.1007/PL00001027}{\emph{Ann. Henri Poincar\'e} \textbf{2}, 1--26 (2001)},
\arxiv{quant-ph/9910124}.

\bibitem[KC22]{KimChitambar2022}
C. Kim and E. Chitambar,
Process-optimized phase-covariant quantum cloning,
\doi{10.1103/PhysRevA.106.022405}{\emph{Phys. Rev. A} \textbf{106}, 022405 (2022)},
\arxiv{2107.03042}.

\bibitem[Kup02]{Kuperberg2002}
G. Kuperberg,
Random words, quantum statistics, central limits, random matrices,
\doi{10.4310/MAA.2002.v9.n1.a3}{\emph{Methods Appl. Anal.}
\textbf{9}, 99--118 (2002)},
\arxiv{math/9909104}.

\bibitem[LN24]{LahiryNussbaum2024}
S. Lahiry and M. Nussbaum,
Minimax estimation of low-rank quantum states and their linear functionals,
\doi{10.3150/23-BEJ1610}{\emph{Bernoulli} \textbf{30}, 610--635 (2024)},
\arxiv{2111.03279}.

\bibitem[LSW18]{LamiSabapathyWinter2018}
L. Lami, K. K. Sabapathy, and A. Winter,
All phase-space linear bosonic channels are approximately Gaussian dilatable,
\doi{10.1088/1367-2630/aae738}{\emph{New J. Phys.} \textbf{20}, 113012 (2018)},
\arxiv{1806.11042}.

\bibitem[LW17]{LemmWilde2017}
M. Lemm and M. M. Wilde,
Information-theoretic limitations on approximate quantum cloning and broadcasting,
\doi{10.1103/PhysRevA.96.012304}{\emph{Phys. Rev. A} \textbf{96}, 012304 (2017)},
\arxiv{1608.07569}.

\bibitem[LTHC26]{LiTheilHarrowChuang2026}
Z. Li, E. Theil, A. W. Harrow, and I. Chuang,
An exponential sample-complexity advantage for coherent quantum inference,
\arxiv{2605.21457} (2026).

\bibitem[MP95]{MassarPopescu1995}
S. Massar and S. Popescu,
Optimal extraction of information from finite quantum ensembles,
\doi{10.1103/PhysRevLett.74.1259}{\emph{Phys. Rev. Lett.} \textbf{74}, 1259--1263 (1995)}.

\bibitem[PRV67]{ParthasarathyRangaRaoVaradarajan1967}
K. R. Parthasarathy, R. Ranga Rao, and V. S. Varadarajan,
Representations of complex semi-simple Lie groups and Lie algebras,
\doi{10.2307/1970351}{\emph{Annals of Mathematics} \textbf{85}, 383--429 (1967)}.

\bibitem[PSTW25]{PelecanosSpileckiTangWright2025}
A. Pelecanos, J. Spilecki, E. Tang, and J. Wright,
Mixed state tomography reduces to pure state tomography,
\arxiv{2511.15806} (2025).

\bibitem[Ras03a]{Rastegin2003Upper}
A. E. Rastegin,
Upper bound on the global fidelity for mixed-state cloning,
\doi{10.1103/PhysRevA.67.012305}{\emph{Phys. Rev. A} \textbf{67}, 012305 (2003)}.

\bibitem[Ras03b]{Rastegin2002Relative}
A. E. Rastegin,
A lower bound on the relative error of mixed-state cloning and related operations,
\doi{10.1088/1464-4266/5/6/017}{\emph{J. Opt. B: Quantum Semiclass. Opt.}
\textbf{5}, S647--S650 (2003)},
\arxiv{quant-ph/0208159}.

\bibitem[Ras03c]{Rastegin2003StateDependent}
A. E. Rastegin,
Global-fidelity limits of state-dependent cloning of mixed states,
\doi{10.1103/PhysRevA.68.032303}{\emph{Phys. Rev. A} \textbf{68}, 032303 (2003)},
\arxiv{quant-ph/0301132}.

\bibitem[SIGA05]{ScaraniIblisdirGisinAcin2005}
V. Scarani, S. Iblisdir, N. Gisin, and A. Ac\'{\i}n,
Quantum cloning,
\doi{10.1103/RevModPhys.77.1225}{\emph{Rev. Mod. Phys.} \textbf{77}, 1225--1256 (2005)},
\arxiv{quant-ph/0511088}.

\bibitem[Str65]{Strassen1965}
V. Strassen,
The existence of probability measures with given marginals,
\doi{10.1214/aoms/1177700153}{\emph{Ann. Math. Statist.} \textbf{36},
423--439 (1965)}.

\bibitem[TWZ25]{TangWrightZhandry2025}
E. Tang, J. Wright, and M. Zhandry,
Conjugate queries can help,
\arxiv{2510.07622} (2025).

\bibitem[Wer98]{Werner1998}
R. F. Werner,
Optimal cloning of pure states,
\doi{10.1103/PhysRevA.58.1827}{\emph{Phys. Rev. A} \textbf{58}, 1827--1832 (1998)},
\arxiv{quant-ph/9804001}.

\bibitem[WZ82]{WoottersZurek1982}
W. K. Wootters and W. H. Zurek,
A single quantum cannot be cloned,
\doi{10.1038/299802a0}{\emph{Nature} \textbf{299}, 802--803 (1982)}.


\end{thebibliography}
\end{document}